\documentclass[11pt]{article}
\pdfoutput=1

\usepackage[utf8]{inputenc}
\usepackage{multirow}
\usepackage{amsmath, amsfonts, amssymb}
\usepackage{comment}
\usepackage{graphicx}
\usepackage{psfrag}
\usepackage{amsthm}
\usepackage[usenames,dvipsnames,svgnames,table]{xcolor}
\usepackage{enumerate}
\usepackage{arydshln}
\usepackage{soul}
 \usepackage{slashed}
\usepackage{subfig}
\usepackage{mathrsfs}
 \usepackage{a4wide}
  \usepackage{tikz}
  \usepackage{tikz-cd}
  \usetikzlibrary{shapes.geometric}
  \usepackage{tcolorbox}

  \usepackage{color}
  \definecolor{dark-gray}{gray}{0.20}
  \definecolor{gray}{gray}{0.30}
  \definecolor{light-gray}{gray}{0.80}
  \definecolor{dark-red}{rgb}{0.7,0,0}
  \definecolor{dark-green}{rgb}{0.1,0.4,0}
  \definecolor{dark-blue}{rgb}{0.3,0.3,0.7}
  \definecolor{light-blue}{rgb}{0.8,0.8,1}
      \definecolor{swamp}{RGB}{240, 199, 197}

  \usepackage{pifont}

\usepackage{newunicodechar} 
\newunicodechar{ỳ}{\`y}

\usepackage{setspace}

\newcommand{\be}{\begin{equation}}
\newcommand{\ee}{\end{equation}}
\newcommand{\eq}[1]{(\ref{#1})}

\def\be{\begin{equation}}
\def\ee{\end{equation}}
\def\bea{\begin{eqnarray}}
\def\eea{\end{eqnarray}}

\newcommand{\shortexact}[3]{
\begin{tikzcd}[ampersand replacement=\&]
	0 \& #1
	\&  #2
	\& #3
	\& 0
	\arrow[from=1-1, to=1-2]
	\arrow[from=1-2, to=1-3]
	\arrow[from=1-3, to=1-4]
	\arrow[from=1-4, to=1-5]
\end{tikzcd}
}

\newcommand{\C}{\mathbb{C}}
\newcommand{\R}{\mathbb{R}}
\newcommand{\Q}{\mathbb{Q}}
\newcommand{\Z}{\mathbb{Z}}
\newcommand{\abs}[1]{\lvert #1 \rvert}
\newcommand{\ang}[1]{\langle #1 \rangle}

\newcommand{\Spin}{\mathrm{Spin}}
\newcommand{\Pin}{\mathrm{Pin}}
\newcommand{\GL}{\mathrm{GL}}
\newcommand{\SL}{\mathrm{SL}}
\newcommand{\Sq}{\mathrm{Sq}}
\newcommand{\Mp}{\mathrm{Mp}}

\newcommand{\RP}{\mathbb{RP}}
\newcommand{\CP}{\mathbb{CP}}
\newcommand{\HP}{\mathbb{HP}}

\newcommand{\Hom}{\mathrm{Hom}}

\def\simleq{\; \raise0.3ex\hbox{$<$\kern-0.75em
      \raise-1.1ex\hbox{$\sim$}}\; }
   \def\simgeq{\; \raise0.3ex\hbox{$>$\kern-0.75em
      \raise-1.1ex\hbox{$\sim$}}\; }

\numberwithin{equation}{section}

\usepackage{jheppub}
\usepackage{hyperref}
\usepackage{cleveref}

\hypersetup{
	colorlinks=true,
	linkcolor=dark-blue,
	citecolor=dark-red,
	urlcolor=dark-green,
	linktoc=page
}
\newtheorem{thm}[equation]{Theorem}
\newtheorem{prop}[equation]{Proposition}
\newtheorem{lem}[equation]{Lemma}
\newtheorem{cor}[equation]{Corollary}
\theoremstyle{remark}

\crefname{appendix}{Appendix}{Appendices}

\title{\centering The anomaly that was not meant IIB}

\author{Arun Debray$^1$,} \affiliation{$^1$Department of Mathematics, The University of Texas at Austin, Austin, TX 78712, USA}
\author{Markus Dierigl$^2$,} \affiliation{$^2$Arnold Sommerfeld Center for Theoretical Physics, LMU, Munich, 80333, Germany}
\author{Jonathan J.~ Heckman$^{3,4}$,} \affiliation{$^3$Department of Physics and Astronomy, University of Pennsylvania, Philadelphia, PA 19104, USA}\affiliation{$^4$Department of Mathematics, University of Pennsylvania, Philadelphia, PA 19104, USA}
\author{Miguel Montero$^5$} \affiliation{$^5$Department of Physics, Harvard University, Cambridge, MA 02138, USA}

\emailAdd{a.debray@math.utexas.edu}
\emailAdd{m.dierigl@lmu.de}
\emailAdd{jheckman@sas.upenn.edu}
\emailAdd{mmontero@g.harvard.edu}

\abstract{
Type IIB supergravity enjoys a discrete non-Abelian duality group, which has potential quantum anomalies. In this paper we explicitly compute these, and present the bordism group that controls them, modulo some physically motivated assumptions. Quite surprisingly, we find that they do not vanish, which naively would signal an inconsistency of F-theory. Remarkably, a subtle modification of the standard 10d Chern-Simons term cancels these anomalies, a fact which relies on the \textit{specific} field content of type IIB supergravity. We also discover other ways to
cancel this anomaly, via a topological analog of the Green-Schwarz mechanism. These alternative type IIB theories have the same
low energy supergravity limit as ordinary type IIB, but a different spectrum of extended objects.
They could either be part of the Swampland, or connect to the standard
theory via domain walls.}

\begin{document}
\hypersetup{pageanchor=false}
\makeatletter
\let\old@fpheader\@fpheader
\preprint{LMU-ASC 24/21}

\makeatother

\maketitle

\hypersetup{
    pdftitle={The Anomaly That Was Not Meant IIB},
    pdfauthor={Arun Debray, Markus Dierigl, Jonathan J.~Heckman, Miguel Montero},
    pdfsubject={String dualities, Anomalies, Cobordisms}
}

\newcommand{\remove}[1]{\textcolor{red}{\sout{#1}}}

\newpage

\section{Introduction}
\label{sec:intro}

Dualities constitute one of the deepest features of string theory. These are \textit{exact} statements relating equivalent
formulations of quantum gravity in different regimes of validity.

Of particular significance is the duality group $\mathcal{G}_{\mathrm{IIB}}$ of type IIB string theory and its geometric uplift to
F-theory \cite{Schwarz:1983qr, Vafa:1996xn, Morrison:1996na, Morrison:1996pp}. At the level of type IIB supergravity, the duality group
is really $\SL(2,\mathbb{R})$, but the presence of quantized BPS objects reduces this to $\SL(2,\mathbb{Z})$,
the group of large diffeomorphisms of an elliptic curve. Including chiral degrees of freedom (the fermions and the self-dual form),
it is actually more appropriate to state that the IIB duality group is $\GL^{+}(2,\mathbb{Z})$,
the Pin$^+$ cover of $\GL(2,\mathbb{Z})$ (see \cite{TY19} as well as \cite{Pantev:2016nze}).

F-theory provides arguably the most complete picture for formulating non-perturbative phenomena in string theory, and has led to a range of applications, from concrete moduli stabilization scenarios \cite{Denef:2005mm, Denef:2007pq}; to the construction of
string-based particle physics scenarios \cite{Donagi:2008ca, Beasley:2008dc, Cvetic:2019gnh}; and to the classification and study of
six-dimensional superconformal field theories (see \cite{Heckman:2018jxk} for a review).

Undergirding all of this is the assumption that the duality group of type IIB strings is really retained at the
quantum level. The duality group is first encountered as a global symmetry of the low-energy type IIB supergravity Lagrangian. But if it is truly a duality of quantum gravity it should actually be gauged: we should be allowed to
introduce duality defects (indeed, these are predicted by  F-theory \cite{Vafa:1996xn,Weigand:2010wm,Weigand:2018rez}, see  \cite{Dierigl:2020lai} for a connection to Swampland arguments and the cobordism conjecture) and the type IIB partition function should sum over  $\mathcal{G}_{\mathrm{IIB}}$ bundles. Furthermore, Swampland arguments also show that if $\mathcal{G}_{\mathrm{IIB}}$ is an exact symmetry of the theory, then it must be gauged (see e.g.~\cite{Vafa:2005ui,Banks:2010zn,Harlow:2018tng} and \cite{vanBeest:2021lhn} for a recent review).  However, in order to be able to gauge a symmetry, it first needs to be anomaly free. Thus, if we take F-theory seriously and insist on gauging $\mathcal{G}_{\mathrm{IIB}}$, we are naturally led to the central question of this paper:

\begin{center}
\begin{tcolorbox}[width=0.6\textwidth]
\begin{center}
Is the IIB duality group $\mathcal{G}_{\mathrm{IIB}}$ anomalous?
\end{center}
\end{tcolorbox}\end{center}

One's first instinct might be that the answer must be a resounding ``No'', on account of the massive amount of evidence
that we have accrued in favor of F-theory and the web of string dualities. Still, to our knowledge, there has been no
direct verification of this important feature (see however \cite{Gaberdiel:1998ui,Minasian:2016hoh} which discuss a related anomaly; see also \cite{Ibanez:1992hc}).
One obstacle in performing such a computation is that the evaluation of the partition function
requires a detailed understanding of the self-dual field, a topic which is notoriously difficult. We follow the presentation in \cite{Hsieh:2020jpj}, which builds on the seminal works \cite{Belov:2006jd, Belov:2006xj} and later work by Monnier (see also \cite{Witten:1996hc, Monnier:2011rk,Monnier:2011mv,Monnier:2013kna,DelZotto:2015isa, Heckman:2017uxe}).

Following a general pattern, the duality anomaly is captured
by an 11d topological quantum field theory, with 10d type IIB supergravity
formally viewed as living on the boundary.\footnote{See for example references \cite{Martucci:2014ema, Assel:2016wcr, Lawrie:2018jut, Seiberg:2018ntt, Hsieh:2019iba} (as well as \cite{Cordova:2019jnf, Cordova:2019uob}) for related analyses of duality anomalies in the context of quantum field theory.}
Anomalies are captured by evaluating the partition function of the 11d anomaly theory on a general Euclidean background with $M = \partial X$:
\begin{equation}
Z[M] = e^{2 \pi i \mathcal{A} (X)} \,,
\end{equation}
where in our case $\mathcal{A}$ depends on the $\eta$-invariants for the gravitini and dilatini, as well as a slight extension of the anomaly theory of the
chiral 4-form of type IIB supergravity developed in \cite{Monnier:2011mv,Monnier:2011rk,Monnier:2013kna,Hsieh:2020jpj}. Specifically, the anomaly theory relies on the existence of a canonical choice for the quadratic refinement of the pairing in differential cohomology (or, more generally, the relevant differential cohomology theory for IIB RR fields, which is differential K-theory if duality backgrounds are ignored), a feature which has only been shown to exist for Spin manifolds without a duality bundle turned on; in this paper, we assume that such a canonical choice exists for general $\text{Spin-}\GL^{+}(2,\mathbb{Z})$ manifolds. This can be motivated by M-/F-duality, since on the M-theory side we do not need to specify a quadratic refinement as additional data. With this assumption, the cobordism classification of invertible topological quantum field theories (tQFTs) \cite{FH16,Yonekura:2018ufj} implies that the partition function above is a bordism invariant of the group $\Omega_{11}^{\text{Spin-}\GL^{+}(2,\mathbb{Z})}$, which we computed explicitly to be
\begin{equation}\Omega_{11}^{\Spin\text{-}\GL^+(2, \Z)} \cong\Z_8 \oplus (\Z_2)^{\oplus 9} \oplus \Z_3 \oplus \Z_{27}.\end{equation}
The details of this computation, which is based on the Adams and Atiyah-Hirzebruch spectral sequences, will be presented in a forthcoming publication \cite{Upcoming}.

To compute anomalies, we just need to evaluate the partition function of the anomaly theory on representatives of the generators of the bordism group above. We explicitly do this and, rather surprisingly, we find that the answer to the question in the box above
is ``Yes'': Type IIB supergravity, in the form in which it is written in
textbooks today, has a duality anomaly, and hence cannot be the low-energy
limit of a theory where the duality symmetry is truly present, like type IIB string theory or F-theory.

Before dismissing F-theory and thus the entire duality web as inconsistent,
we must first assess whether there might be additional (albeit quite subtle) interaction terms present which
would have remained invisible to previous analyses. Happily, we find that there is a specific 10d topological term which appears to rescue the
original IIB superstring theory. The anomalies
can be cancelled by a small modification of the triple Chern-Simons term $C_4 \wedge H_3 \wedge F_3$ of type IIB supergravity (or more precisely, the quadratic refinement mentioned above), by including a particular torsion term that encodes the duality bundle:
\begin{equation}\label{eqn:topolobampo}
I_{\mathrm{new}} =  F_{5} \cup \left[ \left(\beta(a)^2 + \lambda_2 \frac{(p_1)_{3}}{2} \right) \cup a + \frac{1}{2}[(p_1)_4 - \mathcal{P}(w)] \cup b  + \kappa \, \beta(b)^2 \cup b \right],
\end{equation}
where $F_{5}$ is the self-dual five-form field strength $a$ and $b$ are respectively $\mathbb{Z}_3$ and $\mathbb{Z}_4$ torsional
one-forms, $\beta(a)$, $\beta(b)$ are corresponding 2-form ``discrete fluxes'' for $a$ and $b$, respectively, and the $(p_1)_n$ are mod $n$ reductions of the first Pontryagin class, with $\mathcal{P}(w)$ a characteristic class built from the second Stiefel-Whitney class of the background. The coefficient $\lambda_2 = \pm 1$ is a sign, and $\kappa$ is an integer modulo 4. The last term does not contribute in the cases we can explicitly evaluate but might lead to non-trivial contributions on non-Spin manifolds. The delicate character of this term may explain why it seems to have evaded previous tests of string dualities.

The details of the mechanism we propose rely on precise numerical coincidences that would not have worked if the duality properties of the spectrum of type IIB supergravity were even slightly different. For instance, one of the anomalies we find (obtained from an 11-dimensional lens space $S^{11}/\mathbb{Z}_3$) takes values in the group $\mathbb{Z}_{27}$. Out of the 26 non-vanishing (anomalous) possibilities, only two of them can be cancelled via the quadratic refinement mechanism, and \emph{exactly one} can be cancelled with the quadratic refinement of the Chern-Simons term associated to $S^{11}/\mathbb{Z}_3$ (which we determine independently). In this way, the fact that we can cancel the anomaly via this mechanism is vaguely reminiscent of the miraculous cancellation that already occurs at the level of perturbative gravitational anomalies in type IIB supergravity \cite{Alvarez-Gaume:1983ihn}.

Perhaps even more striking, for some choices of 11d background geometry, the proposed term \eqref{eqn:topolobampo}
cannot possibly help. However, by going over the full list of 11d geometries which generate $\Omega_{11}^{\text{Spin-}\GL^{+}(2,\mathbb{Z})}$,
we find the fortuitous coincidence that in nearly all cases where it cannot contribute, the 11d anomaly vanishes anyway!
There is precisely one exceptional case involving a non-Spin manifold with a $\text{Spin-}{D_8}$ structure, and the anomaly (which is at most a sign) depends on the value of the Arf invariant of a certain quadratic refinement, which we do not know how to determine. For one choice of sign, the anomaly would also cancel, and we leave a complete independent verification of this case to future work.

The proposed topological interaction term also correctly accounts for some of the qualitative features of known type IIB compactifications,
which we view as a preliminary check of our general considerations. For example, S-fold backgrounds \cite{Garcia-Etxebarria:2015wns, Aharony:2016kai} contribute a fractional D3-brane charge which is tightly correlated with a specific duality bundle. The term \eq{eqn:topolobampo} can partially capture these charge shifts. Relatedly, a stack of D3-branes has a particular duality anomaly in its worldvolume $\mathcal{N} = 4$ super-Yang-Mills theory, which should appear as a 5d topological term in the gravity dual background. With some caveats, the term of equation \eqref{eqn:topolobampo} correctly captures such a term. Altogether, this suggests a self-consistent picture which relies on the existence of the topological interaction term \eq{eqn:topolobampo}.

All in all, the story that we present here is ultimately a happy ending for F-theory and dualities. But as in most good stories, there is a twist: we also find several \emph{alternative ways to cancel the duality anomaly}, via the topological Green-Schwarz mechanism \cite{Garcia-Etxebarria:2017crf}. All these alternatives have an identical low-energy description (that of type IIB supergravity), and differ at the level of extended operators (or equivalently, massive states, due to the completeness principle \cite{Polchinski:2003bq,Banks:2010zn,vanBeest:2021lhn}). We did not do an exhaustive classification, and there are probably more options. Some of these possibilities do not correspond to the familiar type IIB string theory, since some backgrounds and branes that should be there (such as certain orbifolds and S-folds \cite{Garcia-Etxebarria:2015wns,Aharony:2016kai}) turn out to be absent, instead being confined at the endpoints of other objects. But are they inconsistent, i.e., do they belong to the Swampland? Or did we just find that there are several different quantum theories of gravity (which would be connected by dynamical domain walls \cite{McNamara:2019rup}) with the same type IIB supergravity as their low-energy limit? A similar question arises for the discrete $\theta$ angle in M-theory introduced in \cite{FH21}, as well as for type I strings \cite{Sethi:2013hra}. At present, we do not know how to distinguish between these two intriguing possibilities, but all things considered, there seems to be mounting evidence that the usual M-theory ``star picture'' might be spikier than previously thought (see Figure \ref{fig:Mflower}).

\begin{figure}
\centering
\includegraphics[width = 0.7 \textwidth]{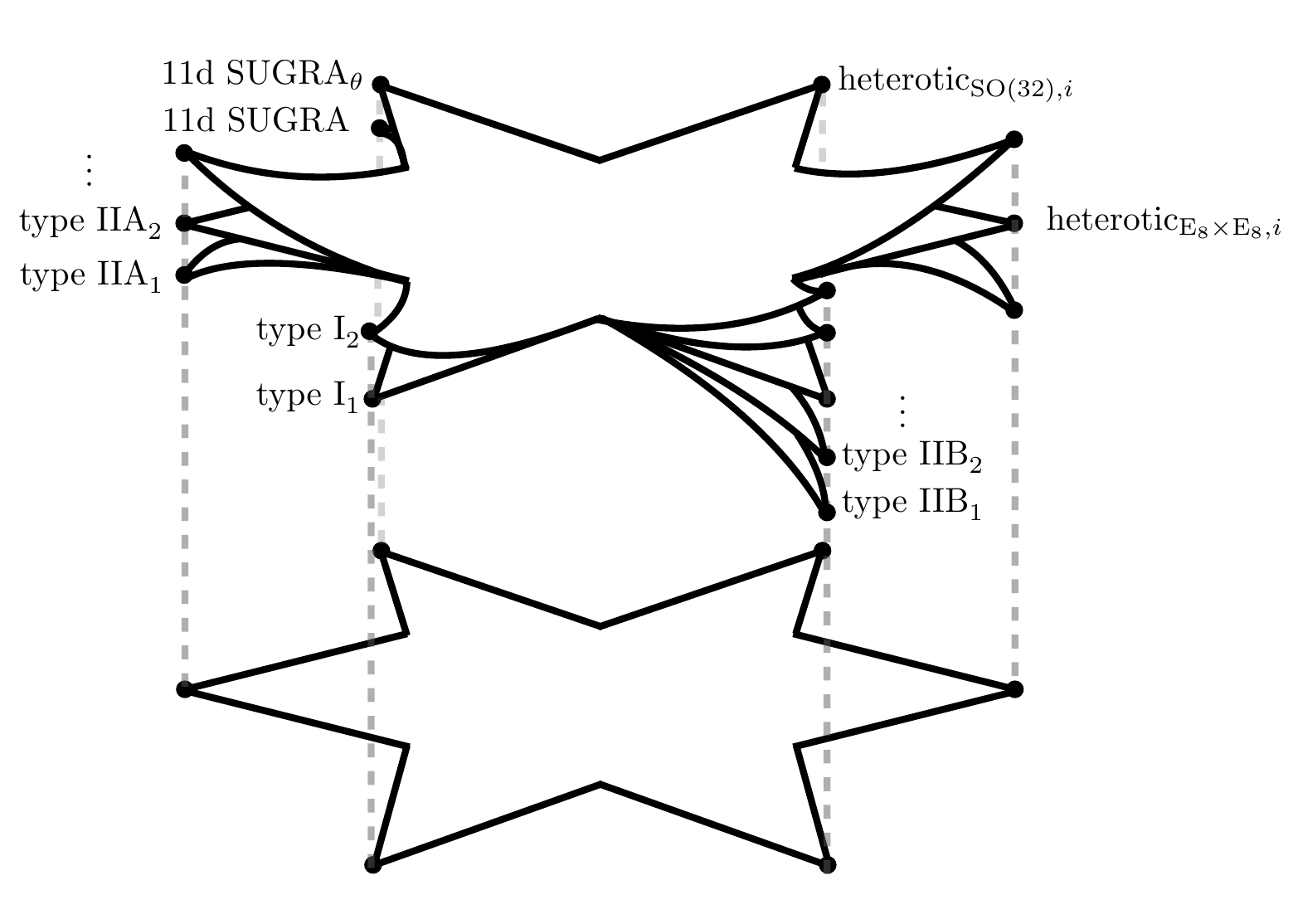}
\caption{We have found several seemingly consistent versions of type IIB theory with a non-anomalous duality group, which differ from each other at the level of massive states and extended objects. Perhaps some (or all) of these alternate versions are in the Swampland. But it does raise the question of whether the usual M-theory star picture should be replaced by a multi-sheeted one with different topological sectors. In that scenario we would expect all of these to be connected to each other via cobordism domain walls \cite{McNamara:2019rup}.}
\label{fig:Mflower}
\end{figure}

The rest of this paper is organized as follows. Section \ref{sec:bordanom} contains a brief introduction to the topic of anomalies, their classification via bordism, and the particular class of ``Dai-Freed'' anomalies that this paper focuses on, reviewing an heuristic argument from \cite{Garcia-Etxebarria:2018ajm} which explains why they should cancel in a quantum theory of gravity where topology changes are allowed (although they must also cancel in systems without gravity where topology can fluctuate, such as D-brane worldvolumes \cite{TY19}). In Section \ref{sec:dual} we review the duality group of type IIB supergravity and its extensions involving fermions and worldsheet orientation-reversal symmetries. We also introduce characteristic classes that will be useful to characterize duality bundles later on. Section \ref{sec:gen} contains the main results of the paper, where we present the bordism group classifying the anomaly theory of type IIB supergravity, the concrete anomaly theory, and we evaluate it for a specific choice of generators. We find several anomalies and explain how to cancel most of them with certain topological couplings in the theory, in a variety of ways. Section \ref{sec:x11} discusses in some detail the generator of the single bordism class where we do not know how to evaluate the anomaly even in principle. Finally, Section \ref{sec:concl} contains our conclusions. Additional information and the technical computations underlying our main results are carried out in the Appendices.


\section{Bordism groups and anomalies}
\label{sec:bordanom}

In its most basic formulation, the appearance of an anomaly indicates the violation of a symmetry due to quantum
effects (see e.g.~\cite{Bilal:2008qx}). This can be diagnosed by a lack of gauge invariance when coupling the theory to a background connection for the symmetry. The usual perturbative anomalies can be deduced from one-loop diagrams which determine the change of the partition function under infinitesimal gauge transformations, i.e., gauge transformations that are close to the identity. The groups we are interested in here, however, are discrete and their anomalies cannot be captured in terms of Feynman diagrams. Original studies on anomalies of discrete symmetries \cite{Ibanez:1991hv}, or some studies on anomalies of transformations not continuously connected to the identity \cite{Elitzur:1984kr} were often ``artisanal'' and relied on techniques such as embedding the discrete symmetry in an auxiliary continuous one. Nowadays, we have a standard procedure that allows us to study anomalies of either discrete or continuous symmetries in a uniform fashion.

This perspective advocates that the presence of an anomaly in a $d$-dimensional quantum field theory $Z$ means that
$Z$ is a boundary theory to a $(d+1)$-dimensional invertible field theory (IFT) $e^{2\pi i\mathcal{A}}$, called
the anomaly field theory of $Z$ \cite{FT12}. Invertible field theories are almost, but not quite, trivial; the state spaces of $\mathcal{A}$ are complex lines. Since $Z$ is a boundary theory to $e^{2\pi i\mathcal{A}}$, the partition function of $Z$ on a closed $d$-manifold $M$ is an element of the state space of $\mathcal{A}$. Formally, we can represent this by saying that the partition function on $M$ arises from evaluating the anomaly theory on an open manifold:
\begin{equation}
Z[M]= e^{2\pi i \mathcal{A}(X)} \,, \quad \partial X = M \,.
\label{34500}
\end{equation}
Here, $X$ is some manifold with boundary $M$, on which any relevant structures (Spin structure, background fields, etc.) are suitably extended. It is in general not possible to uniformly trivialize these state spaces to obtain partition functions that are numbers, representing the ambiguity in the partition function of an anomalous theory. Equivalently, the prescription given in \eq{34500} may depend on the choice of $X$.

This ambiguity will be absent only if the anomaly theory $\mathcal{A}$ is trivial on any closed $(d+1)$-manifold
$X$. When $X$ is a mapping torus, i.e.\ of the form $(M\times [0,1])/\sim$, where $\sim$ acts as a diffeomorphism on $M$ and identifies the two ends of the interval, a non-trivial $e^{2\pi i\mathcal{A}(X)}$ means that the symmetry is anomalous in the background $M$; there is no consistent way to assign a value to the partition function $Z[M]$, as the mapping torus explicitly exhibits a path in configuration space where the phase of this partition function changes.

For more general $X$, the conclusion is not as clear; it is certainly not possible to define $Z[M]$ in a consistent way if splitting and joining of manifolds is allowed \cite{Witten:2015aba,Monnier:2019ytc}, but whether this constitutes an issue depends on the context. Following the nomenclature in \cite{Garcia-Etxebarria:2018ajm}, we will say that a non-vanishing anomaly theory $\mathcal{A}$, which nevertheless vanishes on all mapping tori, has \emph{Dai-Freed} anomalies. In a field theory with Dai-Freed anomalies but no ordinary (perturbative or global) anomalies, the  symmetry is still preserved even at the quantum level, but the theory will not admit a lattice regularization with on site symmetry \cite{Ji:2019eqo}. Depending on the physical context, this can be perfectly fine.

By contrast, in quantum gravity, a Dai-Freed anomaly means that the relevant symmetry is broken by topology-changing processes \cite{Garcia-Etxebarria:2018ajm}. Therefore such a symmetry should not be gauged, and we should demand triviality of $\mathcal{A}$ for any exact (gauged) symmetry in a quantum theory of gravity. A similar argument applies, for the same reasons, for the worldvolume theories of D-branes \cite{TY19}, since the topology of a D-brane worldvolume fluctuates.

So, given a physical theory, how do we determine $\mathcal{A}$? The non-perturbative approach to anomalies uses
classification theorems for various classes of invertible field theories to determine it. In our setting, the
anomaly theory is unitary. In \cite{FH16} Freed and Hopkins classify unitary (by which we actually mean
reflection-positive in the Euclidean setup) invertible field theories in terms of Abelian groups called bordism
groups, which can be computed using standard techniques in algebraic topology. In a little more detail, the
$d^{\mathrm{th}}$ bordism group is the Abelian group of closed $d$-manifolds modulo the equivalence relation where
$M_1$ is equivalent to $M_2$ if $M_1 \amalg M_2$ bounds a compact $(d+1)$-manifold $X$; the Abelian group structure is
disjoint union. The higher-dimensional manifold can intuitively be understood as an appropriate deformation of
one of the boundary components into the other, where topology can change---for example holes or handles can grow and reattach---see Figure \ref{fig:Bordisms}. In the quantum gravity context, the bordism receives a physical interpretation; it can literally be regarded as a generalization of the field theory mapping torus, describing a topology-changing non-contractible path in the configuration space of geometric backgrounds of quantum gravity \cite{Garcia-Etxebarria:2018ajm}.
\begin{figure}
\centering
\includegraphics[width = 0.4 \textwidth]{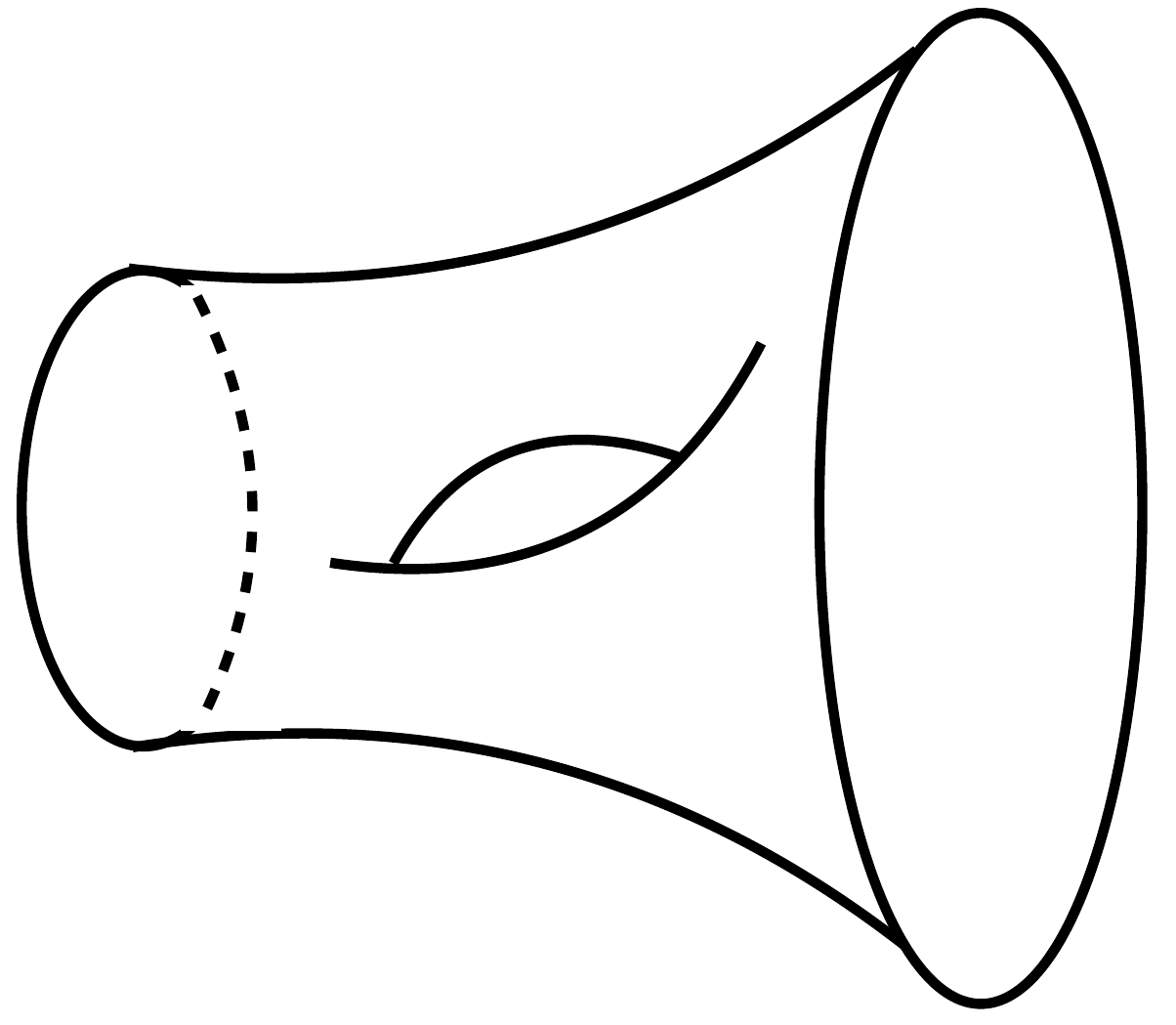}
\caption{Schematic picture of a bordism manifold between two bounding circles.}
\label{fig:Bordisms}
\end{figure}

Importantly, this procedure must take into account the additional structure of the background, which must extend in
the right way to the bordism. For instance, any background gauge fields for symmetries must extend in a
non-singular way, and similarly, if the theory contains fermions, the Spin structure must extend into the
bulk. The additional structure specified by the background is formalized mathematically as a \emph{tangential
structure} $\xi$ \cite{Las63}, and the Abelian group of $d$-manifolds with $\xi$-structure modulo bordism is denoted $\Omega_d^\xi$.

In \cite{FH16} the Abelian group of $(d+1)$-dimensional unitary invertible field theories of $\xi$-manifolds (i.e.\ the group of possible anomalies for $d$-dimensional QFTs on $\xi$-manifolds) is classified as an extension
\begin{equation}
\label{IFT_class}
\shortexact{\mathrm{Tors}(\Hom(\Omega_{d+1}^\xi, \C^\times))}
{\{\text{unitary IFTs}\}}
{\Hom(\Omega_{d+2}^\xi, \Z)}
.\end{equation}
$\mathrm{Tors}(\Hom(\Omega_{d+1}^\xi, \C^\times))$ denotes the torsion subgroup of the Abelian group of
$\C^\times$-valued bordism invariants; these classify the subgroup of unitary IFTs which are topological. That is,
the partition function of a unitary invertible tQFT is a bordism invariant and determines $e^{2\pi i\mathcal{A}}$
up to isomorphism.

The rightmost Abelian group in~\eqref{IFT_class}, $\Hom(\Omega_{n+2}^\xi, \Z)$, captures the perturbative information in an anomaly field theory: It is a group of characteristic classes of manifolds with $\xi$-structure, and the image of the anomaly field theory in this group is the anomaly polynomial. In the particular case of interest in this paper, that of type IIB supergravity, perturbative anomalies vanish famously, due to a miraculous cancellation \cite{Alvarez-Gaume:1983ihn}.

So in short, global as well as Dai-Freed anomalies are captured by an invertible tQFT, whose partition function is
a bordism invariant. Therefore the study of these anomalies is mapped to the question of computing the relevant bordism groups, finding their generators, and evaluating the anomaly theory on them. For the kinds of tangential structures $\xi$ that occur in physics, the bordism groups in~\eqref{IFT_class} are generated by a small number of manifolds, so one can determine the isomorphism class of an anomaly theory by calculating it on that generating set. This approach has been used in
\cite{Wit86, Kil88, Edw91, DMW02, Sat11,Monnier:2011mv,Monnier:2011rk,Monnier:2013kna,
CSR17, TY17,
Hsi18, LT18,Seiberg:2018ntt,Wan:2018bns,Wan:2018zql,
DGL19, Garcia-Etxebarria:2018ajm, HTY19, TY19,Wan:2019oax,Wan:2019soo, TZ19, WWW19,Monnier:2019ytc,
DL20a, DL20b, GOPWW20, Hsieh:2020jpj, Kai20, KPMTD20, Tho20a, Tho20b,
Yu20,Wang:2020mra,Wang:2020xyo,Wang:2020gqr,Wan:2020ynf,
CHTZ21, DL21, FH21, LOT21, SKLT21, Tac21, WY21, Yon21},
and is the approach we will use to determine the duality anomaly of type IIB string theory.


\section{Duality in type IIB string theory}
\label{sec:dual}

In order to identify the correct tangential structure entering the bordism classification discussed above we need to describe a precise version of the duality group of type IIB string theory. Since non-trivial duality backgrounds are then classified by a discrete bundle, we will also construct the relevant characteristic classes that enter in the classification of bordism generators.

\subsection{Duality group of type IIB string theory}
\label{sec:dualgroup}

We will start by introducing the detailed realization of the duality group of type IIB string theory, which will be the main character in the rest of the paper. Some of the material in this section is standard; we refer the reader to \cite{Pantev:2016nze}, and especially the Appendix of \cite{TY19}, for a particularly clear exposition.

It is a well-known fact \cite{Polchinski:1998rr} that the type IIB supergravity action has a perturbative $\SL(2,\mathbb{R})$ symmetry under which the bosonic action remains invariant. In the presence of quantized fluxes and the corresponding charged objects, this symmetry group is broken to the discrete duality group $\SL(2,\Z)$.
This is the group of $2 \times 2$ matrices with integer entries and unit determinant, i.e.,
\begin{equation}
\left(\begin{array}{cc} a & b\\ c& d\end{array}\right) \,,\quad ad-bc=1\,,
\end{equation}
and is generated by two elements conveniently chosen to be \cite{Seiberg:2018ntt}\footnote{The relation to the more common generator $T = \begin{pmatrix} 1 & 1 \\ 0 & 1\end{pmatrix}$ is given by $T =S^{-1} U$.}
\begin{align}
U = \begin{pmatrix} 0 & -1 \\ 1 & 1 \end{pmatrix} \,, \quad S = \begin{pmatrix} 0 & -1 \\ 1 & 0 \end{pmatrix} \,.
\end{align}
In terms of these, $\SL(2,\mathbb{Z})$ can be presented as
\begin{equation}
\SL(2,\mathbb{Z})=\langle U,S \,\vert \, S^4 = 1 \,, \enspace S^2=U^3\rangle \,.
\end{equation}
Interestingly, there is a way to write $\SL(2,\mathbb{Z})$ as an amalgamated free product, which is particularly useful for the computation of bordism groups, see e.g.\ \cite{Seiberg:2018ntt,Hsieh:2020jpj}. This amalgam structure is given as
\begin{align}
\SL(2,\Z) = \Z_4 \ast_{\Z_2} \Z_6 \,,
\end{align}
where the individual factors $\mathbb{Z}_4:\langle S \,\vert \, C\equiv S^2 \,, \enspace C^2= 1\rangle$ and $\mathbb{Z}_6:\langle U \, \vert \, C\equiv U^3 \,, \enspace C^2= 1\rangle$ are identified along a common $\mathbb{Z}_2:\langle C\, \vert \, C^2= 1\rangle$.

The standard action of $\SL(2,\Z)$ on the bosonic fields in the type IIB supergravity action is given by
\begin{align}
\tau = C_0 + i e^{- \phi} \longrightarrow \frac{a \tau + b}{c \tau + d} \,, \quad \begin{pmatrix} C_2 \\ B_2 \end{pmatrix} \longrightarrow \begin{pmatrix} a & b \\ c & d \end{pmatrix} \begin{pmatrix} C_2 \\ B_2 \end{pmatrix} \,, \quad \text{with} \enspace \begin{pmatrix} a & b \\ c & d \end{pmatrix} \in \SL(2,\Z) \,,
\label{eq:SLtrafo}
\end{align}
whereas the RR 4-form field $C_4$ and the spacetime metric are invariant.

Including the fermions of type IIB supergravity one can deduce their transformation under general $\SL(2,\Z)$ transformations, see \cite{Pantev:2016nze}. For that we form complex fermions out of the two Majorana-Weyl gravitini $\Psi_{\mu}^i$ and dilatini $\lambda^i$. These complexified fields transform as
\begin{align}
\Psi_{\mu} = \Psi_{\mu}^1 + i \Psi_{\mu}^2 \longrightarrow \bigg( \frac{c \overline{\tau} + d}{c \tau +d} \bigg)^{1/4} \Psi_{\mu} \,, \quad \lambda = \lambda^1 + i \lambda^2 \longrightarrow \bigg( \frac{c \overline{\tau} + d}{c \tau +d} \bigg)^{-3/4} \lambda \,. \label{tfermions}
\end{align}
The fact that we must take quartic roots in the above expression means that there is a sign ambiguity, and demands an extension of $\SL(2,\Z)$ to a double cover known as the metaplectic group $\text{Mp}(2,\Z)$. This group has the presentation \cite{Hsieh:2020jpj}
\begin{equation}
\text{Mp}(2,\mathbb{Z}) = \langle \hat{U},\hat{S}\vert \, \hat{S}^8 =1 \,, \enspace \hat{S}^2=\hat{U}^3\rangle \,.
\end{equation}
Here, $\hat{S}^4=(-1)^F$ is a central element that gets mapped to the identity under the map $\text{Mp}(2,\Z)\rightarrow \SL(2,\mathbb{Z})$. Given the above, we can also write an amalgam structure for $\text{Mp}(2,\Z)$,
\begin{align}
\Mp(2,\Z) = \Z_8 \ast_{\Z_4} \Z_{12} \,.
\end{align}
where the groups  $\mathbb{Z}_8:\langle \hat{S} \,\vert \, \hat{C}\equiv \hat{S}^2 \,, \enspace \hat{C}^4= 1\rangle$ and $\mathbb{Z}_{12}:\langle \hat{U}\, \vert \, \hat{C}\equiv \hat{U}^3 \,, \enspace \hat{C}^4= 1\rangle$ are identified along a common $\mathbb{Z}_4:\langle \hat{C}\, \vert \, \hat{C}^4= 1\rangle$.

$\SL(2,\mathbb{Z})$ can also be extended in a different way, including an action corresponding to orientation reversal of the type IIB worldsheet, as well as worldsheet left-moving fermion number, see \cite{TY19}. This extends the duality group acting on the bosons to $\GL(2,\Z)$ and the RR 4-form $C_4$ is odd under the additional generator $R$, which can be chosen to be
\begin{align}
R = \begin{pmatrix} 0 & 1 \\ 1 & 0 \end{pmatrix} \,,
\end{align}
of determinant $-1$, thus extending $\SL(2,\Z)$ to $\GL(2,\Z)$. Note that the generators above satisfy
\begin{align}
RSR = S^{-1} \,, \quad RUR= U^{-1} \,,
\end{align}
indicating that $\{R,S \}$ and $\{ R, U \}$ generate the dihedral groups $D_{8}$ and $D_{12}$, respectively, and where we define in general $D_{2n}=\langle R,S \,\vert \, S^n=R^2=1\,, \enspace RSR^{-1}=S^{-1}\rangle$ which is the group of symmetries of a regular $n$-gon. The group $\GL(2,\mathbb{Z})$ can then be presented as
\begin{equation}
\GL(2,\mathbb{Z})=\langle U,S,R \,\vert \, S^4 = 1 \,, \enspace RSR^{-1}=S^{-1} \,, \enspace RUR^{-1}=U^{-1} \,, \enspace S^2=U^3\rangle \,.
\end{equation}
This can also be written as an amalgam
\begin{align}
\GL(2,\Z) = D_{8} \ast_{D_4} D_{12} \,,
\label{eq:GLamalg}
\end{align}
with the subgroups
\begin{align}
\begin{split}
D_8& =\langle R,S \,\vert \, C\equiv S^2 \,, \enspace C^2=R^2=1 \,, \enspace RSR^{-1}=S^{-1}\rangle \,, \\
D_{12}&=\langle R,U \,\vert \, C\equiv U^3 \,, \enspace C^2=R^2=1 \,, \enspace RUR^{-1}=U^{-1}\rangle \,, \\
 D_4 &= \langle C,R \,\vert \, C^2=R^2=1\rangle \,.
\end{split}
\end{align}
Finally, to extend the action of $R$ to fermions, we need to consider a double cover, as above. The Majorana-Weyl
spinors of type IIB supergravity are only compatible with a double cover that squares reflections to the identity,
$\hat{R}^2=1$, i.e.\ the Pin$^+$ cover. We finally arrive at the full duality group of IIB supergravity \cite{TY19}, the Pin$^+$ cover of $\GL(2,\Z)$, which we denote as $\GL^+(2,\Z)$ for short. It has the presentation
\begin{equation}
\label{GL_plus_is_amalgam}
\GL^+(2,\mathbb{Z})=\langle \hat{U},\hat{S},\hat{R} \,\vert \, \hat{S}^8 = 1 \,, \enspace \hat{S}^2=\hat{U}^3 \,, \enspace \hat{R}^2=1 \,, \enspace \hat{R}\hat{S} \hat{R}=\hat{S}^{-1} \,, \enspace \hat{R}\hat{U} \hat{R}=\hat{U}^{-1}\rangle \,.
\end{equation}
As above, this group is an amalgam of dihedral groups,
\begin{equation}
\GL^+(2,\Z) = D_{16} \ast_{D_8} D_{24} \,,
\label{eq:GL+amalg}
\end{equation}
where
\begin{equation}
\begin{split}
D_{24}&=\langle \hat{U},\hat{R}\,\vert\, \hat{C}\equiv \hat{U}^3 \,, \enspace \hat{C}^{4}=1,\, \hat{R}^2=1 \,,\enspace \hat{R}\hat{U}\hat{R}=\hat{U}^{-1}\rangle \,,\\
 D_{16}&=\langle \hat{S},\hat{R}\,\vert\, \hat{C}\equiv \hat{S}^2 \,, \enspace \hat{C}^4=1 \,,\enspace \hat{R}^2=1 \,,\enspace \hat{R}\hat{S}\hat{R}=\hat{S}^{-1}\rangle \,, \\
D_8&=\langle \hat{C},\hat{R}\,\vert\, \hat{C}^4=1,\, \hat{R}^2=1,\, \hat{R}\hat{C}\hat{R}=\hat{C}^{-1}\rangle \,.
\end{split}
\end{equation}
As explained beautifully in the Appendix of \cite{TY19}, the appearance of $\GL^+(2,\Z)$ can be understood geometrically in terms of F-/M-theory duality, which maps the duality group of type IIB to the group of large diffeomorphisms of the M-theory torus. Since M-theory makes sense on non-orientable manifolds, this is $\GL(2,\mathbb{Z})$; but since a $\text{Pin}^+$ structure is required \cite{Witten:1996md,Witten:2016cio}, the symmetry group is actually the $\text{Pin}^+$ cover $\GL^+(2,\Z)$.

For our later calculations it will prove useful to evaluate the expressions at certain values of the axio-dilaton which are invariant under some of the generators of the duality group. At these points the action on the fermions takes a particularly simple form in terms of a complex phase. For the first factor in \eqref{eq:GL+amalg} the relevant value is given by $\tau = i$, which is invariant under $\hat{S}$. For the second factor in \eqref{eq:GL+amalg} one has $\tau = e^{2 \pi i/3}$ invariant under $\hat{U}$. At these special points the transformations of the fermions read
\begin{align}
{\renewcommand{\arraystretch}{1.2}
\begin{array}{c | c | c}
\tau = i & \hat{S} \, \Psi_{\mu} = e^{2 \pi i \frac{1}{8}} \Psi_{\mu} & \hat{S} \, \lambda = e^{- 2 \pi i \frac{3}{8}} \lambda \\ \hline
\tau = e^{2 \pi i / 3} & \hat{U} \, \Psi_{\mu} = e^{2 \pi i \frac{1}{12}} \Psi_{\mu} & \hat{U} \, \lambda = e^{- 2 \pi i \frac{3}{12}} \lambda
\end{array}
}
\label{eq:fermtrafo}
\end{align}
i.e., the gravitino has charge $1$ and the dilatino has charge $-3$ under the corresponding duality transformations. These have to be supplemented by the transformation property of the chiral 4-form
\begin{equation}
\hat{R} \, C_4 = - C_4 \,,
\label{eq:C4trafo}
\end{equation}
which only transforms under orientation reversal.

\subsection{Non-trivial duality backgrounds}
\label{subsec:dualback}

With the complete version of the duality identified, we can investigate the different non-trivial duality backgrounds of type IIB.

The fact that the duality group of type IIB involves fermion parity $(-1)^F$ has interesting consequences, since this $(-1)^F$ should be identified with the center of the Spin cover of the underlying spacetime manifold. This twists the Spin
structure of the spacetime manifold and the duality symmetry into what we call a Spin-$\GL^+(2,\Z)$ structure. In particular,
this means that on a general type IIB background, the fermion fields are sections of an associated vector bundle for the group
\begin{equation}
\frac{\text{Spin} \times \GL^+(2,\Z)}{\mathbb{Z}_2} \,.
\label{32dd}
\end{equation}
Thus, type IIB supergravity makes sense on manifolds that do not have a Spin structure, but which do have a Spin-$\GL^+(2,\Z)$ structure. The most familiar example of manifolds with twisted Spin structures are Spin$^c$ manifolds, where the fermions are charged under an additional $U(1)$ bundle. But there are many other examples of twisted Spin structures; for instance, see \cite{Hsi18,HTY19,Garcia-Etxebarria:2015wns}. In general, a Spin-$G$ structure, for $G$ a group with a $\mathbb{Z}_2$ center, describes fermions whose transition functions take values in a group like \eq{32dd} with $\GL^+(2,\Z)$ replaced by $G$. For instance, in later realizations we consider examples of Spin-$\GL^+(2,\Z)$ manifolds which are in the image of the map $D_{16}\rightarrow \GL^+(2,\Z)$ given by the amalgamation above, and we refer to them as having a Spin-${D_{16}}$ structure. Similarly, we also discuss Spin-${D_{8}}$ manifolds. In these cases the twisting also restricts the allowed representations of the fermions under the factors in \eqref{eq:GL+amalg}, see Appendix \ref{app:dihedralrep}.

In the familiar Spin$^c$ case, the fermions are not sections of a $U(1)$ bundle, but one can construct a principal $U(1)$ bundle by squaring the transition functions of the fermions.
Similarly, on a general Spin-$\GL^+(2,\Z)$ manifold we will not have a well-defined $\GL^+(2,\Z)$ duality bundle, but there is a natural associated $\GL^+(2,\Z) / \Z_2 \simeq \GL(2,\Z)$ principal bundle. We will use this principal  $\GL(2,\Z)$ bundle to characterize the duality background we have turned on; as usual, it can be efficiently described by using characteristic classes, which are obtained by pulling back cohomology classes of the associated classifying space $B\GL(2,\Z)$. The existence of the Spin-$\GL^+(2,\Z)$ structure is equivalent to a certain condition involving tangent bundle Stiefel-Whitney classes and characteristic classes of the $\GL(2,\Z)$ bundle, which we describe below. Again, this is in complete analogy to the more familiar Spin$^c$ case, where the Chern class of the associated $U(1)$ bundle $c_1$ is related to $w_2$ of the tangent bundle by $w_2=c_1 \text{ mod } 2$. We now describe several characteristic classes that will be important in our later discussions:

\begin{description}
\item[Mod 2 characteristic classes:] From the amalgam structure \eqref{eq:GLamalg} it follows that at primes 2 and 3 $B\GL(2,\Z)$ has the same cohomology ring as $BD_8$ and $BD_6$, respectively.\footnote{Although this is not immediate, the logic we follow here is the same as in Exercise 3 of Chapter II.7 of \cite{brown1982cohomology} (see also \cite{Seiberg:2018ntt}). Further details will be given in \cite{Upcoming}.}
The cohomology ring of $BD_8$ with $\Z_2$ coefficients is generated by three classes, $x$, $y$, and $w$ of
degrees $1$, $1$, and $2$, respectively. They are subject to the relation $x y = y^2$, i.e.,
\begin{equation}
H^* (BD_8, \Z_2) = \frac{\Z_2 [x,y,w]}{(xy = y^2)}.
\end{equation}
See~\cite[Theorem 4.6]{Sna13},~\cite[\S 2.3]{Tei92}, or~\cite[Theorems 5.5 and 5.6]{Han93}. The generators can be
described as Stiefel-Whitney classes of associated bundles:
	\begin{itemize}
		\item Let $\rho\colon D_8\to\mathrm O(2)$ denote the standard representation of $D_8$ as the symmetries of
		a square (see Appendix~\ref{app:dihedralrep} for some information on representations of dihedral groups),
		and let $V_\rho\to BD_8$ be the associated rank-$2$ vector bundle. Then $x = w_1(V_\rho)$ and $w =
		w_2(V_\rho)$.
		\item Let $\chi\colon D_8\to\{\pm 1\}$ be the character in which quarter turns are sent to $-1$ and
		reflections are sent to $1$, and let $L_\chi\to BD_8$ be the associated line bundle. Then $y =
		w_1(L_\chi)$.
	\end{itemize}
\item[Mod 3 characteristic classes:] The cohomology of $BD_6$ with $\Z_3$ coefficients is generated by two
classes $q,\tilde q$ in degrees $3$ and $4$, respectively, with the relation $q^2 = 0$:
\begin{equation}
	H^*(BD_6;\Z_3)\cong \frac{\Z_3[q, \tilde q]}{(q^2 = 0)} \,.
\end{equation}
If $\beta\colon H^*(\text{--}; \Z_3)\to H^{*+1}(\text{--};\Z_3)$ denotes the Bockstein homomorphism associated to
the short exact sequence
\begin{align}
	\shortexact{\Z_3}{\Z_9}{\Z_3},
\label{eq:Z3shortextact}
\end{align}
then $\tilde q = \beta(q)$. See \cite{Huy75} and \cite{Man78} for a proof, using the fact that $D_6$ is
isomorphic to the symmetric group of order $3$.

One can also consider characteristic classes obtained by pulling back cohomology classes of $BD_6$ for a local
coefficient system, in which the reflection in $D_6$ acts as multiplication by $-1$.
In this way one obtains a class $\hat{q}$ with twisted $\Z_3$ coefficients in degree 1 and $\hat{q}_5$ in degree
5~\cite[Theorems 5.8 and 5.9]{Han93}. In physics terms, ``twisted'' just means that the corresponding classes are
not invariant under $\GL^+(2,\mathbb{Z})$ reflections; but they will still be useful to us.

In fact, all of these characteristic classes can be naturally associated to the cohomology of $B\Z_3$ via pullback under the map
\begin{align}
B\Z_3 \rightarrow BD_6 \,.
\end{align}
The cohomology ring of $B\Z_3$ with $\Z_3$ coefficients is generated by a
class $a$ in degree 1 and a class $\beta (a)$ in degree 2, again connected by the Bockstein homomorphism
associated to \eqref{eq:Z3shortextact}. The cohomology ring is~\cite[Example 3.41]{Hat02}
\begin{align}
H^* (B\Z_3; \Z_3) = \frac{\Z_3 [a, \beta(a)]}{(a^2 = 0)} \,.
\label{eq:Z3cohom}
\end{align}
Note that the class $a$ fully specifies the $\Z_3$ bundle. Embedding $B\Z_3$ into $BD_6$, reflections send $a$ to $-a$. As a result, the pullback of a characteristic class of a $D_6$ bundle of the form $\beta(a)^n \cup a$ represents an element in $H^*(BD_6, \Z_3)$ if $n$ is odd and an element in $H^* (BD_6, \tilde{\Z}_3)$ if $n$ is even, where $\tilde{\Z}_3$ indicates the twisted coefficient system. With this one has, where $\sim$ denotes equivalence under pullback,
\begin{align}
q \sim \beta(a) \cup a \,, \enspace \tilde{q} \sim \beta(a)^2 \,, \enspace \hat{q} \sim a \,, \enspace \hat{q}_5 \sim \beta(a)^2 \cup a \,.
\end{align}
In the following we will write the mod 3 characteristic classes of the duality bundle in terms of $a$ and $\beta(a)$ keeping in mind that they are associated to the cohomology of $BD_6$ with twisted and untwisted coefficients under pullback. Only classes with untwisted coefficients can give rise, via integration, to bordism invariants.

\item[Mod 4 characteristic classes:] These classes can be constructed analogously to the mod 3 classes above, by studying the cohomology rings $H^*(BD_8, \Z_4)$ and $H^*(BD_8, \tilde{\Z}_4)$ using the embedding
\begin{align}
B \Z_4 \rightarrow BD_8 \,,
\end{align}
with $\Z_4$ generating rotations. This leads to a class $b$ in degree 1; if $\beta$ denotes the Bockstein
associated to the short exact sequence
\begin{equation}
	\shortexact{\Z_4}{\Z_{16}}{\Z_4} \,,
\end{equation}
then $b$ and $\beta(b)$ generate the cohomology ring~\cite[Example 3.41]{Hat02}
\begin{align}
H^* (B \Z_4; \Z_4) = \frac{\Z_4 [b, \beta(b)]}{b^2 = 2 \beta (b)} \,.
\end{align}
Again, $b$ fully specifies the $\Z_4$ bundle. The reflections send $b$ to $-b$ and again one can associate the elements in the (un)twisted cohomology of $BD_8$ via the pullback to combinations of $b$ and $\beta(b)$ as above. This identification is understood in the following, where we will denote the mod 4 classes in terms of $b$ and $\beta (b)$.
\end{description}
Finally, with the identification of the characteristic class of the duality bundle we can formulate the requirement for a well-defined Spin-$\GL^+(2,\mathbb{Z})$ structure. For an orientable $d$-dimensional spacetime manifold $M$ with tangent bundle $TM$ the existence of a
$\text{Spin-} \GL^+(2,\Z)$ structure demands a correlation between the second Stiefel-Whitney class of the
tangent bundle $w_2 (TM)$ and the characteristic class $w$ of the principal $\GL(2,\Z)$ bundle, namely that
\begin{align}
w_2 (TM) = w \,,
\end{align}
where $w$ is the mod 2 characteristic class described above.\footnote{I.e., a Spin-GL$^+ (2,\mathbb{Z})$ structure encodes information about the trivialization of $w_2 (TM) - w$.}


\section{Duality anomalies of type IIB string theory}
\label{sec:gen}

We will now apply the general discussion of Section \ref{sec:bordanom} to the particular case of type IIB string theory. As reviewed in
Section \ref{sec:dualgroup}, there are three versions of the duality group of type IIB string theory, in
which one successively includes the effects of fermions, and of orientation reversing worldsheet symmetries.
Consequently, there are three bordism groups one could discuss:\footnote{For the computation of these bordism groups and their generators
we used techniques including the Atiyah-Hirzebruch as well as the Adams spectral sequence. We will go into these
computations in detail in the upcoming work \cite{Upcoming}; see also \cite{KPMTD20,2018arXiv180107530B}.}
\begin{equation}
\begin{split}
\Omega_{11}^{\Spin} \big(B\SL(2,\Z)\big) &\cong (\Z_2)^{\oplus 2}\oplus (\Z_8)^{\oplus 2} \oplus \Z_{128} \oplus
\Z_3 \oplus \Z_{27} \,,\\
\Omega_{11}^{\Spin\text{-} \Mp(2,\Z)} &\cong \Z_8\oplus (\Z_2)^{\oplus 2} \oplus\Z_3\oplus \Z_{27} \,,\\
\Omega_{11}^{\Spin\text{-}\GL^+(2, \Z)} &\cong\Z_8 \oplus (\Z_2)^{\oplus 9} \oplus \Z_3 \oplus \Z_{27}.
\end{split}\label{e4588}
\end{equation}
In the above, the notation $\Omega_*^{\Spin\text{-}G}$ means $\Spin\text{-}G$ bordism, which is different
from $\Omega_*^\Spin(BG)$, where there is no twist (in the same sense used above line \ref{32dd})).
Since type IIB string theory contains fermions, the first group is not of direct physical relevance. However, this illustrates how the introduction of fermions already gets rid of many potential anomalies that could have been realized by an $\SL(2,\Z)$-invariant bosonic theory. Enlarging the structure group introduces new equivalence relations between manifolds, thereby reducing the potential anomalies (see \cite{McNamara:2019rup} for more instances of the same phenomenon).

As discussed in full generality above, our task is to determine the anomaly theory $\mathcal{A}$ of type IIB supergravity in terms of the characteristic classes of the spacetime manifold and the duality bundle and evaluate it on 11-manifolds that represent the generators of the relevant bordism groups above. We directly construct the anomaly theory associated to the duality group $\GL^+ (2,\Z)$ and then determine a complete list of generators for the bordism classes in the last entry of \eq{e4588}. This enables us to study the presence of duality anomalies in full generality.

\subsection{The IIB  anomaly theory}
\label{ssec:IIBan}

We are finally in a position to study the duality anomaly of type IIB supergravity. As reviewed in Section \ref{sec:dualgroup}, classical type IIB supergravity---the effective field theory that arises as the low-energy limit of IIB string theory---can be formulated on $\Spin\text{-}\GL^+(2, \Z)$ manifolds, and we would like to see if this feature survives the inclusion of quantum effects. Determining the anomaly theory of a general QFT is not an easy task; and here we have one coupled to gravity.  But because it has a large amount of supersymmetry, type IIB supergravity has a path-integral formulation \cite{Denef:2008wq,Polchinski:1998rr}, and the familiar perturbative formulas for fermion anomalies reviewed in \cite{Witten:2015aba,Garcia-Etxebarria:2018ajm} can be used, even at strong coupling.\footnote{Strictly speaking the path integral prescription also includes integration over the graviton variable, since this is a dynamical field in IIB string theory. As usual, by the low-energy EFT we mean the theory of all fields excluding gravity.} A similar situation takes place for 11-dimensional M-theory, or the worldvolume theory of M2-branes \cite{Witten:1996md,Witten:2016cio}, which have no weak coupling but whose low-energy limit is also controlled by a path integral description.

We should also note the works \cite{Gaberdiel:1998ui,Minasian:2016hoh} which study a duality anomaly of type IIB supergravity/F-theory as well. This anomaly is, however, different from the ones discussed below, and strictly speaking, it only arises after gauging the duality group of type IIB supergravity. We elaborate on the connection between the two anomalies in Appendix \ref{app:SL2R}.

So which supergravity fields can have an anomalous variation under duality transformations? As described in Section
\ref{sec:dualgroup}, the usual chiral fields of type IIB supergravity, namely the gravitini, dilatini, and the
self-dual chiral 4-form, all transform under the duality group. Moreover, $(C_2, B_2)$ are not invariant but
instead transform in the two-dimensional representation of the duality group \eqref{eq:SLtrafo} (see Section
\ref{sec:dualgroup} or e.g.~\cite{Denef:2008wq}, sec. 3.1).

A similar story holds for the dual 6-form fields, as well as the axio-dilaton. Although these fields transform non-trivially under dualities, they do not have an anomalous variation. The easiest way to see this is to exhibit the possibility of turning on a symmetry-preserving mass term in the Lagrangian \cite{Witten:2015aba}. Equivalently, if one can construct a Pauli-Villars regulator for the field in question while preserving the symmetries, the field is not anomalous. Although we will not do it in detail, this turns out to be the case for both the axio-dilaton $\tau$ as well as the 2-form fields $(C_2,B_2)$. This means that the integration over these fields in the path integral will not contribute to the duality anomaly; however, the background values of the fields, which can appear in additional topological couplings can (and do) affect the anomaly. Moreover, the supergravity fields can transform under higher-form symmetries, e.g.,
\begin{align}
C_2 \rightarrow C_2 + d\Lambda_C \,,\quad B_2 \rightarrow B_2 + d \Lambda_B \,.
\label{gre0}
\end{align}
These could also have anomalies of their own, or mixed anomalies with dualities, diffeomorphisms, etc. (see e.g.\ \cite{Apruzzi:2020zot} for a recent example of the phenomenon). For these anomalies we expect all fields to contribute.\footnote{Generically, we expect to find anomalies in most of these cases. This is not an inconsistency, since most of these symmetries are broken explicitly by the various branes in type IIB string theory \cite{Bergshoeff:1995sq, Polchinski:1998rr}.} While we do not include these mixed anomalies, elucidating the full symmetry type of ten-dimensional supergravities at the quantum level is an important open problem. To summarize, in the setups that we consider below, the only fields that contribute to the duality anomaly are the usual suspects: the fermions and the self-dual 4-form.

The full anomaly theory of type IIB supergravity, including duality symmetry, and evaluated on an 11-manifold $X$, is given by a generalization of the theory introduced in \cite{Hsieh:2020jpj}, including the representations under the duality group:
\begin{equation}
\mathcal{A}(X) = \eta^{\text{RS}}_{1}(X) - 2 \eta^{\text{D}}_1(X) - \eta^{\text{D}}_{-3}(X) - \tfrac{1}{8}\eta_-^{\text{Sig}}(X) + \text{Arf}(X) - \tilde{\mathcal{Q}}(\breve{c}) \,.
\label{ath}
\end{equation}
 Here, $\eta^{\text{RS}}_q$ denotes the $\eta$-invariant of the 11d Rarita-Schwinger operator (a Dirac operator coupled to the tensor product of the tangent and Spin bundles), coupled to
${\Spin\text{-}\GL^+(2, \Z)}$ in the representations given by \eq{tfermions} (the subscripts of $+1$ for gravitino
and $-3$ for dilatino denote the effective $U(1)$ charges one would get from the representation \eq{tfermions} by
embedding into $\Mp(2,\mathbb{R})$ as done in \cite{Gaberdiel:1998ui,Minasian:2016hoh}).  Similarly, $ \eta^{\text{D}}_q$ is the $\eta$-invariant of the ordinary Dirac operator coupled to the same representation
\cite{atiyah_patodi_singer_1975}.
Finally, the last three terms $\eta_-^{\text{Sig}}$ (the minus subscript indicates the transformation properties under orientation reversal \eqref{eq:C4trafo}), $\text{Arf}$, and $\tilde{\mathcal{Q}}(\breve{c})$ come from the anomaly theory of the self-dual field coupled to the ${\Spin\text{-}\GL^+(2, \Z)}$ structure background as in \cite{Hsieh:2020jpj}, which is more complicated and will be discussed briefly below.

In order to elucidate the terms on the right-hand side of \eqref{ath} as well as their physical origin in the context of associated index theorems we will describe them individually in the following:

\begin{itemize}
\item{The dilatini comprise a complex Weyl fermion $\lambda$ transforming in the representation \eq{tfermions} of the duality group. The anomaly theory is determined in terms of the 11d $\eta$-invariant $- \eta_{-3}^{\text{D}}$ of a charged fermion.\footnote{The additional minus sign corresponds to the fact that we use conventions in which the gravitino has positive chirality, i.e., the dilatino has negative chirality.}

This $\eta$-invariant can be connected to the APS index theorem \cite{atiyah_patodi_singer_1975} in the following way. Let $Y$ be a 12-dimensional Spin-$\GL^+(2, \Z)$ manifold with boundary $\partial Y = X$; then the $\eta$-invariant on $X$ is related to the index on $Y$ as follows
\begin{equation}
\eta_q^{\text{D}}(X) = \text{Index}^{\text{D}} (Y) - \int_{Y} I^{\text{D}} \,,
\label{i1}
\end{equation}
where $I^{\text{D}}$ is the usual index density; it is the same as in the purely gravitational case (i.e., it is determined in terms of the $\hat{A}$-genus), because the duality group is discrete. In this paper we will be interested in 11-manifolds $X$ which are not boundaries; the computation of the $\eta$-invariants is more subtle in this case.}

\item{The complex gravitino $\Psi_{\mu}$ transforms in the corresponding representation \eq{tfermions} of the duality group. To determine the anomaly theory, we note that an 11-dimensional Rarita-Schwinger operator has a 10-dimensional boundary mode consisting of a 10d Rarita-Schwinger field plus a Dirac fermion of opposite chirality. Moreover, the 10-dimensional Rarita-Schwinger field itself decomposes as a gravitino plus a second Dirac field representation. The anomaly theory associated to the gravitino is then \cite{Alvarez-Gaume:1984zlq,Garcia-Etxebarria:2018ajm,Hsieh:2020jpj}
\begin{equation}
\eta_1^{\text{Gravitino}} = \eta_1^{\text{RS}} - 2 \eta_1^{\text{D}} \,.
\end{equation}
In relating $\eta^{\text{RS}}$ to a 12-dimensional index via the APS theorem, we must take into account that the 12-dimensional Rarita-Schwinger operator reduces to an 11-dimensional Rarita-Schwinger field plus a Dirac fermion. Thus, the correct expression is that on $Y$ with boundary $X$,
\begin{equation}
\eta_q^{\text{RS}}(X) = \text{Index}^{\text{RS}} - \int_{Y} \big( I^{\text{RS}} - I^{\text{D}} \big) \,.
\label{i2}
\end{equation}}

\item Finally, as discussed in Section \ref{sec:dualgroup}, the chiral 4-form $C_4$ with self-dual 5-form field strength picks up a sign under reflections in $\GL(2,\Z)$, see \eqref{eq:C4trafo}. The corresponding anomaly theory has been worked out in \cite{Hsieh:2020jpj}, and it includes three terms:
\begin{equation}
\mathcal{A}_{\text{4-form}}(X)= - \tfrac{1}{8} \eta_-^{\text{Sig}}(X) + \text{Arf}(X)- \tilde{\mathcal{Q}}(\breve{c}) \,.
\label{panom}
\end{equation}
Here, $\eta_-^{\text{Sig}}$ is (modulo an integer) the $\eta$-invariant of the operator appearing as the boundary contribution to the APS index theorem for the 12-dimensional signature operator,
\begin{equation}
\eta_-^{\text{Sig}} = \text{Signature} - \int_{Y} \text{L} \,,
\label{i3}
\end{equation}
where $\text{L}$ is the Hirzebruch L-genus (see e.g.~\cite{Alvarez-Gaume:1984zlq}).

As discussed at length in \cite{Monnier:2011mv,Monnier:2011rk,Monnier:2013kna,Hsieh:2020jpj}, to properly define
the partition function of a self-dual field requires specifying a quadratic refinement $\tilde{\mathcal{Q}}$ of the
bilinear pairing in differential cohomology\footnote{Or a more sophisticated differential cohomology theory. For instance, in perturbative string theory, RR fields are quantized in differential K theory. This theory has a map to differential cohomology, and the discussion in the main text holds, with the caveat that one must restrict to differential cohomology classes in the image from the map in differential K theory. The differential K-theory description is difficult to reconcile with duality \cite{DMW02}; We expect that similar features hold for whichever differential cohomology theory must be used when general duality bundles are turned on.}. But because reflections in the duality group act on $C_4$, this
pairing is actually defined on a twisted differential cohomology group. We take the next few paragraphs to explain
how to define this pairing, the quadratic refinement, and the Arf invariant.

If $X$ is a Spin-$\GL^+(2, \Z)$-manifold, it has a canonical local system $L$, defined to be the associated local
system to the duality $\GL(2, \Z)$-bundle via the determinant $\det\colon \GL(2, \Z)$ $\to\mathrm{Aut}(\Z) = \{\pm
1\}$. In other words, if the monodromy of the duality bundle around a class $\gamma\in\pi_1(X)$ is a reflection,
$\gamma$ acts on $L$ by $-1$; otherwise $\gamma$ acts by the identity. Because reflections in $\GL(2, \Z)$ act by
$-1$ on $C_4$, $C_4$ is actually a cocycle for the twisted differential cohomology group $\check H^5(X; L)$.
$\GL(2, \Z)$ acts by the identity on $L\otimes L$, which means that the product of two twisted cohomology classes
untwists.

The chiral $4$-form field $C_4$ is the boundary mode of a $5$-form field $C_5$, which is what is actually used in the construction of the anomaly theory. The field strength $F_6$ of $C_5$ is a cocycle for $\check H^6(X; L)$. The bilinear pairing is a map
\begin{equation}
	\langle \text{--}, \text{--}\rangle \colon \check H^6(X;L)\times\check H^6(X;L) \to \R/\Z
\end{equation}
given by ``cup product, then integrate.'' Specifically, the product in differential cohomology is a map
\begin{equation}
	*\colon \check H^6(X;L)\times\check H^6(X;L) \to \check H^{12}(X; L\otimes L) \cong \check H^{12}(X; \Z).
\end{equation}
Integration lowers the degree by $\dim(X)$, so when $X$ is $11$-dimensional, the map lands in $\check
H^1(\mathrm{pt})\cong\R/\Z$, as promised.\footnote{The anomaly theory $\mathcal A$ is topological, and therefore
one might expect that this pairing can be defined on $H^6(X; L)\times H^6(X; L)$, rather than on differential
cohomology, and by replacing $F_6$ with its image under the characteristic class map, which lives in $H^6(X; L)$.
This is true: in this case this is the usual torsion pairing as described in~\cite{MA}, albeit with a
twist. Once again the fact that the cup product of two twisted classes untwists means the definition goes through.
Computing with ordinary cohomology or with differential cohomology gives the same value for the anomaly theory. See also the discussion in Appendix \ref{app:Q11}.}

Now suppose $\langle\text{--}, \text{--} \rangle\colon A\times A\to \R/\Z$ is any bilinear pairing on an Abelian
group $A$. A quadratic refinement of $\langle\text{--}, \text{--} \rangle$ is defined to be a map $\tilde{\mathcal
Q}\colon
A\to\R/\Z$ satisfying
\begin{equation}
	\langle v, w \rangle = \tilde{\mathcal Q}(v + w) - \tilde{\mathcal Q}(v) - \tilde{\mathcal Q}(w) +
	\tilde{\mathcal Q}(0) \,,
\end{equation}
for all $v,w\in A$. If $A$ is finite, the Arf invariant of $\tilde{\mathcal Q}$, denoted
$\mathrm{Arf}(\tilde{\mathcal Q})\in\R/\Z$, is defined to satisfy
\begin{equation}
 \text{Arf}(\tilde{\mathcal Q})= \frac{1}{2\pi}\mathrm{arg}\left(\sum_{a\in A} e^{2\pi i
	\tilde{\mathcal{Q}}(a)}\right) \,.
\label{eq:arfcalc}
\end{equation}
Since $\check H^6(X; L)$ need not be finite, we restrict the bilinear pairing to flat differential cohomology
classes which are torsion; the subgroup of such classes is finite when $X$ is compact.

The third term in \eq{panom}, $\tilde{\mathcal{Q}}(\breve{c})$, which involves the quadratic refinement, accounts for the coupling of the self-dual form to a background 5-form field $\breve{c}$ \cite{Hsieh:2020jpj}. This term is actually essential in IIB supergravity, because the Chern-Simons coupling
\begin{equation}
S_{\text{CS}} \sim \int C_4\wedge F_3\wedge  H_3
\end{equation}
implies that the potential $B_2\wedge H_3$ (or more precisely, its differential cohomology version $\breve{c}\sim\breve{B}_2*\breve{C}_2$) acts as a background field for the self-dual form \cite{Monnier:2013kna,Hsieh:2020jpj}. An 11-dimensional background in which $\tilde{\mathcal{Q}}(\breve{B}_2*\breve{C}_2)\neq0$ signals a mixed anomaly involving the 4-form and gauge transformations for the $B_2$, $C_2$ fields. Since the latter transform non-trivially under the duality group, there could also potentially be a mixed anomaly as well. As we will see later on, cancelling anomalies in discrete symmetries involves the addition of topological terms to the action; we expect that, in any background where the sum over $B_2$, $C_2$ induces an anomaly, there will be topological terms that cancel it, rendering IIB supergravity well-defined. Studying these terms would be very interesting on its own, but lies outside of the scope of the present paper.\footnote{In fact, the proper formulation of RR fields involves differential K-theory, and treats $C_4$ and $C_2$ in a unified manner \cite{Freed:2000tt,Freed:2002qp,DMW02,Distler:2009ri,Distler:2010an}. It is not known how to make this formulation compatible with duality, and it is conceivable that doing this would also solve the issues with the choice of quadratic refinement, discussed below. The results of this paper point to natural structures in type IIB string theory, which will hopefully be reproduced by a more delicate analysis.}

So far, we have discussed how the anomaly theory of the self-dual field works, but there is an essential issue we have sidestepped;\footnote{We are indebted to Y. Tachikawa and K. Yonekura for bringing this point to our attention.} we have not discussed the construction of the quadratic refinement $\tilde{\mathcal{Q}}$. In fact, canonical quadratic refinements depend on the particular cohomology theory under study, and do not exist in general for oriented manifolds. Reference \cite{Hsieh:2020jpj} (see also \cite{HS05}), constructed a canonical $\tilde{\mathcal{Q}}$ for ten-dimensional Spin manifolds using differential K-theory, but the construction does not extend to non-Spin manifolds or situations with a non-trivial duality bundle, which are precisely the ones of interest in the present paper.

We do not know how to extend the construction of \cite{Hsieh:2020jpj} to provide a quadratic refinement for general $\text{Spin-}\GL^+(2,\mathbb{Z})$ manifolds, or to provide an alternative. Given this, one could entertain the possibility that there is no canonical choice of quadratic refinement outside of the realm of Spin manifolds considered in \cite{Hsieh:2020jpj}.  In this case, specifying the quadratic refinement would be part of the data needed to make sense of the partition function of a 10d theory with self-dual fields, analogous to how, for example, one must specify a choice of Spin structure in theories with fermions. In this case, the bordism groups we have computed, which ignore this information, would not provide an exhaustive classification of the anomalies, and there would be more global anomalies than the ones we consider in this paper.\footnote{Let us note that at least in the context of 6d chiral 2-forms as defined by their coupling to a bulk 7d Chern-Simons-like theory of three-forms,  additional data such as a Wu structure is needed to properly quantize the edge mode theory (see e.g. \cite{Monnier:2016jlo}). Here, a Wu structure functions as the higher-dimensional analog of specifying a Spin structure for 3d Chern-Simons and its coupling to chiral edge modes. In the present context specified by quantum gravity, this option is less natural because fixing such a choice ``from the start'' is somewhat awkward. For example, in the 2d worldsheet theory of a superstring,
one actually sums over possible Spin structures.}

The other possibility is that there is a canonical choice of  quadratic refinement for each
$\text{Spin-}\GL^+(2,\mathbb{Z})$ manifold, even if we cannot construct it at present. Although we cannot
rigorously prove it, this is the more natural possibility consistent with M-/F-theory duality, since M-theory only
requires an $\mathfrak{m}^c$ structure to make sense \cite{FH21}, and at no point does one need to specify data analogous to
a quadratic refinement. A related comment here is that similar considerations apply to compactifications in any number of dimensions.
Indeed, given an F-theory compactification on an elliptically fibered space $Y_{D}$, reduction on a circle takes us to M-theory on $Y_{D}$ (in the limit of large elliptic fiber), so duality again suggests that this additional structure should not be required to make sense of the
corresponding F-theory backgrounds.\footnote{In the context of F-theory in its original formulation as a 12d theory \cite{Vafa:1996xn} on a background geometry of signature $10+2$, the corresponding graviton supermultiplet contains both a 4-form potential $C_4$ and a privileged 1-form $\mu$ (see e.g. \cite{Castellani:1982ke, Nishino:1997gq}), and as proposed in \cite{Heckman:2017uxe}, this can alternatively be formulated in terms of a chiral 5-form $C_5$ which produces the 4-form of $10+2$ supergravity via $C_4 = \mu \cdot C_5$ as in \cite{Heckman:2017uxe}. Viewing 10d type IIB supergravity as an edge mode of this bulk 12d theory, there is a corresponding topological coupling $\mu \wedge C_5 \wedge dC_5$. Reduction on a timelike circle descends to the 11d topological term we have been discussing, while reduction on a null circle passes directly to the 10d edge mode theory and the theory of a chiral 4-form in ten dimensions.
It would be interesting to make further contact between our current analysis and the more speculative aspects of \cite{Heckman:2017uxe},
but we defer such issues to future work.}

Given this state of affairs, in this paper we will \emph{assume} that there is a canonical choice of quadratic refinement for each $\text{Spin-}\GL^+(2,\mathbb{Z})$ manifold. In fact, in most cases that will be of interest to us later on, we will be able to determine which quadratic refinement should be chosen in each manifold we consider, solely from the requirement of anomaly cancellation. Amazingly, we will find that anomalies can always be cancelled by some (essentially unique) choice of quadratic refinement. Thus, our results should be regarded as a ``bottom-up'' approach, in which we are able to bootstrap the correct quadratic refinement. In turn, this can be interpreted as providing experimental evidence suggesting that the choice of quadratic refinement is indeed unique. However, since we are only interested in anomalies involving the duality bundle, we will set $\breve{c}=0$ for the time being. That being said, the term $\tilde{\mathcal{Q}}(\breve{c})$ will make an important appearance later on.

\end{itemize}

Putting the above contributions together, we recover \eq{ath}. As a cross-check of the above, one can evaluate the
anomaly theory on a manifold $X$ with $[X] = 0$ in $\Omega_{11}^{\Spin\text{-}\GL^+(2, \Z)}$, i.e., $X$ bounds
a Spin-$\GL^+(2, \Z)$ $12$-manifold $Y$. The APS index theorems \eq{i1}, \eq{i2} and \eq{i3} (after taking into account the Arf invariant contribution, too) allow one to rewrite \eq{ath} as
 \begin{equation}
 I^{\text{RS}} - 4 I^{\text{D}} - \tfrac{1}{8} \text{L} = 0 \,,
 \label{ewe}
 \end{equation}
 which is the celebrated type IIB anomaly cancellation identity \cite{Alvarez-Gaume:1983ihn}.

\subsection{Computation of the anomaly}
\label{ssec:computation}

We now turn to the central question of this paper: is the theory \eq{ath} non-trivial for some Spin-$\GL^+(2, \Z)$ manifolds? Equation \eq{ewe} shows that the anomaly theory \eq{ath} is a bordism invariant. We have computed the relevant bordism group, which is
\begin{equation}
\Omega_{11}^{\Spin\text{-}\GL^+(2, \Z)} \cong \Omega_{11}^{\Spin\text{-}D_{16}}\oplus \Omega_{11}^{\Spin}(BD_{24})
= \Z_8 \oplus (\Z_2)^{\oplus 9} \oplus \Z_{27} \oplus \Z_3 \,,
\label{wee}
\end{equation}
where again, the notation $\Omega_*^{\Spin\text{-}G}$ means $\Spin\text{-}G$ bordism, which is different from $\Omega_*^{\text{Spin}}(BG)$, where there is no twist (in the same sense used above line (\ref{32dd})).
The result \eq{wee} comes from a combination of Adams spectral sequence techniques and computations of $\eta$-invariants, which we will report (alongside bordism groups of lower degree) in a separate publication \cite{Upcoming}. For our considerations here, the importance of having \eq{wee} is to guarantee that there are no more anomalies than the ones that we will study momentarily.

To check for anomalies we need to find representatives for the generators of each of the factors in \eq{wee} and evaluate
the anomaly theory \eq{ath} on each of them.\footnote{The generators we use are natural from a mathematical perspective, but they may not be the most natural possibilities from a physics standpoint. For instance, we do not care about how many supercharges are preserved etc. This will be more relevant in our upcoming work \cite{Upcoming}, where we will study lower-dimensional bordism groups and their bordism defects (which are generalizations of S-folds).} One of the generators, $X_{11}$ which we will call ``\href{https://letstalktarot.wordpress.com/2011/01/12/arcanum-xi-the-enchantress/}{\textcolor{black}{Arcanum XI}}'', seems to not have been discussed in the mathematical or physics literature, and we describe it in Section \ref{sec:x11} as well as in more detail in Appendix \ref{app:A}. Moreover, we relegate the details of the calculations of the anomalies to Appendix \ref{app:Anomcomp}, where several useful formulas including the $\eta$-invariants of spin-$\tfrac32$ fermions on lens spaces are derived. The results are summarized in the following table, where we list, for each of the factors in \eq{wee}, a generator, a cohomology class or $\eta$-invariant that detects it (using the notation in Section \ref{sec:dualgroup}), and the value of the anomaly theory on each of them:

\begin{equation}
\begin{array}{c|c|c|c}
\text{Factor} & \text{Generator} &\text{Detector} & \mathcal{A}(\text{gen.})\\\hline
\Z_{27} & L_3^{11}&\eta_1^{\text{D}}-\eta^{\text{D}}_3&\tfrac13\\
\Z_{3} & \HP^2\times L_3^{3}& \eta_{1}^{\text{RS}} - \eta_3^{\text{RS}} &\tfrac13\\
\Z_{8} & Q^{11}_4 &\eta_1^{\text{D}}-\eta^{\text{D}}_3& \tfrac{k}{4}  \\
\Z_{2} & \HP^2\times L_4^{3}&\tilde{\eta}_1^{\text{RS}}- 2 \tilde{\eta}_1^{\text{D}} - \tilde{\eta}_{-3}^{\text{D}}&\tfrac12 \\\hdashline
\Z_{2} & \RP^{11} &x^{11}&0\\
\Z_{2} & \widetilde{\RP^{11}}&y^{11}&0\\
\Z_{2} & \HP^2 \times \RP^{3}&w_4^2\, x^3&0\\
\Z_{2} & \HP^2 \times \widetilde{\RP^{3}}& w_4^2\,y^3&0\\
\Z_{2} & X_{10}\times S^1& w_4\, w_6\,x&0\\
\Z_{2} & X_{10}\times \widetilde{S^1}&w_4\,w_6\,y&0\\
\Z_{2} & X_{11}&w_2^4\,x^3&0\,\text{or}\,\frac12\\
\Z_{2} & \widetilde{X_{11}}&w_2^4\,y^3&0\,\text{or}\,\frac12
\end{array}\label{results}\end{equation}

\noindent Here, $\tilde{\eta}$ are reduced $\eta$-invariants introduced in \eqref{eq:reducedeta} in Appendix \ref{app:Anomcomp}.
We now describe some of the manifolds in the second column of \eq{results}:
\begin{itemize}

\item $L^{2k-1}_n$ denotes the lens space $S^{2k-1}/\Z_{n}$, where $\Z_n$ acts as
\begin{equation}
(z_1,z_2,\ldots z_k) \in \mathbb{C}^{2k} \rightarrow e^{\frac{2\pi i}{n}}(z_1,z_2,\ldots z_k) \,,
\end{equation}
and $S^{2k-1}$ is regarded as the unit sphere in $\mathbb{C}^{2k}$. Principal $\Z_n$ bundles over $L^{2k-1}_n$ are
classified by $H^1(L^{2k-1}_n,\Z_n)$. The lens spaces for all the entries in \eq{results} are equipped with the
$\Z_n$ bundle $S^{2k-1}\to L_n^{2k-1}$; the class of this bundle is a generator of the cohomology group $H^1(L_n^{2k-1};\Z_n)$. This $\Z_n$ bundle specifies the associated principal $\GL(2,\Z)$ bundle over the lens space, via the embeddings $\Z_4 \rightarrow \GL(2,\Z)$, $\Z_3 \rightarrow \GL(2,\Z)$ sending the generators to $S$ and $U$, respectively. In these cases, the lift from the $\GL(2,\Z)$ bundle to a $\Spin\text{-}\GL^+(2, \Z)$ bundle is unique, so the solutions are specified completely.

\item $\mathbb{HP}^2$ is the quaternionic projective plane, one of the two generators of $\Omega_8^{\Spin}$ \cite{bams/1183527786}, with a trivial duality bundle over $\mathbb{HP}^2$. As a cross-check, we also computed anomalies for the other generator of $\Omega_8^{\Spin}$, the Bott manifold, although an Adams spectral sequence argument (which we will explain in \cite{Upcoming}) shows that the anomaly on products of Bott manifolds and lens spaces is linearly dependent with the anomaly on $\mathbb{HP}^2 \times L^3_n$ as written here. For completeness, we recall that a Bott manifold is defined as a Spin 8-manifold with unit Dirac index; we consider here the particular example with $p_1=0$ discussed in \cite{FH21}. We take a trivial duality bundle over the Bott manifold. With these choices, all anomalies involving products of Bott manifolds vanish. Although we do not use this fact here, it is worth noting that there are examples of Bott manifolds with $\Spin(7)$ exceptional holonomy \cite{Joyce} which therefore preserve two real supercharges when used as compactification spaces of type II string theory and M-theory.

\item $Q_4^{11}$ is a lens space bundle with fiber $L^9_4$ over the 2-sphere, see e.g.\ \cite{10.4310/jdg/1214459973}. The lens space bundle is obtained by a quotient of the sphere bundle embedded into the rank 5 complex vector composed out of four trivial line bundles and the tensor square of the Hopf line bundle over $S^2$. The group action of $\mathbb{Z}_4$ on the fiber is the same as the action above for lens spaces, with the duality bundle determined by this group action. This manifold is not Spin, but it admits a Spin-$\mathbb{Z}_8$ (and in fact a Spin$^{c}$) structure. Although there is no fundamental obstacle to evaluating the anomaly theory in this background, the non-trivial fibration over the base $S^2$ complicates the computation, which we have not performed. Instead, we study anomalies in this bordism class indirectly by computing the anomalies in the 11-dimensional lens space $L_4^{11}$, which is Spin and generates a $\mathbb{Z}_2$ subgroup of the full $\mathbb{Z}_8$ factor generated by $Q_{4}^{11}$. With the techniques developed in Appendix \ref{app:Anomcomp} we can evaluate the anomaly theory in this background to be $\mathcal{A}(L_4^{11}) = \tfrac12$. Additionally, we can use indirect arguments involving the anomaly cancellation mechanism we discuss below  to establish that $\mathcal{A} (Q_4^{11})=k/4$ for some $k$ an integer as indicated in Table \eqref{results}.  While we were not able to fully determine $k$, we check the anomaly cancellation for $L_4^{11}$ and present some arguments that suggest that also the anomaly associated to $Q_4^{11}$ is cancelled.

\item The generators we have discussed so far, above the dashed line, are in the image of the natural map $\Omega_{11}^{\text{Spin-Mp}(2,\mathbb{Z})}\rightarrow \Omega_{11}^{\text{Spin-GL}^+(2,\mathbb{Z})}$. Equivalently, the duality bundles only involve duality transformations of determinant $+1$ (see Section \ref{sec:dualgroup}). This is not the case for the generators below the dashed line; they are manifolds with duality bundle involving $\GL^+(2,\mathbb{Z})$ reflections. As explained in Section \ref{sec:dualgroup}, a Spin-$D_{16}$ manifold has a $\text{Spin-GL}^+(2,\mathbb{Z})$ structure in a canonical way; the generators below the dashed line in \eq{results} are in fact all Spin-$D_{16}$ manifolds. Additionally, all the examples turn out to have $\Spin\text{-}{D_8}$ structures, which are defined analogously to $\Spin\text{-}{D_{16}}$ structures (see Section \ref{sec:bordanom}). There are two embeddings $i, \widetilde\imath\colon D_8\to D_{16}$, as illustrated in Figure \ref{square_in_octagon}, and therefore a $\Spin\text{-}{D_8}$ manifold $M$ has two associated $\Spin\text{-}{D_{16}}$ structures, which we denoted $M$ and $\widetilde M$, respectively, in the table above.

To describe these embeddings in more detail, consider $D_8$ as the group of symmetries of the square $[-1,
1]\times[-1,1]\subset\R^2$. Let $r$ be rotation by $\pi/2$ counterclockwise and $s$ be a reflection through the
$x$-axis; likewise, consider $D_{16}$ as the group of symmetries of a regular octagon in a plane, oriented such
that four of its sides are parallel to the $x$- and $y$-axes. Then $i\colon D_8\to D_{16}$ sends $r$ to a
counterclockwise rotation by $\pi/2$ and $s$ to reflection through the $x$-axis; $\widetilde\imath\colon D_8\to
D_{16}$ sends $r$ to the same rotation, but sends $s$ to the reflection through the line $y = x/2$, which meets two
vertices of the octagon. In Figure \ref{square_in_octagon}, $i$ corresponds to the embedding on the left and
$\widetilde\imath$ corresponds to the embedding on the right.

\begin{figure}[ht!]
\centering
\resizebox{0.4\textwidth}{!}{
	\begin{tikzpicture}
		\draw[thick] node[regular polygon,regular polygon sides=8,draw,minimum size=4.5cm] {};
		\draw[very thick, Blue] (-1.475, 1.475) -- (1.475, 1.475) -- (1.475, -1.475) -- (-1.475, -1.475)
		-- cycle;
		\fill[Blue] (-1.475, 1.475) circle (0.6mm);
		\fill[Blue] (-1.475, -1.475) circle (0.6mm);
		\fill[Blue] (1.475, -1.475) circle (0.6mm);
		\fill[Blue] (1.475, 1.475) circle (0.6mm);

		\draw[thick] node[regular polygon,regular polygon sides=8,draw,minimum size=4.5cm] at (6, 0) (p8){};
		\draw[very thick, Blue] (p8.corner 1) -- (p8.corner 3) -- (p8.corner 5) -- (p8.corner 7) -- cycle;
		\foreach \y in {1, 3, 5, 7} {
			\fill[Blue] (p8.corner \y) circle (0.6mm);
		}
	\end{tikzpicture}}
\caption{The two embeddings $i,\widetilde\imath\colon D_8\to D_{16}$ of the symmetries of a square
into the symmetries of an octagon.}
\label{square_in_octagon}
\end{figure}
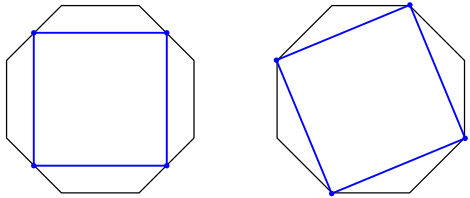

Much like a Spin\textsuperscript{$c$} manifold has an associated principal $U(1)$ bundle, a $\Spin\text{-}{D_8}$ manifold has an associated $D_4 = \Z_2\times\Z_2$ bundle. For the generators $X$ we list below, we will specify this principal $\Z_2\times\Z_2$ bundle, and any $\Spin\text{-}{D_8}$ structure on $X$ with this associated bundle can be chosen.

\item $\RP^n$ is the usual real projective space. The $\Z_2\times \Z_2$ bundle is non-trivial in its first factor only.

\item $X_{10}$ is the Milnor hypersurface that is one of the generators of $\Omega_{10}^{\Spin}= (\Z_2)^{\oplus 3}$
\cite{DMW02}. The Dirac and Rarita-Schwinger indices on this manifold vanish, and it is detected by an integral of Stiefel-Whitney classes of the tangent bundle, $\int w_4 w_6$. The $\Z_2\times \Z_2$ bundle over $X_{10}$ is trivial, and the $\Z_2\times \Z_2$ bundle over $S^1$ is specified by demanding that going around the circle picks up an action of the first $\Z_2$. The two embeddings of this $\Z_2$ into $D_8$ correspond to $S^1$ and $\widetilde{S^1}$, as above.

\item The ``Arcanum XI'' manifold $X_{11}$ and its $\Spin\text{-}{D_8}$ structure do not seem to have appeared in
the literature before, and are described briefly in Section \ref{sec:x11} and also in Appendix \ref{app:A}. Here,
we will only comment that $X_{11}$ is constructed as a quotient of $S^6\times S^5$ by a $\Z_2\times \Z_2$ action, and that it can also be understood as a non-trivial fibration of $\RP^5$ over $\RP^6$. $X_{11}$ is the only class whose anomaly we were unable to compute (or even connect to another closely related computation, as in the case of $Q_{4}^{11}$); because $2 [X^{11}] = 0$ in the bordism group, its anomaly is listed as ``$0$ or $\frac12$'' in \eq{results}. We will also explain the difficulties in computing the Arcanum anomaly in Section \ref{sec:x11}.
\end{itemize}

There are special points in the upper half plane which are invariant under finite subgroups of the duality symmetry group $\Spin\text{-}\GL^+(2, \Z)$. These special points are $\tau=i$, where a $\Z_8$ is restored, and $\tau=e^{2\pi i/3}$, where a $\Z_6$ is restored. The $D_8$ generated by (the Pin$^+$ lift of) reflections leaves any purely imaginary $\tau$ invariant.\footnote{See e.g. reference \cite{Dierigl:2020wen}
for further discussion on the physical significance of this special locus in the context of 4d QFTs with interfaces at strong coupling.}
 As explained in Section \ref{sec:dualgroup}, the full duality group is an amalgam of the finite groups restored at these special points, which suggests that the full duality anomaly can be reconstructed from the anomalies in these finite generating subgroups. The results in \eq{results} show that this is indeed the case: for the generators above the dashed line, the axio-dilaton is constant and fixed to one of the special values  $\tau=e^{2\pi i/3}$ (first two cases) or $\tau=i$ (last two), and the transition functions of the duality bundle are contained within $\Z_6$ and $\Z_8$, respectively. Below the dashed line, the axio-dilaton can take any imaginary value, but the duality bundle is contained within the $D_8$ subgroup which includes reflections.

\subsection{Physical interpretation and anomaly cancellation}
\label{subsec:cancel_anomaly}

The results in \eq{results} indicate that there is no consistent way to define pure type IIB supergravity on an arbitrary manifold. Importantly, none of the generators of the bordism group where the anomaly is non-vanishing is a mapping torus, i.e.\ an 11-manifold of the form $(X \times [0,1])/\sim$, where $\sim$ acts as a diffeomorphism on $X$ and identifies the two ends of the interval. Indeed, mapping tori are precisely the type of manifolds associated to anomalies in the traditional sense of the word \cite{Witten:2015aba}. The anomalies we have uncovered are more subtle, being of the ``Dai-Freed type'' which signify an inconsistency of the theory if one allows for topology changes, as explained in Section \ref{sec:bordanom}. This is a natural thing to do in the context of a theory of gravity, so as discussed in Section \ref{sec:bordanom} we take the point of view that the anomalies \eq{results} signal a pathology of the theory and must be cancelled.

An analogous situation occurs for $\mathcal{N}=1$ supersymmetric theories in 10 dimensions. There, the fermion content of the theory leads to a perturbative anomaly in a gauge symmetry, which is cancelled by the contribution of a topological term via what is known as the Green-Schwarz mechanism \cite{Green:2012pqa}. As we will now see, the anomalies in \eq{results}  above the dashed line can be cancelled by a similar mechanism. Specifically, we will show that a subtle modification of the Chern-Simons term of type IIB supergravity is enough to cancel the anomalies. As discussed above, we check anomaly cancellation in the $\mathbb{Z}_8$ factor for a $\mathbb{Z}_2$ subgroup generated by $L_4^{11}$ explicitly, and present an argument why we believe that this cancellation extends to $Q_4^{11}$. On classes below the dashed line the anomaly cannot be cancelled by the mechanism we are about to describe. However, all of the possible anomalies vanish, except possibly for those associated to $X_{11}$, which we were not able to compute.

To see this, recall that in Subsection \ref{ssec:IIBan} we explained that a background connection can be coupled to the self-dual field $C_4$, and that this is connected with mixed anomalies involving $C_4$, $C_2$, and $B_2$. The proper formalism, developed in \cite{Hsieh:2020jpj}, represents the background connection $\breve{c}$ as an element in differential cohomology, and the anomaly theory when taking into account this background connection contains a term of the form
\begin{equation}
\mathcal{A} \supset -\tilde{\mathcal{Q}}(\breve{c}) \,,
\label{am567}
\end{equation}
in terms of the quadratic refinement involved in the construction of the anomaly of the self-dual field. For the
particular case of Spin manifolds, a canonical choice of quadratic refinement exists \cite{HS05, Hsieh:2020jpj}, and in Appendix \ref{app:A} we determine it from anomaly cancellation in lens spaces without duality bundle. Using these techniques, we can construct the following table, involving the four Spin-Mp$(2,\Z)$ entries in \eq{results}\footnote{An additional line involving manifolds of the form $\text{Lens}\times\text{Bott}$ would show vanishing anomalies, if we take the particular Bott manifold constructed in \cite[\S 5.3]{FH21}.} as well as $L_4^{11}$ (which is bordant to four copies of $Q_{4}^{11}$):

\begin{equation}\begin{array}{c|c|c|c|c|c|c|c}
\text{Class}&\mathcal{A}&\text{Arf}&\tilde{\mathcal{Q}}&\beta(a)^2\cup a &\frac{(p_1)_3}{2}\cup a&\beta(b)^2\cup b&\frac12\left[(p_1)_4-\mathcal{P}(w)\right]\cup b\\\hline
L_3^{11}&1/3&1/4&n^2/3&1&0&0&0\\
L_3^{3}\times\mathbb{HP}^2&1/3&1/4&n^2/3&0&1&0&0\\
L_4^{11}&1/2&3/8&3n^2/8&0&0&0&2\\
L_4^{3}\times\mathbb{HP}^2&1/2&3/8&3n^2/8&0&0&0&2\end{array}\label{anoccc}
\end{equation}

The first column lists representatives of the relevant bordism classes, and the second column lists their anomalies. The third column gives the Arf invariant associated to the corresponding space, determined from the requirement that anomalies without a duality bundle turned on should cancel; see Appendix \ref{app:Anomcomp}. Note that since $Q_4^{11}$ is not Spin, one cannot turn off the duality bundle and derive the Arf invariant in this way. The numbers in the last four entries of the table denote entries in torsion cohomology in terms of the characteristic classes of the tangent bundle of the spacetime manifold as well as the duality bundle discussed in Section \ref{subsec:dualback}. To evaluate these for the given manifolds we use that the Pontryagin class of the lens space $L_n^{2k-1}$ is given by \cite{bams/1183527946,10.2307/1994874}
\begin{equation}
p=(1+x^2)^k \,,
\label{pclasslens}
\end{equation}
where $x$ is a generator of $H^2(L_n^{2k-1},\Z)=\Z_n$, together with the Pontryagin classes of $\mathbb{HP}^2$ (discussed in Appendix \ref{app:Anomcomp}, or also e.g.\ in \cite{10.2307/2373152,FH21}). These classes should be understood as integers modulo $n$ for $\Z_n$. This means that the first two classes are elements of $\mathbb{Z}_3$ and the latter two are classes in $\mathbb{Z}_4$.  The fourth column lists the only quadratic refinement compatible with the given Arf invariant, which can be determined using \eqref{eq:arfcalc}. Finally, in the last entry, $\mathcal{P}$ is the Pontryagin square operation \cite{gamkrelidze2019selected}.

To compute the entries in the last two columns of the table, it is important to take into account that $L_4^{11}$ and $L_4^3$ are Spin manifolds, and so, per the general considerations of Section \ref{sec:bordanom}, the principal $\mathbb{Z}_4$ bundle associated to the $\text{Spin-}{\mathbb{Z}_8}$ structure has a class which is an even number of times the generator of $H^1$ with $\mathbb{Z}_4$ coefficients.

Notice the peculiarity that all the anomalies that we can explicitly determine can be cancelled by the term associated to the quadratic refinement for a particular choice of $\breve{c}$, where we further set the background fields $B_2$ and $C_2$ to zero. A bit of guesswork reveals
that the following combinations can cancel the anomalies:
\begin{equation}
\breve{c}_0 = Y_5 = \left(\lambda_1\beta(a)^2+\lambda_2\frac{(p_1)_3}{2}\right)\cup a + \frac{\lambda_3}{2}\left[(p_1)_4-\mathcal{P}(w)\right]\cup b + \kappa \, \beta(b)^2 \cup b \,,
\label{Y5def}
\end{equation}
where $\lambda_i$ are signs (they can only be $\pm 1$), which we leave undetermined for now. Here, $(p_1)_j$ denotes the reduction of the first Pontryagin class modulo $j$. The coefficient $\kappa$ cannot be determined from the anomalies in \eqref{anoccc}, but potentially contributes for $Q_4^{11}$ as we will argue below.

This opens up a very simple and elegant possibility to cancel all anomalies. Suppose that the full action of IIB string theory contains a term of the form
\begin{equation}
(C_4,Y_5)\approx \int F_5 \cup\left[  \left(\lambda_1\beta(a)^2+\lambda_2\frac{(p_1)_3}{2}\right)\cup a+  \frac{\lambda_3}{2}\left[(p_1)_4-\mathcal{P}(w)\right]\cup b + \kappa \, \beta(b)^2 \cup b \right],
\label{fstr}
\end{equation}
where the second term should be understood as taking values in cohomology with $U(1)$ coefficients. Since
$\lambda_3 [(p_1)_4 - \mathcal P(w)]$ lives in mod $4$ cohomology, dividing it by $2$ is not automatic, and indeed
there are closed, oriented manifolds for which $\lambda_3 [(p_1)_4 - \mathcal P(w)]$ is not equal to twice another
mod $4$ class.\footnote{This is a combination of a calculation of $\mathcal P(w_2)$ in~\cite{Wu54, Wu59} and a
description of $H^4(B\mathrm{SO}; \Z_4)$ in~\cite[Theorem 1]{CV95}.} On Spin-$\Mp(2, \Z)$ manifolds, though, we can
divide, as proven in Appendix \ref{app:divi}.

Having the term~\eqref{fstr} in the action ensures that $\breve{c}$ is turned on according to Table \eq{anoccc}, so that anomalies cancel in all accessible cases. The topological coupling we discuss here mixes discrete and continuous fields in type IIB supergravity, realizing a version of the Green-Schwarz mechanism different from both the ordinary one in \cite{Green:2012pqa} and the topological Green-Schwarz mechanism described in \cite{Garcia-Etxebarria:2017crf}, whose possible application we discuss in Section \ref{subsec:topoGS}.

We will now explain why \eq{Y5def} is the right kind of class to couple to the chiral 4-form. In particular, it is possible to interpret it as a differential cohomology class with coefficients twisted by the determinant representation of $\text{GL}^+(2,\Z)$.  That the twisting must appear was shown in Section \ref{sec:dualgroup}, since reflections map $a\rightarrow-a$, $b\rightarrow-b$. The characteristic classes on the right hand side of \eq{Y5def} are elements in $H^5(X,\tilde{\Z}_3)$ and $H^5(X,\tilde{\Z}_4)$, respectively. Both of these can be regarded as elements of $H^5(X,\widetilde{U(1)})$ via the natural maps instead. As explained in \cite{Hsieh:2020jpj}, elements of this group precisely define a flat twisted differential  cohomology class. Thus, the coupling above is well-defined.

The value of the coefficients $\lambda_1$ and $\lambda_3$ can be fixed by an indirect, heuristic argument demanding agreement with the charge quantization properties of S-fold backgrounds \cite{Garcia-Etxebarria:2015wns,Aharony:2016kai}, which we now present. It would be desirable to check this argument rigorously; this would involve dimensional reduction of the anomaly theory of the self-dual field with fixed boundary conditions.

Consider a stack of $N$ D3-branes in flat space. The type IIB duality symmetry acts as a global symmetry of the worldvolume theory. As studied in \cite{HTY19}, this symmetry is anomalous, and the anomaly can be related to the quantization properties of the charges in generalized S-fold backgrounds. These are type IIB backgrounds of the form $S^5/\Z_n$, for $n=2,3,4,6$, which also involve a non-trivial duality bundle on the lens space.

As shown in \cite{TY19,Hsieh:2020jpj}, a stack of $N$ D3-branes moving in a closed loop around this space picks up a total phase given by
\begin{equation} \int_{S^5/\Z_n} F_5 + \mathcal{A}\in\Z,\label{543}\end{equation}
where $\mathcal{A}$ is the duality anomaly of the Maxwell theory in the S-fold background on the worldvolume theory of the brane stack. The value of this anomaly was computed in \cite{HTY19} for several cases, it is fractional, and matches the values of $\int F_5$ determined by an M-theory computation.

Thus, consistency forces the background $S^5/\Z_n$ to have fractional $\int F_5$ flux threading it. While the
argument involving Dirac quantization is solid, in principle the quantization condition on $F_5$ should come
directly from an analysis of the 10-dimensional theory. We will now see how the term \eq{fstr} may give such a
fractional contribution, and in doing so, we will fix the value of $\lambda_1$ and $\lambda_3$. Consider the
near-horizon geometry of the $N$ D3-branes. This is described by the familiar $\mathrm{AdS}_5\times S^5$ geometry. In the holographic context, we expect anomalies of the boundary field theory to arise as topological sectors in the action of the bulk dual geometry \cite{Aharony:1999rz}. In particular, this means that  dimensional reduction of 10d type IIB supergravity should produce a 5d action including a term of the form
\begin{equation}
S_{\text{5d}} \supset \int_{\mathrm{AdS}_5} \mathcal{A} \,,
\end{equation}
involving the duality bundle. A naive reduction of the 10d term \eq{fstr} on a sphere with $N$ units of five-form flux produces
\begin{equation}
S_{\text{5d}} \supset N \int_{\mathrm{AdS}_5}   \left[\left(\lambda_1\beta(a)^2+\lambda_2\frac{(p_1)_3}{2}\right)\cup a+  \frac{\lambda_3}{2}\left[(p_1)_4-\mathcal{P}(w)\right]\cup b + \kappa \, \beta(b)^2 \cup b\right] \,,
\label{3ed}
\end{equation}
which has the right form to match the duality anomaly of the boundary theory. We will see this is the case for the $\Z_3$ and $\Z_4$ part, separately.

For $\Z_3$, consider the 5d theory on $S^5/\Z_3$. The duality anomaly has a contribution of $- \tfrac29$ coming
from the gauge fields, and a second contribution of $- \tfrac19$ per complex fermion. Since there are four
of these, we get a total anomaly of $- \tfrac69 \sim + \tfrac13$.  Since according to formula \eq{pclasslens} the
Pontryagin classes of $S^5/\Z_3$ are trivial, the $\lambda_2$ term does not contribute and we get agreement if we
set $\lambda_1=+1$. This matching is not entirely trivial; had the anomaly not been a multiple of 3, it would not have been possible to capture it with a term of the form \eq{3ed}. As we will see in the next Section, this is also related to the fact that the $\mathbb{Z}_3$ part of anomaly theory \eq{eq:3anomalt} is in fact a bosonic tQFT, even though it is the anomaly of a theory which includes fermions.

For the $\Z_4$ S-folds, things do not work out so simply. The Pontryagin class is $p_1 = 3 x^2$, where $x$ is the generator of $H^2(S^5/\Z_4,\Z)=\Z_4$. Since the space $S^5/\Z_4$ is not Spin, $p_1$ is not divisible by 2; the term $\mathcal{P}(w)=\beta(b)^2$ solves this, with the result that the class $\tfrac12[(p_1)_4-\mathcal{P}(w)]= x^2$. Adding the possible contribution from the $\kappa \beta(b)^2\cup b$ term in \eq{Y5def}, this becomes $(\lambda_3+\kappa) \, x^2\cup b$. This is a modulo 4 class, yet the correct S-fold charge is $3/8$. The discrepancy, which is just a factor of two, is possibly related to the fact that we did not perform the dimensional reduction of the theory of the self-dual field properly, but rather treated it as if it was an ordinary bilinear pairing. This seems to be fine at prime three, but not at prime 2. It would be very interesting to solve this problem and provide a more rigorous consistency check of our proposal.

So far we have (partially) reproduced the duality anomaly of the D3-brane worldvolume theory in the case where no background fields have been turned on. The $U(N)$ worldvolume theory (including the center of mass degrees of freedom) also has electric and magnetic $U(1)$ 1-form symmetries for which 2-form backgrounds can be turned on. This is what happens in some S-fold backgrounds, and we should account for their charges too. In the following we will do this, again heuristically. Analogous to the procedure above, we would expect that the difference of $F_5$ charges between these two backgrounds should be given by a term of the form
\begin{equation}
N \int B_2 \cup F_3 = \frac{N}{2} \int (B_2 F_3 - C_2 H_3) = \frac{N}{2} \int (B_2,C_2) \,,
\label{342}
\end{equation}
where in the last equality we have replaced the product in cohomology by $\tfrac12$ of the $\SL(2,\Z)$ duality invariant differential cohomology pairing constructed in \cite{Hsieh:2020jpj}. This division by two does not make sense in general, but it can be replaced by its quadratic refinement
\begin{equation}
\frac{N}{2} \int (B_2,C_2) \,\rightarrow \, N \mathcal{Q}(\breve{c}_2) \,,
\label{3423}
\end{equation}
where $\breve{c}_2$ is the differential cohomology class encoding the discrete flux on the S-fold background. As was proven in \cite{HTY19,Hsieh:2020jpj}, this term correctly accounts for the charge differences between the different S-folds, provided the appropriate quadratic refinement is chosen in each example.

Thus, we have reproduced, at least heuristically, part of  the duality and 1-form anomalies of the worldvolume theory of D3-branes, providing some support that the term \eq{fstr}, which also cancels the ten-dimensional anomalies, is indeed correct. Nevertheless, a full consistency check involving additional constraints that can further fix the coefficients $\lambda_2$ and $\kappa$ would be desirable.

 Another check we have not performed comes from the anomaly on $Q_4^{11}$, which we did not evaluate due to technical difficulties. If the mechanism we propose is correct, the anomaly on $Q^{11}$ should cancel as well once the contribution of the duality bundle is taken into account. We will now explain the entry for $Q_4^{11}$ in \eq{results}, justifying why the anomaly should be of the form $k/4$, for $k$ an integer. The actual anomaly theory of IIB, including the $\mathcal{Q}(\breve{c})$ piece, should be an invertible tQFT, and hence a $\text{Spin-}\GL^+(2,\Z)$ bordism invariant. As we have just seen, it evaluates to zero on the lens space $L_4^{11}$ with nontrivial duality bundle, which is in turn bordant to four copies of $Q_4^{11}$. That means that the anomaly theory satisfies
 \begin{equation}\mathcal{A}(Q_4^{11})-\mathcal{Q}(\breve{c}_{Q_4^{11}})=\frac{k'}{4},\end{equation}
 for some integer $k'$. On the other hand, as reviewed in Appendix \ref{app:Q11}, any possible quadratic refinement of the torsion pairing on $H^6 (Q_4^{11}, \mathbb{Z}) = \mathbb{Z}_4 \oplus \mathbb{Z}_4$ evaluates to $k''/4$ for some integer $k''$. Taking $k=k'-k''$ gives the advertised result. It would be extremely interesting to check this explicitly.

The anomaly cancellation mechanism we outlined here can also be presented in a different light. As explained in Section \ref{ssec:IIBan}, we have assumed the existence of a canonical quadratic refinement for a given $\Spin\text{-}\GL^+(2, \Z)$ structure. For most of the generators of the bordism group, we were able to determine this canonical quadratic refinement indirectly, by demanding that anomalies without a duality bundle are cancelled. Our results can be recast as the statement that anomalies are cancelled by an appropriate choice of quadratic refinement, even with a duality bundle. Then, the term \eq{fstr} encodes the shift in the quadratic refinement that takes place when the duality bundle is turned on. It may very well be that, once an appropriate top-down general prescription for the quadratic refinement is found, it will reproduce the effects of \eq{fstr}. Doing this would be extremely interesting on its own, and would also provide yet another non-trivial check of our proposal.

We end this section with a final comment. We have seen how the duality symmetry $\mathcal{G}_{IIB}$ of type IIB has an
anomaly that can be cancelled, rendering the theory consistent. Of course, we expected this, since the fact that
this symmetry is gauged is the essential ingredient in F-theory. In fact, the duality bundle is interpreted in
terms of the Spin cover of the group of large diffeomorphisms of the elliptic fiber. One could imagine that this
trivializes the duality anomaly, i.e.\ perhaps an elliptic fibration cannot produce the most general
$\Spin\text{-}\GL^+(2, \Z)$ bundle, and so F-theory would automatically be anomaly free. This is not the case, as
all the bundles and generators in our table \eq{results} can be explicitly shown to correspond to an elliptic
fibration in this way.


\section{Alternative type IIB theories}
\label{sec:altIIB}

In the previous section we described a simple way to cancel the duality anomalies by modifying the topological couplings of the theory. Importantly, this did not require the introduction of new fields in the supergravity theory. However, there are other possibilities to cancel the anomaly using a topological version of the Green-Schwarz mechanism \cite{Garcia-Etxebarria:2017crf}. If these have a UV completion they would lead to alternative versions of type IIB string theories. Another natural question is whether anomaly considerations of the type above can rule out certain alternative duality groups of type IIB, associated to certain subgroups of $\SL(2,\mathbb{Z})$. We will discuss these two classes of different realizations of type IIB theories in the following.

\subsection{Cancellation via a topological Green-Schwarz mechanism}
\label{subsec:topoGS}

Since the duality anomaly is discrete, in some cases it can be cancelled via the introduction of anomalous \emph{gapped} degrees of freedom. Following \cite{Garcia-Etxebarria:2017crf}, where a similar mechanism was introduced to match anomalies of 8d theories, we will call this the \emph{topological Green-Schwarz mechanism}. In this scenario, the correct low-energy limit of IIB string theory would not only contain the ordinary type IIB supergravity, but it would also include a (non-invertible) tQFT denoted $\Xi$, with no local degrees of freedom, and the property that its partition function on a 10-manifold $M$ which is the boundary of an 11-dimensional manifold $X$ is given by
\begin{equation}
Z_{\Xi}(M) = e^{- 2 \pi i \mathcal{A}(X)} \,.
\end{equation}
Together with the explicit expressions for the anomalous phase of the path integral of IIB supergravity given in Table \eqref{anoccc}, the combination
\begin{equation}
Z_{\text{IIB}}(M)\,Z_{\Xi}(M) \,,
\end{equation}
has vanishing anomaly. Notice that the quadratic refinement of the self-dual field (the term indicated in \eq{am567}) is not turned on in this alternate anomaly cancellation mechanism.

Another way of saying that $\Xi$ is a tQFT is that the anomaly theory $\mathcal{A}$ admits a (symmetry-preserving) gapped boundary condition. The topological Green-Schwarz mechanism only works if this is indeed the case.
However, not every tQFT admits a gapped boundary condition, and we are presently lacking a classification concerning that property (see, however, \cite{Cordova:2019bsd,Cordova:2019jqi} for partial results). The situation is better for \emph{bosonic} tQFT's, those which can be completely written in terms of cohomology classes.

We now demonstrate that at least the $\mathbb{Z}_3$ part of the duality anomaly of type IIB supergravity is of this form, and then comment on the $\mathbb{Z}_4$ part of the anomaly. The $\mathbb{Z}_3$ part of the anomaly can be recast using the characteristic classes of the duality as well as the tangent bundle in the following way
\begin{align}
\mathcal{A}_{\mathbb{Z}_3} = \int [\beta(a)^4 + (p_2)_3] \cup \beta (a) \cup a] \,.
\label{eq:3anomalt}
\end{align}
Here, $(p_2)_j$ denotes the mod $j$ reduction of the second Pontryagin class. The fact that this is possible guarantees that the anomaly can be cancelled by a bosonic tQFT, which is special for the following reason. Since perturbative anomalies cancel, $\mathcal{A}_{\mathbb{Z}_3}$ is a bordism invariant, and so in a class like the generator of $\Z_{27}$ one would have expected it to take a generic value mod 27. A value like, e.g., $\tfrac{1}{27}$, can be captured by a fermionic tQFT, but not by a bosonic theory. The exact value that $\mathcal{A}$ takes, 9 mod 27, is \textit{precisely} such that it can be represented by a cohomology class.

Using the fact that the Pontryagin class of the lens spaces is given by \eqref{pclasslens}, one can check that the right hand side of \eq{eq:3anomalt} equals the value of $\mathcal{A}$ in the first two bordism classes above the dashed line in \eq{results}, which we will demonstrate explicitly for $L_3^{11}$. The duality bundle is specified by an element in $H^1 (L_3^{11}, \mathbb{Z}_3) = \mathbb{Z}_3$ under pullback from the classifying space $B\mathbb{Z}_3$, which can be identified with the characteristic class $a$ of the bundle. Similarly, we can pullback the class in $H^{11} (B\mathbb{Z}_3, \mathbb{Z}_3)$ which, using the cohomology ring \eqref{eq:Z3cohom}, can be identified with $\beta(a)^5 \cup a$. This means that if we integrate the class $\beta (a)^5 \cup a$, or more precisely its pullback over the lens space with the given duality bundle, we obtain the class 1 mod 3 in $H^{11} (L_3^{11}, \mathbb{Z}_3)$. From \eqref{pclasslens} we further know that the second Pontryagin class of $L_3^{11}$ is given by $15 x^4 = 0 \text{ mod } 3$, and the second term in \eqref{eq:3anomalt} vanishes in this background. Therefore, we find
\begin{align}
\mathcal{A}_{\mathbb{Z}_3} (L_3^{11}) = \int_{L_3^{11}} \beta(a)^5 \cup a = \tfrac13 \,,
\end{align}
which reproduces the desired result \eqref{results}. Similarly, $\mathcal{A}_{\mathbb{Z}_3}$ reproduces the anomaly for $L_3^3 \times \mathbb{HP}^2$.

Bosonic tQFTs always admit families of gapped boundary conditions which can be explicitly constructed, as described in detail in \cite{Garcia-Etxebarria:2017crf}. These provide alternatives to the anomaly cancellation mechanism outlined in the previous section, which involved a modification to the background coupling of an existing massless field (the self-dual form). These alternative anomaly free versions are completely fine from the point of view of the low-energy supergravity. However, the equations of motion of the new fields often impose constraints for the allowed field backgrounds of the type IIB supergravity fields, which forbid configurations we know are present in ordinary type IIB string theory.  Thus, these variant possibilities cannot correspond to the low-energy limit of type IIB string theory. This opens up two possibilities:
\begin{itemize}
	\item{These alternative versions of type IIB supergravity are in the Swampland \cite{Vafa:2005ui}. That is, they cannot be UV completed to consistent theories of quantum gravity. It would be nice to understand exactly why.}
	\item{The alternative theories are consistent, and correspond to different UV completions 	of the ordinary IIB supergravity. Although these possibilities are indistinguishable at the level of massless fields, the completeness principle  \cite{Polchinski:2003bq,Banks:2010zn,Harlow:2018tng} demands that they differ in their spectrum of extended objects, as we will discuss below. In this case, the cobordism conjecture \cite{McNamara:2019rup} would require that these different UV completions are connected to ordinary IIB string theory via dynamical domain walls.}
\end{itemize}
It would be highly interesting to investigate which of the two possibilities is actually realized.

We will now explain how to cancel the $\mathbb{Z}_3$ part of the discrete anomalies in Section \ref{sec:gen} using the topological Green-Schwarz mechanism, and we start by summarizing the discussion in \cite{Garcia-Etxebarria:2017crf}. Consider a higher-form version of the topological $BF$ theory, where the explicit degrees of freedom are a $(d-p)$-form field $B$ and a $(p-1)$-form field $A$, with field strength $F$. The Lagrangian of the theory is given explicitly by\footnote{Strictly speaking, this should be done at the level of differential cohomology, replacing the pairing $B\wedge F$ by its differential-cohomology version $B*A$ described in \cite{Hsieh:2020jpj}.}
\begin{equation}
Z_{\Xi}(M) = e^{-S_{BF}} \,,\quad S_{BF} = 6 \pi i  \int_M B \wedge F \,.
\label{act5}
\end{equation}
This topological theory can be shown to be equivalent in a different duality frame to pure $(p-2)$-form $\Z_3$ gauge theory \cite{Banks:2010zn}. The theory has manifest discrete $(p-1)$-form electric and $(d-p)$-form magnetic symmetries, which we can couple to background fields $X_{p}$ and $Y_{d-p+1}$, modifying \eq{act5} to
\begin{equation}
S_{BF}= 2\pi i \int_M \left( 3 \, B \wedge F + B\cup X_p + A\cup Y_{d-p+1} \right) \,.
\label{act6}
\end{equation}
This theory has a mixed anomaly between the electric and magnetic higher-form symmetries, which is captured by the anomaly theory $\mathcal{A}_{BF}$ in $(d+1)$ dimensions that can be written in terms of the background fields $X_p$ and $Y_{d-p+1}$,
\begin{align}
\mathcal{A}_{BF} = \tfrac{1}{3} (-1)^{d-p+1} \int_X X_p \cup Y_{d-p+1} \,.
\label{brer}
\end{align}
It takes values in $\mathbb{Z}_3$. It follows that if a given anomaly theory can be written as the product of two $\Z_3$ cohomology classes the corresponding anomaly can be cancelled by coupling to the $BF$ theory discussed above, thus realizing the topological Green-Schwarz mechanism. Due to the expression \eq{eq:3anomalt}, at least the $\mathbb{Z}_3$ part of the duality anomaly theory of IIB supergravity is of such a form that it can be cancelled by the inclusion of additional discrete $\mathbb{Z}_3$ fields in type IIB string theory. Note that they are not part of the massless fields in the low-energy supergravity description. Further, because the anomalies they take part in are at strong coupling, it is difficult to imagine they can be detected in perturbative string theory as well.

 $BF$ theories like the one above also contain non-trivial extended operators which are the higher-dimensional generalization of ``Wilson and 't Hooft'' lines. These are given by
\begin{equation}
\exp\Big( 2\pi i \int_{\Sigma_{p-1}} A\Big) \,, \quad \exp\Big( 2\pi i \int_{\Sigma_{d-p}} B\Big) \,.
\label{topo}
\end{equation}
The completeness principle \cite{Polchinski:2003bq,Banks:2010zn,Harlow:2018tng} requires the existence of extended objects on which these operators can end  (see \cite{Rudelius:2020orz,Casini:2020rgj} for recent extensions involving non-invertible symmetries). These objects have a $(p-1)$- and $(d-p)$-dimensional worldvolume, respectively, and are the electrically and magnetically charged objects of the higher-form symmetries discussed above. The fact that they are mutually non-local objects, i.e., they have a non-trivial Dirac pairing, has to be remedied by the introduction of localized degrees of freedom on the brane worldvolumes  \cite{TY19,Hsieh:2020jpj}. Moreover, the associated anomaly theories are determined in terms of the background fields of the higher-form symmetries and read
\begin{equation}
\mathcal{A}_{A} = \int X_p \,, \quad \mathcal{A}_{B} = \int Y_{d-p+1} \,.
\end{equation}
 These objects are very similar to Nielsen-Olesen strings and the
$\Z_n$ charged particles in ordinary four-dimensional $BF$ theory \cite{Nielsen:1973cs,Banks:2010zn,Heidenreich:2021tna}. Note also that the topological operators \eqref{topo} coupling to these objects are torsional. This implies that a sufficient number of the charged states, three in the examples discussed above, can decay to the vacuum since they do not carry charge.

In our particular case, the inclusion of a term \eqref{act5} to cancel the discrete duality anomalies in connection with the completeness hypothesis demands the existence of branes with localized degrees of freedom charged under the duality group. Clearly, the specific realization of objects depends on the choice of the Green-Schwarz mechanism. We now briefly discuss different possibilities that cancel the $\mathbb{Z}_3$ part of the duality anomaly as well as the correlated spectrum of extended objects:\footnote{Note that this is a far from exhaustive list of possibilities.}
\begin{itemize}
\item $X_3=\beta(a)\cup a$, $Y_8=\beta(a)^4+(p_2)_3$. The electrically charged extended objects are strings, and their magnetic duals are 7-branes.
\item $X_4=\beta(a)^2$, $Y_7=\beta(a)^3\cup a$, together with $X_4'=(p_1)_3$ and  $Y'_7=(p_1)_3 \cup \beta(a) \cup a$ (this requires two copies of the topological GS mechanism). The electric objects are 2-branes, and their magnetic duals are 6-branes.
\end{itemize}

Finally, the topological Green-Schwarz mechanism also places restrictions on the allowed backgrounds. This can be seen from the equations of motion of the theory \eq{act5} which require that $X_{p}=Y_{d-p+1}=0$. Integrating out these fields we hence end up with the original theory with additional constraints imposed as restrictions on the duality background. These are enough to show that type IIB supergravity with the topological Green-Schwarz mechanism \eq{act5} does not uplift to type IIB string theory. For instance the lens space $S^3/\mathbb{Z}_3$ with $\Z_3$ duality bundle is obstructed, but in type IIB string theory, this is just the near-horizon geometry of a $\mathbb{C}^2/\mathbb{Z}_3$ singularity. This shows that the first example in the list above is incompatible with the known string dualities; similar caveats apply to other examples of the form \eq{act5}.

So far, we have focused on the cancellation of the $\mathbb{Z}_3$ part of the anomaly above, since for this case we have a simple presentation of the anomaly theory in terms of cohomology classes of the duality bundle. Our understanding of the $\mathbb{Z}_4$ part of the anomaly is at the moment more limited, since in this case, the evaluation of the 11d anomaly theory on manifolds such as $Q_{4}^{11}$ suggests that a fermionic tQFT would be necessary to realize a topological Green-Schwarz mechanism. The main complication here is that as far as we know, a general treatment about the conditions necessary to arrange for gapped boundary conditions in fermionic tQFTs are less well understood when compared with their bosonic counterparts. If such a boundary condition is available, then we can repeat a quite similar analysis to what we presented in the $\mathbb{Z}_3$ case. That being said, one could also imagine a ``hybrid scenario'' in which the $\mathbb{Z}_3$ part is cancelled by the topological Green-Schwarz mechanism, while the $\mathbb{Z}_4$ part is cancelled by the additional Chern-Simons-like coupling we proposed for type IIB supergravity.\footnote{If one could establish that a gapped boundary condition for the $\mathbb{Z}_4$ part exists, then one could also contemplate a hybrid situation where the $\mathbb{Z}_3$ part is cancelled by our Chern-Simons-like coupling, while the $\mathbb{Z}_4$ part is cancelled by a topological Green-Schwarz mechanism.}

To summarize, we argued that there are several distinct versions of the topological Green-Schwarz mechanism that cancel at least part of the duality anomaly found in Section \ref{sec:gen}. These correspond to seemingly consistent supergravity theories but, as far as we can tell, are incompatible with F-/M-theory duality and so cannot be the low energy limit of the type IIB string theory we know.
As emphasized in Section \ref{ssec:IIBan}, we assume the existence of a canonical quadratic refinement for each $\text{Spin-}\GL^+(2,\mathbb{Z})$ manifold, depending solely on its $\text{Spin-}\GL^+(2,\mathbb{Z})$ structure, but we do not have a way to construct it.  It may very well be that a top-down prescription for the quadratic refinement in general $\text{Spin-}\GL^+(2,\mathbb{Z})$ manifolds is incompatible with the variant IIB theories outlined before, explaining why they are inconsistent. For instance, it may be that the term \eq{fstr} is included as part of the prescription for the ``correct'' quadratic refinement when a duality bundle is turned on. We will further discuss possible fates for these variant IIB theories in Section \ref{sec:concl}.

\subsection{Comment on congruence subgroups and the Swampland}

An anomaly can signal an inconsistency of the theory, and so it can be a great tool to place effective field theories in the Swampland. This is particularly true of theories with high supersymmetry, where anomalies and Swampland considerations can become extremely powerful \cite{Lee:2019skh,Kim:2019ths,Katz:2020ewz,Cvetic:2020kuw,Montero:2020icj,Hamada:2021bbz,Tarazi:2021duw}. Here, we briefly report on what we could say about type IIB and its duality symmetry by thinking along these lines.

When one says that type IIB string theory has a
$\Spin\text{-}\GL^+(2, \Z)$ symmetry, what is meant is that the spectrum of non-perturbative D-brane states is
compatible with such a symmetry. There is ample evidence coming from dualities that this is the right answer. But
who is to say that there is not a different version of type IIB string theory, with the same low-energy
supergravity limit, and a different duality group? For instance, one could replace the role of $\SL(2,\Z)$ by a
such as $\Gamma_0(n)$ even though there is evidence that these theories are in the
Swampland.\footnote{Replacing $\SL(2, \Z)$ with
$\Gamma_0(n)$ for $n$ sufficiently large leads to non-contractible cycles in the corresponding moduli space;
see also the discussion in \cite{Dierigl:2020lai}.}
One may wonder if anomalies and and their cancellation might put constraints or
even rule out these possibilities. Unfortunately, the answer is no. Suppose we replace $\GL^+(2,\Z)$ by a finite
index subgroup $\mathcal{D}$. Then, any $\Spin\text{-}\mathcal{D}$ manifold is also a $\Spin\text{-}\GL^+(2, \Z)$
manifold in a natural way. Therefore we cannot gain any new information about the anomaly by examining
Spin-$\mathcal D$ manifolds. The anomaly cancellation mechanisms discussed above still work, provided the characteristic classes that define the duality bundle survive. Thus, it does not seem possible to constrain the duality group via these anomalies, at least in a straightforward way.


\section{The manifold \texorpdfstring{$X_{11}$}{X11}}
\label{sec:x11}

The purpose of this section is to briefly discuss $X_{11}$, ``Arcanum XI,'' and summarize some of its properties.\footnote{Following \cite{ArcanumXI}, Arcanum XI, or Arcanum 11 is the tarot card of Force, Strength, or Fortitude, but is also often referred to as the ``Enchantress''. Quoting from \cite{ArcanumXI}:

\begin{centering}
``According to the Egyptian tradition, this describes the transmuting of energies though [sic] directed thinking and induced emotions.  What does this mean?  Sometimes we don’t have control over our circumstances, but we do have control over how we decide to respond to our circumstances.''
\end{centering}

But to some, there is also another side to Arcanum XI. Again quoting from \cite{ArcanumXI}:

\begin{centering}
``Unfortunately, there is another side to this card, too.  The Enchantress has a nickname among a small group of people who belong to the Church of Light, which is the 10\% solution card.  The reason for this nickname is that, more often than not in a reading, this card does not represent spiritual struggle but represents a con artist, or someone that it pulling the wool over the eyes of the querent; for example a beguiling lover who is not truthful or faithful, or an agent who is not honest.''
\end{centering}

Let us also note that the arch nemesis of the Swamp Thing (and thus indirectly the entire Swampland program)
is Anton Arcane, the evil necromancer genius (first appearance in Swamp Thing \#1 (November 1972),
first full appearance in Swamp Thing \#2 (January 1973)).} We provide the corresponding proofs in \cref{app:A,app:Anomcomp}.

Most of the manifolds appearing in our list of generators are familiar objects, such as lens spaces and projective spaces; the kinds of manifolds one would expect to find as generators of a Spin bordism group.\footnote{For example, compare with the lists of generators in~\cite{gilkey1989geometry, Garcia-Etxebarria:2018ajm, Hsi18, FH21}.} Though a Spin-$\GL^+(2, \Z)$ manifold need not be Spin, these familiar generators with the exception of $Q_{4}^{11}$ admit Spin structures, making the calculation of the anomaly theory for the self-dual field easier, partly because there is a canonical choice of quadratic refinement.

However, it is not possible to represent all elements of $\Omega_{11}^{\Spin\text{-}\GL^+(2, \Z)}$ below the dashed line in Table \eqref{results} with Spin manifolds: as a consequence of the Adams spectral sequence computation we undertake in~\cite{Upcoming}, the bordism invariants $\int w_2^4 x^3$ and $\int w_2^4 y^3$, thought of as homomorphisms $\Omega_{11}^{\Spin\text{-}\GL^+(2, \Z)}\to\Z_2$, are linearly independent, and in particular non-zero (which also implies that the manifold is not Spin, since $w_2\neq0$). Therefore, any set of generating manifolds for the bordism group below the dashed line must contain at least two non-Spin manifolds. Moreover, many of the standard examples one might use to try to realize these last two generators do not work, and the generators we found -- the same underlying manifold, called $X_{11}$, equipped with two different Spin-$\GL^+(2, \Z)$ structures -- are unwieldy to work with; in particular, we were not able to calculate the anomaly theory for the self-dual field on these generators, though we can constrain it.

In the following we will briefly outline the construction of $X_{11}$. Label the coordinates in $\R^{13}$ by $(x_1,\dotsc,x_7, y_1,\dotsc,y_6)$, and consider $S^6\times S^5\subset\R^{13}$ as the locus $\lvert{\vec x\rvert} = \lvert{\vec y\rvert} = 1$. The Klein four-group $\Z_2\times\Z_2 = \{1,\alpha,\beta, \alpha\beta\}$ acts freely on $S^6\times S^5$ as follows:
\begin{equation}
\begin{aligned}
	\alpha(x_1,\dotsc,x_7, y_1,\dotsc,y_6) &= (-x_1,\dotsc,-x_7, -y_1,y_2, y_3\dotsc,y_6)\\
	\beta(x_1,\dotsc,x_7, y_1,\dotsc,y_6) &= (x_1,\dotsc,x_7, -y_1,\dotsc, -y_6).
\end{aligned}
\end{equation}
$X_{11}$ is defined to be the quotient
\begin{equation}X_{11}=\frac{S^6\times S^5}{\Z_2\times\Z_2}.\end{equation}
 This is a closed, oriented $11$-manifold. Projecting onto the $x_i$ coordinates defines a map $X_{11}\to\RP^6$; this is a fiber bundle with fiber $\RP^5$.

The quotient map $S^6\times S^5\to X_{11}$ is a principal $\Z_2\times\Z_2$ bundle. We use this to define two
different principal $\GL(2, \Z)$ bundles on $X_{11}$, via two different embeddings $i, \widetilde\imath\colon
\Z_2\times\Z_2\to \GL(2, \Z)$.\footnote{These are the two embeddings $D_4\to D_8$ analogous to the ones depicted in Figure \ref{square_in_octagon} followed by the standard inclusion $D_8\hookrightarrow \GL(2, \Z)$.} In
\ref{X11_w2}, we show that these two duality bundles give rise to Spin-$\GL^+(2, \Z)$ structures, which we denote
$X_{11}$ and $\widetilde X_{11}$; \ref{X11_w2} also shows that $X_{11}$ is not Spin.

For our purposes in this paper, we need to know the following facts about $X_{11}$:
\begin{itemize}
	\item For $X_{11}$, $\int w_2^4x^3 = 1$ and $\int w_2^4y^3 = 0$; for $\widetilde X_{11}$, $\int w_2^4x^3 = \int
	w_2^4y^3 = 1$ (see \ref{mod_2_coh_invariants_X11}). This means that we can use $X_{11}$ and $\widetilde X_{11}$ as the
	two remaining non-Spin manifolds in our generating set.
	\item For both $X_{11}$ and $\widetilde X_{11}$, the $\eta$-invariants associated to the Dirac, Rarita-Schwinger,
	and signature operators vanish (\ref{eta_X11_vanish}). Therefore, $\mathcal A(X_{11})$ and $\mathcal
	A(\widetilde X_{11})$ depend solely on the quadratic refinement associated to the self-dual field.
	\item The torsion subgroup of $H^6(X_{11}; L)$, where $L$ is the local system that $C_4$ transforms in, is
	isomorphic to $\Z_2\oplus\Z_2$ (\ref{twisted_calc}). Therefore, (\ref{twisted_implies}), the (exponentiated)
	Arf invariant of any quadratic refinement of the torsion pairing must be either $\pm 1$ or $\pm i$. Moreover, the fact that $X_{11}$
	generates a $\mathbb{Z}_2$ factor suggests that the Arf invariant is constrained to be $\pm 1$. The same is
	true for $\widetilde X_{11}$.
\end{itemize}
Unlike in most of the previous examples,  we were not able to completely determine the quadratic refinement on $X_{11}$ or $\widetilde X_{11}$.  In the end, this is the reason why we cannot determine its anomaly. In fact, it is
not clear that there is a canonical quadratic refinement: as discussed briefly in Section \ref{ssec:IIBan}, the constructions given in \cite{HS05} and \cite{Hsieh:2020jpj} do not apply (at least in a straightforward way) to Spin-$\GL^+(2, \Z)$ manifolds
for which $w_2\ne 0$. The point of view that we took elsewhere in this paper (that a canonical quadratic refinement exists, even if we do not know to construct it generally) is not helpful now, since we cannot switch off the duality bundle on $X_{11}$. Showing that a canonical quadratic refinement exists in general Spin-$\GL^+(2, \Z)$ manifolds and providing a prescription to compute it are the main obstacles in evaluating the anomaly for Arcanum XI.


\section{Conclusions and future directions}
\label{sec:concl}

Type IIB string theory has an exact $\GL^+(2,\Z)$ symmetry, which is supposed to be gauged. This is not only demanded by general Swampland considerations about the absence of global symmetries \cite{Vafa:2005ui,Banks:2010zn,Harlow:2018tng,vanBeest:2021lhn}; it is also a fundamental ingredient in the web of dualities in string theory, and in particular is mapped to a geometric (and thereby manifestly gauged) symmetry under F-theory and its duality with M-theory.

Yet for this picture to be consistent, $\GL^+(2,\Z)$ must be free from anomalies. The goal of this paper was to figure out whether it actually is. Of course, in the end the answer turns out to be yes, but in a very interesting way. We determined the bordism group controlling the anomaly, the precise anomaly theory of type IIB supergravity (relying heavily on the results of the seminal works \cite{TY19,HTY19,Hsieh:2020jpj}), and evaluated it to find that the theory indeed has $\Z_4$- and $\Z_3$-valued anomalies. The anomalous backgrounds involve three and 11-dimensional versions of lens spaces, and are natural generalizations of the S-fold backgrounds discussed in the literature \cite{Garcia-Etxebarria:2015wns,Aharony:2016kai}.

Amazingly, the anomalies we found are exactly of the right kind to be cancelled by a subtle modification of the Chern-Simons coupling of type IIB supergravity. The anomaly cancellation mechanism we describe only works because the anomalies discussed above take very particular values (concretely, the values the anomaly theory takes are always the quadratic refinement of the torsion linking pairing of the manifolds where the anomaly is evaluated). This is particularly clear for the $\mathbb{Z}_{27}$ factor of the bordism group, for which only one possible value can be cancelled by this mechanism -- precisely the one realized by string theory. We checked that the same mechanism works for all the other $\text{Spin-}\Mp(2,\mathbb{Z})$ anomalies, with the caveat that we were able to explicitly check the cancellation only for a $\mathbb{Z}_2$ subgroup of the $\mathbb{Z}_8$ factor. Even in this case, we gave arguments why we expect this to extend to the full $\mathbb{Z}_8$. Furthermore, we provided some heuristic arguments that the modification that we construct is also related to known quantization conditions on the RR charges of S-fold backgrounds, which gives some support that the terms are indeed part of the ten-dimensional action. It would be very interesting to provide further checks that can help us elucidate whether the terms we propose are correct and explicitly evaluate the anomaly theory on $Q_4^{11}$. In particular, we should explore its ramifications under the duality web, and pinpoint a concrete counterpart in the M-theory dual. This is true more generally, as a comment on S-fold backgrounds. For instance, the $S^5/\Z_3$ background must have a fractional $F_5$ flux. In the dual M-theory perspective, where the background is $(S^5\times T^2)/\Z_3$, the $F_5$ flux becomes a quantization condition on $G_7$ flux. This is directly related to a diffeomorphism anomaly of M5-branes in this background, as was shown in \cite{Hsieh:2020jpj}.

One of the results of our investigations is particularly intriguing. We have constructed several families of type IIB supergravities in which the usual duality group $\text{GL}^+(2,\Z)$ can be gauged, but which at the same time cannot be the low-energy limit of IIB string theory, because some familiar orbifold backgrounds are forbidden. The symmetry type of these theories is different from ordinary type IIB supergravity, and so maybe one can show some of them are inconsistent, by studying compactifications in generalized backgrounds required by the cobordism conjecture as in \cite{McNamara:2019rup}. Yet another, exciting possibility is that in fact these theories are just fine quantum gravities, different UV completions of the same IIB supergravity theory, in a similar fashion to the possible global forms that a gauge group can have for a given algebra. There are hints of similar discrete choices in 11-dimensional M-theory \cite{FH21} and type I strings \cite{Sethi:2013hra}. Of course, we would expect all of these possibilities to be part of a single landscape, connected by domain walls \cite{McNamara:2019rup}.

Strictly speaking, our results involved extending the anomaly theory of the self-dual field in \cite{Hsieh:2020jpj} beyond its realm of validity. This anomaly theory depends on a crucial way on a certain quadratic refinement of a differential cohomology theory pairing, and  \cite{Hsieh:2020jpj}  constructed it only for Spin manifolds and the particular case of differential K theory. An essential assumption of the present paper is that the construction can be extended to general Spin-$\GL^+(2, \Z)$ manifolds in some way, which we left unspecified. Interestingly, this assumption together with anomaly cancellation seems to be enough to ``bootstrap'' the quadratic refinement in all cases discussed above, which correspond to  ``Spin-Mp'' anomalies, for which the
corresponding duality bundles only involve elements in the subgroup $\Mp(2,\Z)\subset \GL^+(2,\Z)$ (up to some subtleties with the $\mathbb{Z}_8$ generator $Q_{4}^{11}$). We have also explored all nine possible non-perturbative genuine $\GL^+(2,\Z)$ anomalies. For seven cases, we could show that the anomalies vanish identically. Interestingly, it is not possible to modify the topological term that cancels the $\Mp(2,\Z)$ anomalies so that it cancels genuine $\GL^+(2,\Z)$ anomalies as well. In this sense, it is good news for our proposed mechanism that these anomalies cancel by themselves.

There are just two bordism classes, both generated by the manifold we called $X_{11}$, Arcanum XI, for which we could not constrain the anomaly. The manifold $X_{11}$, which we constructed explicitly, is a non-Spin manifold with a $\Spin\text{-}{D_8}$ structure, where $D_8$ is the dihedral group with eight elements. The anomaly of type IIB supergravity in this background is at most a sign, but we could not determine which one. Doing so will likely involve an extension of the formalism for self-dual fields in \cite{Hsieh:2020jpj}, since as outlined above one needs a canonical  choice of a quadratic refinement using tools other than the Spin structure. Studying a general prescription for the quadratic refinement and computing this anomaly are natural directions for future research.

Another direction is the full determination of the duality bordism groups of type IIB supergravity in any number of dimensions. Via the cobordism conjecture \cite{McNamara:2019rup}, they naturally lead to a classification of generalized S-fold backgrounds and defects in type IIB string theory. Furthermore, the bordism groups should be extended to take into account backgrounds of RR fields. We will report on this in an upcoming publication \cite{Upcoming}.

All in all, our work raises more questions than answers -- on the mathematical framework underlying self-dual
fields, topological terms in type IIB supergravity, and the uniqueness of string theory. As a final question that we
attacked but could not solve with anomalies alone, we could not exclude that there are versions of type IIB
supergravity where the duality group is replaced by e.g.\ a congruence subgroup of $\SL(2,\Z)$. Together with the alternate anomaly free versions of type IIB involving topological Green-Schwarz terms we described above, we regard it as a very important question to determine whether they exist as new phases of string theory with the same low-energy dynamics as type IIB string theory, or they belong to the Swampland. We must find out whether there is a (discrete) Landscape of duality groups, or if, on the contrary, the $\GL^+(2,\Z)$ duality symmetry was truly meant IIB.

\subsection*{Acknowledgements}

We thank M. Del Zotto, D. Freed, I. Garc\'{i}a Etxebarria, J. McNamara and C. Vafa for helpful discussions.
We thank D. Freed and E. Torres for comments on an earlier draft.
We especially thank Y. Tachikawa and K. Yonekura for correspondence and comments on an earlier draft.
JJH and MM thank the 2021 Simons summer workshop at the Simons Center for Geometry and Physics
for kind hospitality during the completion of this work.
The work of JJH is supported in part by the DOE (HEP) Award DE-SC0013528.
MM was supported by a grant from the Simons Foundation (602883, CV).


\appendix


\section{Varying axio-dilaton and the \texorpdfstring{$\SL(2,\R)$}{SL2R} anomaly}
\label{app:SL2R}

This paper is not the first to study non-perturbative duality anomalies in type IIB string theory. References \cite{Gaberdiel:1998ui,Minasian:2016hoh} explore a duality anomaly which is only present when a singular axio-dilaton profile is turned on. The two anomalies are different, and cancelled by different couplings (though many of the techniques used to study both of them are related). We explain the precise relation in this Appendix. Since the anomalies in \cite{Gaberdiel:1998ui,Minasian:2016hoh} involve singular axio-dilaton profiles, they really only make sense in the context of F-theory, in contrast with the IIB anomalies discussed above. In hindsight, studying anomalies of F-theory backgrounds before checking that duality anomalies vanish as we just did above is a bit putting the cart before the horse; we only know that F-theory is consistent because $\Spin\text{-}\GL^+(2, \Z)$ anomalies cancel in the first place.

The duality group of type IIB supergravity that we discussed so far embeds in a perturbative $\Spin\text{-}\Mp(2,\mathbb{R})$ symmetry group, broken to the duality group $\Spin\text{-}\Mp(2,\Z)$ by the effect of massive states (the reflections in $\GL^+(2,\Z)$ remain unbroken, but will not play a role in the discussion). Since $\Spin\text{-}\Mp(2,\mathbb{R})$ is not a symmetry, type IIB supergravity does not provide a dynamical gauge field, and anomalies for this perturbative symmetry need not cancel.

Even though we have no dynamical connection, type IIB supergravity provides the next best thing: the low energy $\Spin\text{-}\Mp(2,\mathbb{R})$ is also spontaneously broken by the axio-dilaton $\tau$, down to a $\Spin^c$ subgroup. Because of this, the axio-dilaton target space $\mathcal{M}$ has a natural $U(1)$ bundle over it, and the pullback of any connection $A_{\mathcal{M}}$ via $\tau$ defines a $\Spin^c$ connection on spacetime (see also \cite{Seiberg:2018ntt}):
\begin{equation}
A_{\mathrm{Spacetime}}=\tau^* A_{\mathcal{M}} \,.
\end{equation}
The connection $A_{\mathrm{Spacetime}}$ encodes the correct holonomies of fermions (which follow from \eq{tfermions})
for a non-constant axio-dilaton background. So when computing the fermion path integral in a non-trivial spacetime
background, we need to couple the fermions to $A_{\mathrm{Spacetime}}$, even if there is no corresponding dynamical
field in the type IIB string theory. The field strength of $A_{\mathrm{Spacetime}}$ has a Chern-Weil representative given by \cite{Gaberdiel:1998ui,Minasian:2016hoh}
\begin{equation}
F= i\frac{d\tau\wedge d\bar{\tau}}{4\tau_2^2} \,.
\label{chwer}
\end{equation}
The crucial point in \cite{Gaberdiel:1998ui,Minasian:2016hoh} is that there is some ambiguity in this procedure, since it depends on the particular choice of connection $A_{\mathcal{M}}$. In physical terms, if we shift  $A_{\mathrm{Spacetime}}$ by a gauge transformation, we get another equally valid connection for the fermions, yielding the same holonomies. The actual phase of the partition function should be independent of this choice, which means that the corresponding $\Spin^c$ fermion anomalies must cancel. These are described by a twelve-dimensional anomaly polynomial, $P_{12}$. There is a purely gravitational contribution to $P_{12}$, but it is famously cancelled by the chiral 4-form \cite{Alvarez-Gaume:1983ihn}; so we will not consider it here.

Up to now, the discussion was the same whether we take $\Mp(2,\Z)$ to be gauged or not, but this matters in the end, since the allowed axio-dilaton backgrounds are very different. A supergravity person, who has no reason to gauge $\Mp(2,\Z)$, would take the axio-dilaton target space to be the upper half plane $\mathcal{H}$. Since $\mathcal{H}= \SL(2,\mathbb{R})/U(1)$, we have a natural fibration
\begin{equation}
U(1)\rightarrow \SL(2,\mathbb{R}) \rightarrow \mathcal{H} \,,
\label{fibby}
\end{equation}
and because $\mathcal{H}$ is contractible, the $U(1)$ bundle is trivial, so the Chern class of the connection
$A_{\mathrm{Spacetime}}$ obeys the constraint that $\int F=0$ over any two-cycle. Said differently, the connection behaves more like an
$\mathbb{R}$-connection. This constraint must be satisfied in the auxiliary space where we compute anomalies. The anomaly theory of the fermions, which is obtained from $P_{12}$ via descent, will then vanish identically. There is no anomaly to discuss, because if we do not gauge $\Mp(2,\Z)$, only non-singular axio-dilaton profiles are allowed, and the resulting connection admits a canonical trivialization.

On the other hand, after including the anomaly cancellation mechanism, the duality symmetry is gauged in IIB supergravity. In this case, the axio-dilaton moduli space is replaced by $\mathcal{H} / \SL(2,\mathbb{Z})$, the moduli space of complex structures of a torus, and the fibration \eq{fibby} descends to a $U(1)$ fibration over this space (it is just the canonical $U(1)$ bundle of any complex manifold). Strictly speaking, we should take the one-point compactification of this space, since the axio-dilaton can run all the way to $\tau=i\infty$ at the core of D7-branes. The one-point compactification is topologically a Riemann sphere, whose canonical $U(1)$ bundle is non-trivial.

The connection $A_{\mathcal{M}}$ now must have has non-trivial fieldstrengths, such that $\int F$ can be non-vanishing. All in all, we are simply saying that \eq{chwer} can have a $\delta$-function singularity at the core of a D7-brane, but not in non-singular axio-dilaton profiles. So now we can really demand anomaly cancellation in the standard sense. References \cite{Gaberdiel:1998ui,Minasian:2016hoh} found that one can indeed make anomalies cancel, but only after introducing a coupling
\begin{equation}
S''=\int  \ln g(\tau) \wedge W_{10} \,,
\label{gaberd}
\end{equation}
in the 10d type IIB action. Here, $W_{10}$ is a particular linear combination of characteristic classes which can be found in \cite{Gaberdiel:1998ui} and whose precise form will not matter to us, and $g(\tau)$ is an explicit function of the axio-dilaton, designed to exactly cancel the shift of the fermion path integral under gauge transformations. To achieve this, $g(\tau)$ must shift under infinitesimal $\Spin^c$ transformations as in  \eq{tfermions}, and more concretely, \eq{eq:fermtrafo} along the self-dual points. One possibility suggested in \cite{Gaberdiel:1998ui,Minasian:2016hoh} is
\begin{equation}
g(\tau) = \left(\frac{\eta(\tau) j^{1/12}(\bar{\tau})}{\bar{\eta}(\bar{\tau}) j^{1/12}({\tau})}\right) \,,
\label{seol}
\end{equation}
but $g(\tau)$ is intrinsically ambiguous, since it could be multiplied by any function of $\tau$ which does not transform under infinitesimal $\Spin^c$ gauge transformations. A final subtlety is that the coupling \eq{gaberd} is  anomalous under $\Mp(2,\Z)$ transformations; one could say that \eq{gaberd} has a mixed $\SL(2,\Z)\text{-U}(1)$ anomaly. The duality anomaly must cancel in all allowed backgrounds, since $\SL(2,\Z)$ was gauged to begin with; this puts a non-trivial constraint on F-theory backgrounds. This constraint can be demystified by expanding \eq{seol} in terms of the RR axion and dilaton, to find that \eq{gaberd} induces a coupling
\begin{equation}
S''\supset \int C_0 \wedge F \wedge \tilde{X}_8+\ldots \,,
\label{rrf}
\end{equation}
where $\tilde{X}_8$ only involves gravitational couplings. As discussed in \cite{Gaberdiel:1998ui,Minasian:2016hoh}, these couplings are related to tadpole cancellation in non-trivial ways.

To recapitulate: Type IIB supergravity fermions in non-trivial $\tau$ backgrounds must couple to a composite $\Spin^c$ connection, which should be non-anomalous. If the duality group is not gauged, anomalies cancel automatically. If it is gauged, anomaly cancellation requires the introduction of a counterterm like \eq{gaberd}, which forces constraints related to tadpole conditions on allowed F-theory backgrounds.

As we have emphasized, cancellation of the anomalies we studied in Section \ref{sec:gen} are a precondition to
discuss the ones in \cite{Gaberdiel:1998ui,Minasian:2016hoh}. The two anomalies are independent: A term like
\eq{gaberd} will not contribute to any anomaly that one can detect at constant axio-dilaton.


\section{Representations of the dihedral group}
\label{app:dihedralrep}

In this Appendix we briefly describe the irreducible representations of the dihedral group with $2n$ elements denoted by $D_{2n}$. We further argue that fermions on $\Spin\text{-}D_{2n}$ manifolds have to transform in two-dimensional representations of a certain kind.

The group structure and elements of $D_{2n}$ are defined as follows
\begin{align}
D_{2n} = \langle R, S | S^n = R^2 = 1 \,, \enspace RSR = S^{-1} = S^{n-1} \rangle \,.
\end{align}
Here, $R$ acts via reflections and $S$ via discrete rotations. Since all one-dimensional representations are commutative they are representations of the Abelianization of $D_{2n}$ given by
\begin{align}
\text{Ab} (D_{4k + 2}) = \Z_2 \,, \quad \text{Ab} (D_{4k}) = \Z_2 \oplus \Z_2 \,,
\end{align}
generated by $\{ R \}$ and $\{ R , S \}$, respectively. In the case of interest to us, the fermions of type IIB transform under $D_{16}$, i.e., $n$ even. Moreover, for a non-trivial $\Spin\text{-}D_{16}$ structure to be possible they have to transform non-trivially under $S^4$. However, for one-dimensional representations one has $S^4 = 1$ and they are not suitable for the fermions in the system.

The real two-dimensional representations are given by
\begin{align}
R = \begin{pmatrix} 0 & 1 \\ 1 & 0 \end{pmatrix} \,, \quad S = \begin{pmatrix} \text{cos} (2 \pi k / n) & - \text{sin} (2 \pi k / n) \\ \text{sin} (2 \pi k / n) & \text{cos} (2 \pi k/ n)\end{pmatrix} \,,
\end{align}
which clearly demonstrates their interpretation as reflection and rotation. Setting again $n = 8$ for the case relevant to us we see that fermions need to transform in representations with $k \in \{ 1, 3\} \text{ mod } 4$ on manifolds with a $\Spin \text{-} D_{16}$ structure. This is clearly the case for type IIB fermions as can be seen from \eqref{eq:fermtrafo}.


\section{More details on \texorpdfstring{$X_{11}$}{X11}}
\label{app:A}

Most of $\Omega_{11}^{\Spin\text{-}\GL^+(2, \Z)}$ can be generated by lens spaces (or lens space bundles) and their products with $\HP^2$ and $X_{10}$ -- the standard examples to try in this sort of bordism problem. But these manifolds only generate a subgroup of $\Omega_{11}^{\Spin\text{-}\GL^+(2, \Z)}$ whose complement is isomorphic to $\Z_2\oplus\Z_2$. In this Appendix, we discuss the two missing generators, defining them and showing that they represent the remaining $\Z_2\oplus\Z_2$ summand.

\subsection{Narrowing the search space}
\label{ssec:narrow}
Searching for missing generators of a bordism group can be quite difficult in general. We approached this problem
by applying techniques that narrowed the space we had to search over: calculations that imply that we can
choose generators which have a relatively simple form. Once we have enough information, it is easier to make a good
guess as to what the generators should be. In this subsection only, we assume some familiarity with Thom spectra
and the Steenrod algebra as discussed in~\cite{2018arXiv180107530B}.

Because we need to account for a $\Z_2\oplus\Z_2$ subgroup of $\Omega_{11}^{\Spin\text{-}\GL^+(2, \Z)}$, we can
approximate Spin-$\GL^+(2, \Z)$ bordism with a simpler notion of bordism which captures all $2$-torsion information
in bordism groups. This simpler notion is Spin-$D_{16}$ bordism, in which the duality bundle is associated to a
principal $D_8$-bundle via the standard embedding $D_8\hookrightarrow \GL(2, \Z)$.

The reason that we know to look at Spin-$D_{16}$ manifolds comes down to the way $\GL^+(2, \Z)$ factors as an
amalgamated product, as we saw in~\eqref{eq:GL+amalg}: this implies the map $D_{16}\hookrightarrow\GL^+(2,
\Z)$ is an isomorphism on $\Z_2$ cohomology. The Thom isomorphism implies that the induced map on Thom spectra for
Spin-$D_{16}$ and Spin-$\GL^+(2, \Z)$ bordism also induces a mod $2$ cohomology isomorphism, and by the stable
Whitehead theorem, the map $\Omega_*^{\Spin\text{-}D_{16}}\to\Omega_*^{\Spin\text{-}\GL^+(2, \Z)}$ is an
isomorphism on $2$-torsion subgroups.

Let $V\to BD_8$ be the associated vector bundle for the standard two-dimensional real representation of $D_8$ as
the symmetries of a square, and let $\mathit{MT}(\Spin\text{-}D_{16})$ denote the Thom spectrum for Spin-$D_{16}$
bordism: the homotopy groups of this spectrum are the Spin-$D_{16}$ bordism groups. In \cite{Upcoming}, we will
prove an equivalence of Thom spectra
\begin{equation}
	\mathit{MT}(\Spin\text{-}D_{16}) \simeq \mathit{MTSpin}\wedge (BD_8)^{V\oplus 3\mathrm{Det}(V) - 5}.
\end{equation}
For brevity, let $M_8= (BD_8)^{V\oplus 3\mathrm{Det}(V) - 5}$. The mod 2 cohomology of $M_8$ is easy to
understand -- it is a free $H^*(BD_8;\Z_2)$ module on a single generator $U$, which is in degree zero -- so we
would like to convert information in the cohomology of $M_8$ to information about Spin-$D_{16}$ bordism.

First suppose $B$ is a space and $M$ is a Thom spectrum over a space $B'$ such that some notion of ``Spin-$B$
bordism'' is computed as the Spin bordism of $M$. For example, we are interested in Spin-$D_{16}$ bordism, for
which $M = M_8$ and $B' = BD_8$. Then elements of $H^n(B';\Z_2)$ define characteristic classes for manifolds with a
Spin-$B$ structure, and the integrals of these classes are $\Z_2$-valued bordism invariants. Let $\Sq^2\colon
H^*(\text{--};\Z_2)\to H^{*+2}(\text{--};\Z_2)$ denote the second Steenrod square, a mod $2$ stable cohomology
operation.
\begin{thm}
\label{margolis}
Let $c\in H^n(B';\Z_2)$ and suppose the corresponding class $Uc\in\widetilde H^n(M;\Z_2)$ is such that
$\Sq^2\Sq^2\Sq^2(Uc)\ne 0$. Then there is an $n$-dimensional Spin-$B$ manifold $W$ with $\int_W c\ne 0$, and $W$
generates a $\Z_2$ summand of $\Omega_n^{\Spin\text{-}B}$.
\end{thm}
This is a combination of several theorems: Margolis' theorem~\cite{Mar74} relating free summands over the Steenrod
algebra $\mathcal A$ to spectrum-level splittings; \cite{Sto63, bams/1183527786} allowing us to change rings from
$\mathcal A$ to $\mathcal A(1)$; a result on using mod $2$ cohomology classes detect bordism classes (see~\cite[\S
8.4]{FH21}); and the fact that free, rank-one $\mathcal A(1)$-modules are characterized as the cyclic $\mathcal
A(1)$-modules on which $\Sq^2\Sq^2\Sq^2$ is non-zero~\cite[Lemma D.8]{FH16}. We will discuss \cref{margolis} in more
detail in \cite{Upcoming};
the point is that information about bordism (generally hard) can be deduced from computations in cohomology (usually easier).

Recall from Section \ref{sec:dualgroup} the cohomology classes $x$, $y$, and $w$ in the ring $H^*(BD_8;\Z_2)$.
In $\widetilde H^*(M_8;\Z_2)$, $\Sq^2\Sq^2\Sq^2(Uw^4x^3)$ and $\Sq^2\Sq^2\Sq^2(Uw^4y^3)$ are linearly independent, so by \cref{margolis} the characteristic classes $w^4x^3$ and $w^4y^3$ detect two $\Z_2$ summands in $\Omega_{11}^{\Spin\text{-}{D_{16}}}$. Furthermore, both of these characteristic classes vanish on the generators we built using lens spaces, $\mathbb{HP}^2$, and $X_{10}$. Therefore we need to find two Spin-${D_{16}}$ manifolds, one on which $w^4x^3\ne 0$ and $w^4y^3 = 0$, and the other on which $w^4x^3 = 0$ and $w^4y^3\ne 0$.

The last simplification we make replaces $D_{16}$ with a simpler subgroup. An oriented manifold $M$ admits a
Spin-${D_{16}}$ structure if it has a principal $D_8$ bundle $P\to M$ with $w(P) = w_2(M)$. If the principal $D_8$
bundle is induced from a principal $D_4 = (\Z_2\times \Z_2)$ bundle, we call a Spin-${D_{16}}$ structure on $M$ a
Spin-${D_8}$ structure. Because $D_4$ is Abelian, Spin-${D_8}$ structures are easier to construct and study than
Spin-${D_{16}}$ structures.

If $r,s$ denote the standard generators of $D_4$ and $r',s'$ denote the standard generators of $D_8$, consider the
two maps $i,\widetilde\imath\colon D_4\to D_8$ defined by $i(r) = \widetilde\imath(r) = r'$, $i(s) =
s'$, and $\widetilde\imath(s) = r's'$; see Figure \ref{square_in_octagon} for an analogous picture for $D_8 \to D_{16}$. There is an isomorphism $H^*(BD_4;\Z_2)\cong\Z_2[a, b]$ with $|a| = |b| = 1$ such that
\begin{itemize}
	\item $i^*(w^4x^3) = a^7b^4 + a^3b^8$ and $i^*(w^4y^3) = 0$,
	\item $\widetilde\imath^*(w^4x^3) = \widetilde\imath^*(w^4y^3) = a^7b^4 + a^3b^8$, and
	\item $i^*(w) = \widetilde\imath^*(w) = ab + b^2$.
\end{itemize}
We are thus led to look for an 11-dimensional manifold $X_{11}$ with a principal $D_4=\Z_2\times\Z_2$ bundle with characteristic classes $a,b\in H^{1}(X_{11},\Z_2)$ which satisfy $w_2(TX_{11})=ab+b^2$ and $a^7b^4 + a^3b^8\ne 0$ -- and one can show that in the Thom spectrum for Spin-${D_8}$ bordism,
\begin{equation}
	\Sq^2\Sq^2\Sq^2(U(a^7b^4 + a^3b^8))\ne 0,
\end{equation}
so we know such a Spin-${D_8}$ manifold actually exists. Now we go and find it.
\subsection{Finding the manifold}
\label{ssec:finding_X11}
The Klein group $\Z_2\times\Z_2$ does not have a free action on a sphere of any dimension \cite{10.2307/87918,free-groups}; the next best thing is to look at products of spheres. In particular, we have
\begin{equation}
X_{11}=\frac{S^6\times S^5}{\Z_2\times\Z_2} \,,
\end{equation}
where  $\Z_2\times \Z_2$ acts as follows. Take $\mathbb{R}^{13}$, and label coordinates there as $(x_1,\ldots x_7, y_1,\ldots y_6)$.  Then $S^6\times S^5$ is  the locus specified by the equations
\begin{equation}
\vert \vec{x}\vert^2=1,\quad \vert \vec{y}\vert^2=1 \, \subset \mathbb{R}^{13} \,.
\end{equation}
The $\Z_2\times \Z_2$ is given by generators $\alpha$ and $\beta$ given by
\begin{equation}
\label{Z2_act_X11}
	\alpha:\, (\vec{x}, \vec{y})\,\rightarrow\, (-\vec{x},-y_1,y_2,\ldots,y_6),\quad \beta:\, (\vec{x}, \vec{y})\,\rightarrow\, (\vec{x},-\vec{y}).
\end{equation}
The second action corresponds to antipodal mapping on the $S^5$, and that projection onto the first factor is well-defined up to a sign. This means that $X_{11}$ is a fiber bundle over $\RP^6$, with fiber $\RP^5$. It is also orientable, since both $\alpha$ and $\beta$ preserve the induced volume form.

We have $ H^{1}(X_{11},\Z_2)\cong\Z_2\oplus\Z_2$. The quotient map $S^6\times S^5\to X_{11}$ is a principal
$D_4$ bundle, and its characteristic classes $a,b$
generate $H^{1}(X_{11},\Z_2)$. We can construct associated vector bundles in the sign representation of each $\Z_2$, sections of these vector bundles are represented by real-valued functions on $S^6\times S^5$ that are even or odd under the actions of $\alpha$ or $\beta$, respectively.

\begin{prop}
\label{mod_2_coh_invariants_X11}
On $X_{11}$, $\int a^3b^8 = 1$ and $\int a^7b^4 = 0$.
\end{prop}
That is, $X_{11}$ indeed represents the bordism classes we are searching for.
\begin{proof}
We can compute these cup products dually, by passing to intersections in homology. Taking $k$ generic functions
$f_i$ on $S^6\times S^5$ which correspond to class $a$, and $\ell$ functions $g_i$ that correspond to class $b$,
with $k + \ell=11$, the number of isolated solutions of the equations
\begin{equation}
f_1=f_2=\ldots=g_1=g_2=\ldots=0 \,,
\end{equation}
divided by four, is the intersection pairing $\int a^k b^\ell$. We can obtain some such functions simply from the transformation properties under $\alpha$ and $\beta$ of the coordinate functions in the ambient $\mathbb{R}^{13}$, which are
\begin{equation}
\begin{array}{c|cccc}
\text{Action}&x_i&y_1&y_{i>1}\\\hline
\alpha&-&-&+\\
\beta&+&-&-\end{array}
\label{trre}
\end{equation}
From this, we see that $\int a^3 b^8$ corresponds to the number of solutions of the equations
\begin{equation}
x_1=x_2=x_3=y_1\cdot x_4=y_1\cdot x_5=y_1\cdot x_6=y_2=y_3=y_4=y_5=y_6=0\subset S^6\times S^5 \,.
\end{equation}
There are four such points, with $y_1=\pm1$, $x_7=\pm1$. So the pairing $\int a^3 b^8$ is 1. On the other hand, to compute $\int a^7 b^4$, we count solutions of
\begin{equation}
x_1=x_2=x_3=x_4=x_5=x_6=x_7=y_2=y_3=y_4=y_5=0\subset S^6\times S^5 \,,
\end{equation}
which has no solutions.
\end{proof}

All that is left is to compute $w_2(X_{11})$. This can be done via techniques similar to those used to compute the cohomology of $\RP^n$.
\begin{lem}
\label{X11_w2}
$w_2(X_{11}) = ab + b^2$, and therefore $X_{11}$ has a Spin-$D_8$ structure.
\end{lem}
\begin{proof}
The normal bundle of $S^6\times S^5\subset\mathbb{R}^{13}$ is trivial, since the unit normal vectors to the two spheres provide two everywhere non-vanishing sections, so $T(S^6\times S^5)\oplus\mathbb{R}^2= T\mathbb{R}^{13}$. After quotienting by $\Z_2\times \Z_2$, the sections of the normal bundle still survive, so the left hand side becomes the sum of $T X_{11}\oplus\mathbb{R}^2$ while the right hand side becomes a particular sum of line bundles over $X_{11}$ as indicated in table \eq{trre}. Taking the Stiefel-Whitney class of both bundles, we obtain the formula
\begin{equation}
w(TX_{11})=( 1+a)^7(1+b)^5(1+a+b) \,,
\end{equation}
and in particular $w_2=ab+b^2$, proving that $X_{11}$ has a $\Spin\text{-}{D_8}$ structure.
\end{proof}


\section{Computation of the anomaly}
\label{app:Anomcomp}

In this Appendix we compute the $\eta$-invariants involved in evaluating the anomaly theory \eq{ath} in the bordism classes in \eq{results}.

As described in Section \ref{sec:gen}, the generators of the bordism groups are products of Spin manifolds and lens
spaces, apart from a couple exotic cases. To evaluate $\mathcal{A}$ on product backgrounds, we will use repeatedly the formula for the $\eta$-invariant of a Dirac operator on a product space $A\times B$ with $A$ odd-dimensional and $B$ even-dimensional \cite{atiyah_patodi_singer_1976,gilkey_1988,gilkey1989geometry,FH21},
\begin{equation}
\eta^{\text{D}} (A\times B) = \eta^{\text{D}}(A) \times \text{Index}^{\text{D}}(B) \,,
\label{jejeje}
\end{equation}
and its generalization for the Rarita-Schwinger operator,
\begin{equation}
\eta^{\text{RS}} (A \times B) = \eta^{\text{RS}}(A) \times \text{Index}^{\text{D}}(B) + \eta^{\text{D}}(A) \times \text{Index}^{\text{RS}}(B) \,.
\label{jejeje2}
\end{equation}
These formulas just encode the physics of dimensional reduction; reducing the theory of a Dirac fermion on $B$
produces $ \text{Index}^{\text{D}}(B)$ chiral fermion zero modes, whose anomaly is then captured by
$\eta^{\text{D}}(A)$. Similarly, reduction of a vector-spinor $\psi_M$ on $B$ produces vector-spinor modes
$\psi_\mu\otimes \lambda$, where $\lambda$ is a Dirac zero mode, and Dirac modes $\psi\otimes \lambda_i$, where
$\lambda_i$ is a vector-spinor zero mode. As a cross-check of the Rarita-Schwinger formula, we note that a 10d
chiral vector-spinor reduced on $K3$ produces $-40$ modes\footnote{In reduction of the heterotic theory to 6d, these
$-40$ real modes are grouped into $-20$ complex ones, as the relevant $6$-dimensional spinors are not real. Adding
an opposite chirality 10d Dirac fermion to account for the constraint on the gravitino we obtain an index of $-42$,
and a total of $-21$ complex modes, the superpartners of the K3 moduli. In our applications in this Appendix, we
will always reduce on an $8$-manifold, and so the reality properties of the fermions remain the same.}, which is
the correct quantity \cite{degeratu_wendland_2010}.

The inclusion
\begin{equation}
\frac{\Mp(2,\Z) \times \Spin}{\Z_2} \,\rightarrow \, \frac{\GL^+(2,\Z)\times\Spin}{\Z_2} \,,
\end{equation}
induces a group homomorphism of the corresponding bordism groups,
\begin{equation}
\Omega_{11}^{\Spin \text{-} \Mp(2,\Z)}\,\rightarrow\, \Omega_{11}^{\Spin \text{-} \GL^+(2,\Z)} \,.
\label{eq:bordhom}
\end{equation}
The finite factors in $\Omega_{11}^{\Spin \text{-} \GL^+(2,\Z)}$ naturally fall into two classes, those in the image of the homomorphism in \eq{eq:bordhom} and those in the complement of the image. They are separated by the dashed line in \eq{results}. In physical terms, classes in the image from $\Omega_{11}^{\Spin \text{-} \Mp(2,\Z)}$ have transition functions involving only elements in the duality group of determinant $+1$; while the others involve reflections. We will consider both cases separately.

\subsection{Anomalies for Spin-Mp\texorpdfstring{$(2,\Z)$}{(2,Z)} manifolds}

As listed in \eq{results}, most of these bordism classes are generated by products of lens spaces with Spin manifolds. As
described in the main text, for each of these generators the $\Mp(2,\Z)$ bundle is an associated bundle for a
principal $\Z_3$ or $\Z_8$ bundle. Consequently, the fermions transforming in representations \eq{tfermions} behave as if they had $\Z_3$ or $\Z_8$ charge only, and the self-dual form does not see the duality bundle at all, as discussed in Section \ref{sec:dualgroup}. Thus, we need to know the $\eta$-invariants of fermions and the signature operator coupled to $\Z_n$ bundles on lens spaces. In the following we will focus only on the $\mathbb{Z}_2$ subgroup of the $\mathbb{Z}_8$ factor, which is generated by $L_4^{11}$, which is Spin and allows for an application of the techniques below.\footnote{The formula for $\eta$-invariants of Dirac fermions on spaces of the form $Q_4^{11}$ is provided in \cite{10.4310/jdg/1214459973}. However, at present we were not able the extend these to the other contributions to the anomaly, which would be needed to explicitly evaluate the anomaly theory in $Q_4^{11}$.}

Consider a lens space $L^{2k-1}_n\equiv S^{2k-1}/\Z_{n}$ with a principal $\Z_n$ bundle whose class
is given by that of a generator of
$H^1(L^{2k-1}_n; \Z_n)$ associated to $S^{2k-1} \rightarrow L^{2k-1}_n$. The $\eta$-invariant of a fermion of charge $q$ under this bundle is \cite{gilkey1989geometry,Hsieh:2020jpj}
\begin{equation}
\eta^{\text{D}}_{q} \big( L^{2k-1}_n \big) = - \frac{1}{n \, i^k} \sum_{j=1}^{n-1}\frac{e^{-2\pi i q j/n}}{\big(2\sin (\pi j/n)\big)^k} \,.
\label{eff}
\end{equation}
The $\eta$-invariant for the same fermion but a bundle whose class is $m$ times the generator of  $H^1(L^{2k-1}_n,\Z)$ is given by the above formula under the replacement $q \rightarrow m q$. For the applications in this paper, where $k$ is always even, $q$ is directly the $\mathbb{Z}_n$ charge of the fermion (see \cite{Hsieh:2020jpj} for the general case).
 Finally, when $n$ is even and $k$ is odd, the space $L^{2k-1}_n$ above is not Spin, but it admits a $\Spin \text{-} \Z_{n}$
 structure, hence a Spin-$\Mp(2, \Z)$ structure for $n \in \{2, 4 \}$.

We also need the $\eta$-invariant of the Rarita-Schwinger operator in a lens space $L_n^{2k-1}$, which for $n>2$
can be determined as follows \cite{Witten:2015aba,Hsieh:2020jpj}. In this case, $\eta^{\text{RS}}_q$ is simply the
$\eta$-invariant of a Dirac operator coupled to $TL^{2k-1}_n \otimes \mathcal{L}^{q}$, where $\mathcal{L}$ is the
complex line bundle associated to the $\Z_n$ bundle by the representation of $\Z_n$ on $\C$ by rotations.
The complexification of the tangent bundle  $TL^{2k-1}_n$ satisfies (see Lemma 1.1 of
\cite{EWING1977177})
\begin{equation}
\big( TL^{2k-1}_n \otimes \mathbb{C} \big) \oplus \mathbb{C} \approx k(\mathcal{L}\oplus \mathcal{L}^{-1}) \,.
\end{equation}
Tensoring with $\mathcal{L}^{q}$, we get
\begin{equation}
\big( TL^{2k-1}_n \otimes \mathcal{L}^{q} \big) \oplus \mathcal{L}^{q} \approx k (\mathcal{L}^{q+1} \oplus \mathcal{L}^{q-1}) \,,
\end{equation}
so that
\begin{equation}
\begin{split}
{\eta}^{\text{RS}}_q \big( L^{2k-1}_n \big) &= k \, \eta^{\text{D}}_{q+1} \big( L^{2k-1}_n \big) + k \, \eta^{\text{D}}_{q-1} \big( L^{2k-1}_n \big) - \eta^{\text{D}}_{q} \big( L^{2k-1}_n \big) \\
&= - \frac{1}{n \, i^k} \sum_{j = 1}^{n-1} \frac{e^{- 2 \pi i q j/n}}{\big( 2 \sin (\pi j / n) \big)^k} \big( 2 k \cos (2 \pi j/n) - 1 \big).
\end{split}
\label{gfor1}
\end{equation}
This is also what one gets from direct application of the equivariant index theorem \cite{Hsieh:2020jpj,10.2307/24892232}. We also record the $\eta$-invariant for the signature operator \cite{Hsieh:2020jpj}:
\begin{equation}
\tfrac18\eta_s^{\text{Sig}} \big(L^{2k-1}_n\big) = - \frac{1}{8 n \, i^{k}} \sum_{j =1}^{n-1} \frac{e^{- 2 \pi i j s/ n}}{\big( \tan (\pi j/n) \big)^k}.
\label{effsd}
\end{equation}

With these formulas, we can compute the anomaly theory directly for the 11-dimensional lens spaces, save for the
term involving the Arf invariant in \eq{ath}. To compute it, we need to determine a canonical quadratic refinement
for the torsion pairing
\begin{equation}
H^6 (X, \Z) \times H^6 (X, \Z) \rightarrow \mathbb{Q} / \Z \,,
\end{equation}
of the $11$-dimensional manifold \cite{Hsieh:2020jpj}. While there is a canonical quadratic refinement for
11-dimensional Spin manifolds \cite{Hsieh:2020jpj}, it is not easy to determine. In fact, it is easier to use
anomaly cancellation of type IIB supergravity to determine the Arf invariant. Take for instance the 11-dimensional
lens space $L_3^{11}$, but with trivial duality bundle; we will denote it as $\widetilde{L}_3^{11}$. This manifold is Spin, and since $\Omega_{11}^{\Spin}=0$, its total anomaly must cancel. One can compute, using \eq{eff}, \eq{gfor1}, and \eq{effsd} that
\begin{equation}
\eta^{\text{RS}}_{1}(\widetilde{L}_3^{11}) - 2 \eta^{\text{D}}_1(\widetilde{L}_3^{11}) - \eta^{\text{D}}_{-3}(\widetilde{L}_3^{11}) - \tfrac{1}{8} \eta_+^{\text{Sig}}(\widetilde{L}_3^{11}) = -\tfrac14 \,,
\label{ar5}
\end{equation}
which implies $\text{Arf}(\widetilde{L}_3^{11})= \tfrac{1}{4}$. The relevant cohomology group is $\Z_3$, and indeed, the two possible quadratic refinements for a binary form over $\Z_3$ have Arf invariants $\pm \tfrac14$, so this is consistent. One can similarly consider $L_4^{11}$ without duality bundle and conclude that the Arf invariant is $\tfrac38$, which is one of the possible values for a quadratic refinement of the torsion pairing of $H^6(L_4^{11}, \mathbb{Z}) = \Z_4$.

While these are nice consistency checks, we need a way to bypass an explicit evaluation of the Arf invariant if we are to find anomalies when duality bundles are introduced. We will do so by rewriting the anomaly of the self-dual field in terms of fictitious fermion fields. Note that aside from $Q_{4}^{11}$, all the generators above the dashed line in table \eq{results} are Spin manifolds, and $Q_{4}^{11}$ is a $\Spin\text{-}\mathbb{Z}_8$ manifold (and in fact $\Spin^c$). This can be checked by explicit computation of the Stiefel-Whitney classes, which for a lens space with $n>2$ are just the mod 2 reduction of the Chern classes of the complexified tangent bundle \cite{milnor1974characteristic}. It also follows directly from group theory arguments \cite{Garcia-Etxebarria:2015wns,Hsieh:2020jpj}. As a result, for the Spin manifolds we can subtract from \eq{ath} the anomaly theory of a fictitious gravitino, dilatino, and self-dual field, which transform trivially under the $\Mp(2,\Z)$ part of the duality group:
\begin{equation}
\mathcal{A}^{\text{aux}}= \eta^{\text{RS}}_{0} - 3 \eta^{\text{D}}_0 - \tfrac{1}{8}\eta_{+}^{\text{Sig}} +
\text{Arf} = 0 \,.
\label{we45}
\end{equation}
This is because, as we saw in~\eqref{IFT_class}, the group of anomaly theories is an extension of the group
$\Hom(\Omega_{12}^\Spin, \Z)$ (the perturbative anomalies) by the torsion subgroup of $\Hom(\Omega_{11}^\Spin,
\C^\times)$. Because the
type IIB anomaly polynomial vanishes, the image of this anomaly theory in $\Hom(\Omega_{12}^\Spin, \Z)$ is $0$, and
since $\Omega_{11}^\Spin = 0$, the anomaly~\eqref{we45} must vanish.
As a result, we can rewrite \emph{for a Spin manifold},
\begin{equation}
\tilde{\mathcal{A}} = \mathcal{A} - \mathcal{A}^{\text{aux}} = \tilde{\eta}^{\text{RS}}_{1} - 2 \, \tilde{\eta}^{\text{D}}_1 - \tilde{\eta}^{\text{D}}_3  = 0 \,,
\label{athr}
\end{equation}
where we have defined reduced $\eta$-invariants, denoted by $\tilde{\eta}$,
\begin{equation}
\tilde{\eta}^{\text{RS}}_{q} \equiv \eta^{\text{RS}}_{q} - \eta^{\text{RS}}_{0} \,, \quad \tilde{\eta}^{\text{D}}_{q} \equiv \eta^{\text{D}}_{q} - \eta^{\text{D}}_{0} \,.
\label{eq:reducedeta}
\end{equation}
Unlike their unreduced counterparts, these $\eta$-invariants are bordism invariants in the group $\Omega^{\Spin \text{-} \GL^+(2,\Z)}$. Another way of looking at this procedure is that we have used the fact that the self-dual field does not see the duality bundle to compute the relevant Arf invariant implicitly, as in \eq{ar5}, and use it to evaluate the anomaly when a nontrivial duality bundle is turned on.

For the cases involving products with an 8-manifold we can use the formula
\begin{equation}
\text{Index}^{\text{RS}} (X_8) = 24 \, \text{Index}^{\text{D}}(X_8) - \sigma (X_8) \,,
\end{equation}
where $\sigma$ is the signature of the 8-manifold (see e.g.~\cite{Alvarez-Gaume:1984zlq} for a discussion in the context of anomalies). Using this, we obtain (for the Bott manifold constructed in \cite{FH21})
\begin{equation}
\text{Index}^{\text{D}} (\text{Bott}) = 1 \,, \quad \sigma (\text{Bott}) = -224 \,, \quad \text{Index}^{\text{RS}}(\text{Bott}) = 248 \,,
\label{botindex}
\end{equation}
as well as
\begin{equation}
\text{Index}^{\text{D}} (\mathbb{HP}^2) = 0 \,, \quad \sigma (\mathbb{HP}^2) = 1 \,, \quad \text{Index}^{\text{RS}}(\mathbb{HP}^2) = -1.
\label{hp2index}
\end{equation}
We now compute the anomalies in the corresponding cases:

\begin{itemize}
\item The manifold $L_3^{11}$ is Spin, so we can use the modified anomaly theory \eq{athr}.  Doing this, we have that
\begin{equation}
\tilde{\mathcal{A}} \big(L_3^{11} \big) = \tilde{\eta}^{\text{RS}}_{1} \big( L_3^{11} \big) - 2 \, \tilde{\eta}^{\text{D}}_{1} \big(L_3^{11} \big),
\label{ditt}
\end{equation}
where the other terms vanish because both the dilatino and self-dual field transform trivially under the $\Z_3$ subgroup of the duality group.
We can now evaluate \eq{ditt} using \eq{gfor1} and \eq{eff}, obtaining
\begin{equation}
\tilde{\mathcal{A}} \big( L_3^{11} \big) = \tfrac{1}{3} \,.
\end{equation}
 \item The manifold $L_3^3$ can be treated in a similar manner as the one above, using the $\tilde{\eta}$-invariants. Using again \eq{ditt} and \eq{gfor1}, together with \eq{eff} \eq{botindex} and \eq{hp2index}, we get
\begin{equation}
\tilde{\mathcal{A}} \big( \text{Bott} \times L_3^3 \big)=0,\quad  \tilde{\mathcal{A}} \big( \mathbb{HP}^2 \times L_3^3 \big)= \tfrac{1}{3} \, \text{mod} \, 1 \,.
\end{equation}
Again, the anomaly does not cancel in $\mathbb{HP}^2$. The Arf invariant in this case is $\tfrac14$.

\item We now switch to manifolds involving an action of $\Z_8$. In general, these manifolds are not Spin, and only
have a $\Spin\text{-}{\Z_8}$ structure as is the case for $Q_4^{11}$; but the lens spaces we use as generators of a $\mathbb{Z}_2 \oplus \mathbb{Z}_2$ subgroup of the full $\mathbb{Z}_8 \oplus \mathbb{Z}_2$ bordism group are actually
Spin, and may be regarded as classes in $\Omega_{11}^{\Spin}(B\Z_4)$ for the purposes of computing the anomaly.
Related to this, the gravitino and dilatino charges under the $\Z_4$ principal bundle are $1$ and $-3$ with respect to the associated principal $\Z_4$ bundle, respectively, and the self-dual form is uncharged so it drops out of the reduced anomaly theory \eq{athr}. The $\Spin \text{-} {\Z_8}$ structure requires that the principal $\Z_4$ bundle has a classifying class given by twice that of the generator of $H^1(L_4^{11},\Z_4)=\Z_4$, in order for the fermion transition functions to be well-defined \cite{Hsieh:2020jpj}. This effectively multiplies the charges of every field above by two, and so, using again \eq{gfor1}, \eq{ditt} and \eq{eff}
\begin{equation}
\tilde{\mathcal{A}} \big( L_4^{11} \big) =  \tfrac{1}{2} \, \text{mod} \, 1 \,.
\end{equation}

\item The anomaly on $\text{Bott}\times L_4^{3}$ as well as $\mathbb{HP}^2\times L_4^{3}$ can be computed by now familiar techniques, obtaining
\begin{equation}
\tilde{\mathcal{A}} \big( \text{Bott} \times L_4^3 \big)=0 \, \text{mod} \, 1 \,,\quad\tilde{\mathcal{A}} \big( \mathbb{HP}^2 \times L_4^3 \big)= \tfrac12 \, \text{mod} \, 1 \,.
\end{equation}

\end{itemize}
This completes the calculation for $\Spin\text{-}\Mp(2,\Z)$ manifolds.

\subsection{Divisibility of \texorpdfstring{$(p_1)_4-\mathcal{P}(w)$}{p14-Pw} by two}
\label{app:divi}
In Section \ref{subsec:cancel_anomaly}, we introduced a class $\frac12[(p_1)_4-\mathcal{P}(w)]$ in the anomaly theory of IIB supergravity, which is necessary to cancel anomalies via the quadratic refinement. For this class to make sense, we need to show that $(p_1)_4-\mathcal{P}(w)$ is twice another cohomology class. We will now show that this is the case for $\text{Spin-}\Mp(2,\Z)$ manifolds.

Let $V\to X$ be a vector bundle with a Spin-$\Mp(2, \Z)$ structure. Then there is a canonical class $\mu(V)\in
H^4(X;\Z_4)$ such that
\begin{equation}
	2\mu(V) = p_1(V)\bmod 4 - \mathcal P(w) \,.
\end{equation}
From~\cite[Theorem 1]{CV95} we learn that $H^4(B\mathrm{SO}; \Z_4)\cong \Z_4\oplus\Z_2$, where $p_1\bmod 4$
generates the $\Z_4$ summand and the $\Z_2$ summand is generated by a class $\theta(w_4)$, which is the
image of $w_4$ under the multiplication-by-$2$ map $\theta\colon H^*(\text{--};\Z_2)\to H^*(\text{--}; \Z_4)$. For
any oriented vector bundle $V\to X$, we have~\cite{Wu54, Wu59} (see also~\cite[Theorem C]{Tho60} and
\cite[(6.1)]{AST13})
\begin{equation}
\label{wu_thm}
	p_1(V)\bmod 4 - \mathcal P(w_2(V)) = \theta(w_4(V)) \,.
\end{equation}
A Spin-$\Mp(2, \Z)$ vector bundle has an associated real rank-$2$ oriented vector bundle $E\to X$, and there is
an induced Spin structure on $V\oplus E$.\footnote{This is an analogue of the more familiar fact that if $V$ has a
Spin$^c$ structure with determinant line bundle $L$, $V\oplus L$ has a canonical Spin structure.}
Apply~\eqref{wu_thm} to $V\oplus E$; since this bundle is Spin, $w_2(V\oplus E) = 0$. Therefore
\begin{equation}
	p_1(V\oplus E) \bmod 4 = \theta(w_4(V\oplus E)) \,.
\end{equation}
This is the class we want to divide by $2$.

Since $w_1(E) = 0$, $w_2(E) = w_2(V)$, and $w_4(E) = 0$, $w_4(V\oplus E) = w_4(V) + w_2(V)^2$, and therefore
\begin{equation}
	\theta(w_4(V)) = p_1(V\oplus E)\bmod 4 + \theta(w_2(V)^2) \,.
\end{equation}
Since $V\oplus E$ is Spin, it has a characteristic class $\lambda(V\oplus E)\in H^4(X;\Z)$ with $2\lambda(V\oplus
E) = p_1(V\oplus E)$. And $\theta(w_2^2) = 2(p_1\bmod 4)\in H^4(B\mathrm{SO};\Z_4)$. Therefore we can choose
\begin{equation}
	\mu = (\lambda(V\oplus E) + p_1(V))\bmod 4 \,.
	\qedhere
\end{equation}
We note, however, that it does not seem possible to do this on Spin-$\GL^+(2, \Z)$ manifolds in general.

\subsection{The anomaly for \texorpdfstring{$Q_4^{11}$}{Q411}}
\label{app:Q11}

Though we cannot completely determine the value of the anomaly theory on the generator $Q_4^{11}$ of the
$\mathbb{Z}_8$ factor, we collect some useful properties of this manifold.

Recall from Section \ref{ssec:computation} that $Q_4^{11}$ is a $L_4^9$-bundle over $S^2$, specifically the ``lensification'' of the bundle
$V := \underline\C^4\oplus\mathscr O(2)\to\CP^1 = S^2$. That is, $m\in \Z_4$ acts on $V$ fiberwise by
\begin{equation}
	(z_1, \dotsc, z_5)\mapsto i^m\left(z_1, \dotsc, z_5\right).
\end{equation}
If  $S(V)\to S^2$ denotes the unit sphere bundle for $V$ with respect to a $\Z_4$-invariant Hermitian metric, then
$Q_4^{11}:= S(V)/\Z_4$.\footnote{The name ``lensification'' is by analogy with projectivization; if we took
the quotient by all of $\mathrm U(1)$ instead of just $\Z/4$, we would have obtained $\mathbb P(V)\to S^2$. The
condition that the Hermitian metric be $\Z_4$-invariant is no obstacle: if $h$ is any Hermitian metric on $V$,
average it over the $\Z_4$-action to obtain an invariant one.}

Since $V$ is a sum of line bundles over $\CP^1$, $S(V)\to S^2$ admits a section. Composing with the quotient map
$S(V)\to Q_4^{11}$, we obtain a section $\sigma\colon S^2\to Q_4^{11}$ of $\pi$.
\begin{prop}
\label{Q411_coh}
For $A = \Z$ and $A = \Z_4$, $H^*(Q_4^{11}; A)\cong H^*(L_4^9;A)\otimes H^*(S^2;A)$. That is,
\begin{enumerate}
	\item $H^*(Q_4^{11};\Z)\cong\Z[x, z, t]/(4x, x^5, xz, z^2, y^2)$, where $\abs x = 2$, $\abs z = 9$, and $\abs t
	= 2$; and
	\item $H^*(Q_4^{11};\Z_4) \cong\Z[b, \beta(b), \tilde t]/(b^2 = 2\beta(b), \beta(b)^5, \tilde t{}^2)$,
	where $\abs b = 1$, $\abs{\beta(b)} = 2$, and $\abs{\tilde t} = 2$.
\end{enumerate}
Reduction mod $4$ sends $x\mapsto\beta(b)$ and $t\mapsto\tilde t$. In the above $\vert \cdot \vert$ refers to the degree of a
class.
\end{prop}
\begin{proof}
Set up the Serre spectral sequence for the fiber bundle $L_4^9\to Q_4^{11}\to S^2$ with $\Z$ and $\Z_4$
coefficients. We draw this in Figure \ref{serre_SS_Q}. This uses $H^*(L_4^9;\Z)\cong\Z[x, z]/(4x, xz, z^2, x^5)$ with $x$
in degree $2$ and $z$ in degree $9$; the mod $4$ cohomology of $L_4^9$ is the same as that of $B\Z_4$ except with
the additional relation $\beta(b)^5 = 0$.
\begin{figure}[ht!]
\centering
\includegraphics{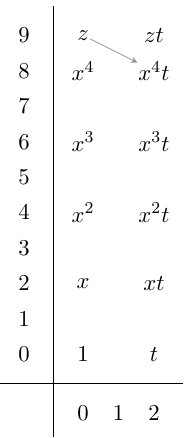}
\qquad\qquad
\includegraphics{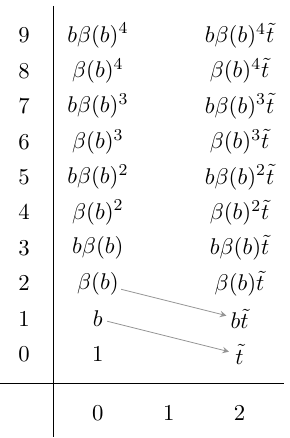}
\caption{The Serre spectral sequence for the cohomology of $Q_4^{11}$. Left: $\Z$ coefficients. Right: $\Z_4$
coefficients. We will show the three indicated differentials vanish, which implies both spectral sequences
collapse.}
\label{serre_SS_Q}
\end{figure}

In the spectral sequence for $\Z$ coefficients, all differentials save $d_2(z)$ vanish because either their domain
or codomain is $0$. Since $\pi_1(Q_4^{11})\cong\Z_4$, the Hurewicz theorem tells us $H_1(Q_4^{11};\Z)\cong\Z_4$.
Poincaré duality says $H^{10}(Q_4^{11};\Z)\cong\Z_4$ too, forcing $d_2(z) = 0$. Therefore the spectral sequence for
$\Z$ cohomology collapses, which means that as a $H^*(S^2;\Z)$-module, $H^*(Q_4^{11};\Z)\cong H^*(S^2;\Z)\otimes
H^*(L_4^9;\Z)$. This splits all extension questions: we know the extension question in $H^2$ splits, because
$\pi\colon Q_4^{11}\to S^2$ has a section. We know the $H^*(S^2;\Z)$-module structure on $H^*(Q_4^{11};\Z)$, which
uniquely constrains multiplication by $t$. A priori, it is possible that multiplication by $x$ sees an extension,
but if this were true, we would see it on the $E_\infty$-page, and we do not. This finishes the proof for $\Z$
cohomology.

The proof for $\Z_4$ cohomology is similar. By degree reasons, all differentials $d_r$ with $r > 2$ vanish, and
since this spectral sequence is multiplicative, $d_2$ is determined by its values on $b$ and $\beta(b)$. The
section of $\pi\colon Q_4^{11}\to S^2$ forces $d_2(b) = 0$. Since $\beta(b) = x\bmod 4$, $d_2(\beta(b)) =
d_2(x)\bmod 4 = 0$. Therefore this spectral sequence also collapses. The argument that there are no multiplicative
extensions is similar to the case with $\Z$ coefficients.

Finally, we knew $x\bmod 4 = \beta(b)$ because this is true for $L_4^9$, and we know $t\bmod 4 = \tilde t$ because
this is true for $S^2$.
\end{proof}
We also want to know the Arf invariant of the quadratic refinement of the torsion pairing.
Since $Q_4^{11}$ is not Spin, the tricks we used in the previous subsection do not work, and we are unable to
determine the quadratic refinement completely, though we are able to constrain it somewhat. Since $Q_4^{11}$ is
Spin-$\Mp(2, \Z)$, the torsion pairing is on untwisted cohomology. Recall from \cref{Q411_coh} that
$H^6(Q_4^{11};\Z)\cong\Z_4\oplus\Z_4$, generated by $x^3$ and $x^2t$.
\begin{prop}
The torsion pairing
\[\ang{\text{--}, \text{--}}\colon \mathrm{Tors}(H^6(Q_4^{11};\Z))\otimes
\mathrm{Tors}(H^6(Q_4^{11};\Z))\to\Q/\Z\]
has the following values: $\ang{x^3, x^3} = \ang{x^2t, x^2t} = 0$ and $\ang{x^3,
x^2t} = 1/4$.
\end{prop}
\begin{proof}
Though we defined the torsion pairing using differential cohomology in Section \ref{ssec:IIBan}, it can also be defined
in ordinary cohomology. Specifically, let $X$ be an oriented $11$-manifold and $\beta_{\Q/\Z}\colon H^*(X;
\Q/\Z)\to H^{*+1}(X;\Z)$ be the Bockstein for the short exact sequence
\begin{equation}
	\shortexact{\Z}{\Q}{\Q/\Z}.
\end{equation}
Let $a,b\in\mathrm{Tors}(H^6(X; \Z))$. Since $b$ is torsion, its image in $H^6(X;\Q)$ vanishes, and therefore there
is some class $\overline b\in H^5(X;\Q/\Z)$ such that $\beta_{\Q/\Z}(\overline b) = b$. Then $a\cup \overline b\in
H^{11}(X;\Q/\Z)$, and the torsion pairing of $a$ and $b$ is defined to be
\begin{equation}
	\ang{a, b}:= \int_X a\cup\overline b\in\Q/\Z \,.
\end{equation}
This value does not depend on the choice of $\overline b\in\beta_{\Q/\Z}^{-1}(b)$.

Because $4\mathrm{Tors}(H^*(Q_4^{11};\Z)) = 0$, we can do something easier. There is a commutative diagram of short
exact sequences
\begin{equation}
\label{coeff_change}
\begin{gathered}
\begin{tikzcd}
	0 & \Z & \Z & {\Z_4} & 0 \\
	0 & \Z & \Q & {\Q/\Z} & 0
	\arrow[from=1-1, to=1-2]
	\arrow[from=2-1, to=2-2]
	\arrow["{1\mapsto 4}", from=1-2, to=1-3]
	\arrow[from=2-2, to=2-3]
	\arrow[from=1-3, to=1-4]
	\arrow[from=2-3, to=2-4]
	\arrow[from=1-4, to=1-5]
	\arrow[from=2-4, to=2-5]
	\arrow["{1\mapsto 1/4}", from=1-4, to=2-4]
	\arrow["{1\mapsto 1/4}", from=1-3, to=2-3]
	\arrow[Rightarrow, no head, from=1-2, to=2-2]
\end{tikzcd}
\end{gathered}
\end{equation}
Everything in the above definition of the torsion pairing is natural in maps of short exact sequences, so if we
replace $\Q/\Z$ with $\Z_4$, we obtain the same value of the torsion pairing. And we can take advantage of the ring
structure on mod $4$ cohomology: if $x\in H^*(X;\Z)$ and $y\in H^*(X;\Z_4)$, then $x\cup y = (x\bmod 4)\cup y$.
This is compatible with the maps in~\eqref{coeff_change}.
\begin{lem}
\label{4bock_lemma}
Let $\beta_{0,4}\colon H^*(Q_4^{11};\Z_4)\to H^{*+1}(Q_4^{11};\Z)$ denote the Bockstein associated to the short
exact sequence
\begin{equation}
	\shortexact{\Z}{\Z}{\Z_4} \,.
\end{equation}
Then $\beta_{0,4}(b\beta(b)^k) = x^{k+1}$ and $\beta_{0,4}(b\beta(b)^k\tilde t) = x^{k+1}t$.
\end{lem}
\begin{proof}
Let $\beta_4$ be the Bockstein for the short exact sequence $0\to\Z_4\to\Z_{16}\to\Z_4\to 0$. Then
$\beta_{0,4}(\alpha)\bmod 4 = \beta_4(\alpha)$. \Cref{Q411_coh} shows us that the mod $4$ map $H^*(Q_4^{11};\Z)\to
H^*(Q_4^{11};\Z_4)$ is injective on $4$-torsion, so it suffices to understand $\beta_4$, then pull back to $\Z$
cohomology. Because $H^3(Q_4^{11};\Z) = 0$, $\beta_{0,4}(\tilde t) = \beta_{0,4}(\beta(b)) = 0$, and therefore
$\beta_4(\tilde t) = \beta_4(\beta(b)) = 0$, and we know $\beta_4(b) = \beta(b)$. Now, $\beta_4$ satisfies the
Leibniz rule~\cite[\S 3.E, $(*)$]{Hat02}
\begin{equation}
	\beta_4(\alpha_1\alpha_2) = \beta_4(\alpha_1)\alpha_2 + (-1)^{\abs{\alpha_1}}\alpha_1\beta(\alpha_2) \,,
\end{equation}
so we know the value of $\beta_4$, hence also of $\beta_{0,4}$, on products including $b\beta(b)^k$ and
$b\beta(b)^k\tilde t$.
\end{proof}
Now we directly compute the torsion pairing.
\begin{subequations}
\begin{itemize}
	\item For $\ang{x^3, x^3}$, Lemma \ref{4bock_lemma} tells us $\beta_{0,4}(b\beta(b)^2) = x^3$, and Proposition \ref{Q411_coh}
	tells us $x^3\bmod 4 = \beta(b)^3$, so
	\begin{equation}
		\ang{x^3, x^3} = \int_{Q_4^{11}} b\beta(b)^5 = 0 \,.
	\end{equation}
	\item For $\ang{x^3, x^2t}$, Lemma \ref{4bock_lemma} says $\beta_{0,4}(b\beta(b)\tilde t) = x^2t$, and we know
	$x^3\bmod 4 = \beta(b)^3$. Therefore
	\begin{equation}
		\ang{x^3, x^2t} = \int_{Q_4^{11}} b\beta(b)^4\tilde t = \tfrac 14 \,.
	\end{equation}
	\item Finally, for $\ang{x^2t, x^2t}$, we know $\beta_{0,4}(b\beta(b)\tilde t) = x^2t$ and Proposition \ref{Q411_coh}
	gives us $x^2t\bmod 4 = \beta(b)^2\tilde t$, so
	\begin{equation}
		\ang{x^2t, x^2t} = \int_{Q_4^{11}} b\beta(b)^3\tilde t^2 = 0 \,.
		\qedhere
	\end{equation}
\end{itemize}
\end{subequations}
\end{proof}
For this torsion pairing, the possible quadratic refinements' Arf invariants are of the form $k/4$ with $k\in\Z$. Let us see how this comes about. Let us denote a general element of $\mathbb{Z}_4 \times \mathbb{Z}_4$ by a pair of integers modulo 4, $(n_1,n_2)$. As explained in \cite{Hsieh:2020jpj}, a quadratic refinement of the bilinear pairing $(A,B)$ is defined by the property that
\begin{equation} \mathcal{Q}(A+B)-\mathcal{Q}(A)-\mathcal{Q}(B)+\mathcal{Q}(0)= \ang{A,B} \,. \label{ee0}\end{equation}
We are interested in quadratic refinements with $\mathcal{Q}(0)=0$, so we drop the last term on the left hand side. Evaluating \eq{ee0} for $A=B=(1,0)$ or $(0,1)$, we get
\begin{equation} \mathcal{Q}(2A)=2\mathcal{Q}(A),\quad \mathcal{Q}(3A)=3\mathcal{Q}(A) \,. \label{ee1}\end{equation}
We will denote $\mathcal{Q}((0,1))=k_1/4$ and $\mathcal{Q}((1,0))=k_2/4$ for some integers $k_1,k_2$. Finally, we also have
\begin{equation} \mathcal{Q}((n_1,n_2))= k_1n_1+k_2n_2+\tfrac14 n_1n_2 \,. \label{ee2}\end{equation}
This is the most general possible quadratic refinement. The Arf invariants can now be computed explicitly, and they turn out to be multiples of $\tfrac14$ in all cases.

\subsection{Anomalies for the rest of the classes}
Below the dashed line in \eq{results}, the duality bundle only involves the action of reflections and half-turns;
the corresponding relevant representations are real, and the bundles are generated by an inclusion $D_4\rightarrow
\GL(2, \Z)$ (or on the double cover, $D_8\to\GL^+(2, \Z)$). Furthermore, for each class, there are two different
variants, which we denote in \eq{results} as tilded and untilded, and which differ in the way the inclusion takes
place; the reflections in $D_8$ can be embedded as reflections along the sides or along the diagonal in $\Spin
\text{-} {D_{16}}\subset \Spin\text{-}\GL^+(2, \Z)$, as illustrated in Figure \ref{square_in_octagon}. At the level of
the associated principal $\GL(2, \Z)$ bundle, these correspond to a $\Z_2$ reflection subgroup acting on the fermions as multiplication by either of the  two matrices
 \begin{equation} \mathbf{R}_+=
 	\begin{pmatrix}
		1 & \phantom{-}0\\
		0 & -1
	\end{pmatrix}
 \quad\text{or}\quad \mathbf{R}_\times=
 	\begin{pmatrix} 0 & 1\\1 & 0\end{pmatrix}
.\label{re5r}\end{equation}
Since both matrices are equivalent by a change of basis, the fermion anomalies will not be able to distinguish
tilded and untilded versions of the theory; they will all have the same anomaly. The same applies to the self-dual
field, which couples via the determinant representation which sends both matrices in \eq{re5r} to a sign. As a
result, we only need to evaluate the anomaly theory for untilded classes, lowering our workload from eight to four
classes.

One can also see this directly using bordism. Let $\GL^+(2, \R)$ denote the Pin\textsuperscript{$+$} cover of
$\GL(2, \R)$; then one can define Spin-$\GL^+(2, \R)$ structures analogously to Spin-$\GL^+(2, \Z)$ structures.
Using the inclusion $j\colon \GL^+(2, \Z)\to\GL^+(2, \R)$, a Spin-$\GL^+(2, \Z)$ structure on a manifold $M$
induces a Spin-$\GL^+(2, \R)$ structure on $M$.
\begin{lem}
Let $M$ be a closed Spin-$D_8$ manifold, and let $M$ and $\widetilde M$ denote the two Spin-$\GL^+(2,
\Z)$-structures on $M$ given by the maps $i$ and $\widetilde\imath$. As Spin-$\GL^+(2, \R)$-manifolds, $M$ and
$\widetilde M$ are bordant.
\end{lem}
\begin{proof}
We will show that there is a homotopy $H$ between the two maps $j\circ i,j\circ\widetilde\imath\colon
D_8\to \GL^+(2, \R)$. Using $H$, define a Spin-$\GL^+(2, \R)$ structure on $[0, 1]\times M$ which
on $\{t\}\times M$ is induced by $H_t\colon D_8\to\GL^+(2, \R)$. This is a Spin-$\GL^+(2, \R)$ bordism from $M$ to
$\widetilde M$, as desired.

The image of $D_{16}$ in $\GL^+(2, \R)$ is contained in $\Pin^+(2)\cong\mathrm O(2)$; in the rest of this
paragraph, we regard $j\circ i$ and $j\circ\widetilde\imath$ as maps into $\mathrm O(2)$. Let $r$ be a generating
rotation of $D_8$ and $s$ be a generating reflection. Then, both $j\circ i$ and $j\circ\widetilde{\imath}$ send $r$
to $0$ and $s$ to a reflection through some line through the origin. These maps are homotopic because we can rotate
the line for $j\circ i$ to the line for $j\circ\widetilde\imath$.
\end{proof}
The partition functions of the anomaly theories for the fermion and the self-dual field only depend on the
underlying Spin-$\GL^+(2, \R)$-structure, so are equal on $M$ and on $\widetilde M$.

To compute anomalies in the classes involving real projective spaces, we need the analogs of \eq{eff} and
\eq{gfor1} for these cases. Both examples that concern us, $\HP^2\times \RP^3$ and $\RP^{11}$, are Spin, and so the anomalies are the same as those of fermions coupled to an ordinary $\Z_2$ bundle.

 We will denote the $\eta$-invariant of a real Dirac operator coupled to the $\Z_2$ in either the trivial or sign representation as $\eta^{\text{D}}_\pm$. We have that $\eta^{\text{D}}_\pm(\RP^n)$ is given by $\tfrac12$ of the result \eq{eff} for $n=3,11$, since the $\mathbb{Z}_2$ action behaves as charge conjugation on the complex fermions. Following a similar logic as the one leading to \eq{gfor1},
\begin{equation}
\eta ^{\text{RS}}(\RP^n) = n \, \eta^{\text{D}}_\pm(\RP^n) - \eta^{\text{D}}_\mp(\RP^n) \,.
\label{ras2}
\end{equation}
We also need to write down the $\eta$-invariants that appear in \eq{ath} in terms of $\eta^{\text{D}}_{\pm}$. Both
matrices in \eq{re5r} can be diagonalized, resulting in exactly one $+1$ and one $-1$ eigenvalue. As a result, the $\eta$-invariants of the IIB supergravity fermions decompose as
\begin{equation}
\eta_q^{\text{D}}=\eta_+^{\text{D}}+\eta_-^{\text{D}} \,,\quad \eta_q^{\text{RS}}=\eta_+^{\text{RS}}+\eta_-^{\text{RS}} \,.
\label{e344}
\end{equation}
Finally, we will also need the $\eta$-invariant for the self-dual field coupled to the sign representation. The explicit expression is given by \eq{effsd}.

We are now in a position to evaluate the anomalies on the classes in \eq{results} involving real projective spaces. One can check explicitly that
\begin{equation}
\eta_q^{\text{D}}(\RP^n) = \eta_q^{\text{RS}}(\RP^n) = \eta_-^{\text{Sig}}(\RP^n)=0 \,.
\end{equation}
For the fermions, the contribution with $\eta_+$ cancels that of $\eta_-$ in \eq{e344}. For the self-dual field, the result vanishes identically. These results can be understood as a consequence of anomaly cancellation, together with the facts that the $\eta$-invariants are defined over the reals and that that these bordism classes are $\Z_2$. For concreteness, consider $\RP^{11}$. The 12-manifold is defined as
\begin{equation} Z_{12}=\frac{\RP^{11}\times S^1}{\Z_2},\label{boc}\end{equation}
where the $\Z_2$ acts by reflecting the $S^1$ coordinate. This has boundary given by two copies of $\RP^{11}$, with opposite induced orientations.
Now consider the anomaly theory $8\mathcal{A}$, given by eight times the actual anomaly theory in \eq{ath}. The
partition function of $8\mathcal A$ is a bordism invariant, but it only depends on $\eta$-invariants. We can
use the APS index theorems to evaluate these explicitly as the equivariant indices of the corresponding operators
in $Z_{12}$; the zero modes of an elliptic operator in $Z_{12}$ correspond to $\Z_2$-invariant zero modes in
$\RP^{11}\times S^1$. There are no fermion zero modes in $\RP^{11}\times S^1$, and so the corresponding $\eta$'s
must vanish. For the self-dual field, there is a single $\Z_2$-odd zero mode (given by the polyform obtained as the
sum of the harmonic forms in degrees 1 and 10), but we need to subtract the contribution from the zero modes of the
boundary \cite{Hsieh:2020jpj}, of which there is just one. As a result, the anomaly for the self-dual field vanishes as well.

Having dealt with the $\eta$-invariants, we must now deal with the Arf invariant in \eq{ath}. This is simpler because
we must use cohomology in the local coefficient system $L$ that we defined in Section \ref{ssec:IIBan}, and
\begin{equation}
H^6(\RP^{11}; L)= H^2(\RP^{3}; L)=0 \,,
\end{equation}
so the Arf invariant must vanish. The anomaly in all these cases therefore vanishes.

 We can evaluate $\mathcal{A}$ in the classes $X_{10}\times S^1$ and $X_{10}\times \widetilde{S}^1$ directly by using \eq{jejeje}. Since both manifolds are products, \eq{jejeje} relates the anomaly to the Dirac, Rarita-Schwinger, and self-dual indices\footnote{In the case of the self-dual operator, we would apply \eq{jejeje} to the Dirac operator for the bispinor field, as in \cite{Alvarez-Gaume:1984zlq,Hsieh:2020jpj}.} of $X_{10}$, respectively. But this being a 10-manifold, the index densities all vanish. As a result,
\begin{equation}\mathcal{A}=0\quad \text{for}\quad X_{10}\times S^1,\quad X_{10}\times\widetilde{S^1}.\end{equation}

We are left with the classes involving $X_{11}$.
\begin{lem}
\label{eta_X11_vanish}
The $\eta$-invariants of the Dirac, Rarita-Schwinger, and signature operators vanish on $X_{11}$.
\end{lem}
\begin{proof}
We can realize $X_{11}$ as the boundary of a 12-manifold
\begin{equation}Y_{12}\equiv \frac{S^6\times S^6}{\Z_2\times\Z_2},\end{equation}
where, following the same conventions as in Appendix \ref{app:A}, we realize $S^6\times S^6$ as the locus
\begin{equation} \vert \vec{x}\vert=\vert\vec{y}\vert=1,\end{equation}
in an ambient $\mathbb{R}^{14}$ parametrized as $(\vec{x},\vec{y})$ where each of the vectors has seven components. The two
$\Z_2$ factors act by flipping $\vec{x}$ and $y_1$, and by flipping all the coordinates in $\vec{y}$ except for $y_7$,
respectively. The corresponding action has fixed points at $y_7=\pm1$; excising these produces a manifold with
boundary homeomorphic to two copies of $X_{11}$ with opposite orientation. We can then run the same argument as
around \eq{boc}; the Dirac and Rarita-Schwinger fields have no zero modes on $S^6$.  Taking into account the $\Z_2$
action, the self-dual field has one zero mode given by the sum of the harmonic forms in degree 1 and 12; but one
has to subtract the boundary contribution, which is again a single class at top degree, so that there are no zero modes either.
\end{proof}
As a result,
\begin{equation}\mathcal{A}(X_{11})= \text{Arf}(X_{11}),\end{equation}
depends only on the Arf invariant. A priori this is an element of $\mathbb{Z}_8$, but since $2[X_{11}] = 0$ in
$\Omega_{11}^{\Spin\text{-}\GL^+(2, \Z)}$, we know that this Arf invariant is either $0$ or $\tfrac12$ mod $1$.

As a consistency check, we compute the torsion subgroup of $H^6(X_{11}; L)$, where $L$ is the local coefficient
system from Section \ref{ssec:IIBan}. The torsion subgroup is isomorphic to $\Z_2\oplus\Z_2$; the possible Arf invariants
of a quadratic refinement of a bilinear pairing on this Abelian group are $\{0, 1/4, 1/2, 3/4\}\bmod 1$, which is
consistent with what we learned from the bordism consideration above.

In~\eqref{Z2_act_X11}, we explicitly identified $\pi_1(X_{11})\cong\{1, \alpha, \beta, \alpha\beta\}$. For both
Spin-$\GL^+(2, \Z)$ structures we care about, which we denoted $X_{11}$ and $\widetilde X_{11}$, $\alpha$ maps to a
reflection in $\GL(2, \Z)$ and $\beta$ maps to a rotation, so $L$ is the local system in which $\alpha$ acts by
$-1$ and $\beta$ acts by $1$.
\begin{prop}
\label{twisted_calc}
$H^6(X_{11}; L)\cong \Z^{\oplus 2} \oplus \Z_2^{\oplus 2}$ \,.
\end{prop}
\begin{proof}[Proof]
For brevity, let $V := \pi_1(X_{11})$. The universal cover of
$X_{11}$ is the principal $V$-bundle $S^6\times S^5\to X_{11}$; this induces a $V$-action on $S^6\times S^5$, hence
also on the singular chain complex $C_*(S^6\times S^5)$. Cohomology with local coefficients is defined to be the
cohomology of the cochain complex of maps to $L$:
\begin{equation}
	H^k(X_{11}; L) := H^k(\Hom_{\Z[V]}(C_*(S^6\times S^5), L)) \,.
\end{equation}
Abstract nonsense implies we can replace $C_*(S^6\times S^5)$ with any quasi-isomorphic chain complex of
$\Z[V]$-modules, in particular the complex of CW chains for any CW structure on $S^6\times S^5$ for which the
$V$-action is cellular. Consider the sequence of embeddings $S^{n-1}\hookrightarrow S^n$ as the equator and
inductively define a CW structure $\Pi_n$ on $S^n$ by beginning with $\Pi_{n-1}$ on the equatorial $S^{n-1}$, then
attaching the northern and southern hemispheres as two $n$-cells with orientations inherited from the standard
orientation on $S^n$. If we define the equator in $S^5$ to be the intersection of $S^5$ with $\{y_1 = 0\}$, then
the $V$-action on $S^6\times S^5$ is cellular for the product CW structure $\Pi_6\times\Pi_5$: the image of a cell
under any $x\in V$ is a union of cells.

Let $p_i,q_i$ denote the two $i$-cells in $\Pi_6$ ($0\le i\le 6$) and $r_j,s_j$ denote the two $j$-cells in $\Pi_5$
($0\le j\le 5$). Then
\begin{equation}
	\partial p_i = \partial q_i = \begin{cases}
		0, &i\text{ odd}\\
		p_{i-1} + q_{i-1}, &i\text{ even,}\\
	\end{cases}
\end{equation}
and $\partial r_j$ and $\partial s_j$ are analogous. To determine the boundary map in $\Pi_6\times\Pi_5$, use the
Leibniz formula $\partial(C\times D) = \partial(C)\times D + (-1)^{\deg(C)} C\times\partial(D)$. We also need to
know how $V$ permutes the cells; $\alpha$ exchanges $p_i\leftrightarrow q_i$ for all $i$, exchanges
$r_5\leftrightarrow s_5$, and fixes the remaining $r_j$ and $s_j$. $\beta$ fixes $p_i$ and $q_i$ for all $i$ and
exchanges $r_j\leftrightarrow s_j$ for all $j$. With all this data, we know the CW chains as a complex of
$\Z[V]$-modules.

If $e\in\Hom_{\Z[V]}(C_k(S^6\times S^5), L)$, $e$ is uniquely specified by its values on cells of the form
$p_i\times r_j$, where $i+j = k$, $0\le i\le 6$, and $0\le j\le 5$, because every $V$-orbit of the set of $k$-cells
contains exactly one cell of this form. Moreover, the $V$-action on the set of cells is free, so $e(p_i\times r_j)$
can be chosen arbitrarily. Let $e^{i,j}$ be the unique element of $\Hom_{\Z[V]}(C_k(S^6\times S^5), L)$
with value $1$ on $p_i\times r_j$, and which vanishes on all $p_{i'}\times r_{j'}$ for $(i',j')\ne (i, j)$. The
coboundary map $\delta$ can then be calculated as usual, and we find
\begin{equation}
\begin{alignedat}{2}
	\delta(e^{2, 3}) &= 2e^{2,4} &\qquad\qquad \delta(e^{3, 3}) &= -2 e^{3,4}\\
	\delta(e^{4, 1}) &= 2e^{4,2} & \delta(e^{5, 1}) &= -2e^{5,2},
\end{alignedat}
\end{equation}
and on the remaining $e^{i,j}$, $\delta(e^{i,j}) = 0$. Therefore $H^6(X_{11}; L) = \Z\cdot e^{1,5} \oplus \Z\cdot
e^{6, 0} \oplus \Z_2\cdot e^{2, 4} \oplus\Z_2\cdot e^{4, 2}$.
\end{proof}
\begin{cor}
\label{twisted_implies}
Let $\mathcal Q\colon H^6(X_{11}; L)\to \R/\Z$ be a quadratic refinement of the torsion pairing. Then the Arf
invariant of $\mathcal Q$ is either $0$, $1/4$, $1/2$, or $3/4$.
\end{cor}
The proof amounts to a calculation of the possible Arf invariants of a bilinear form on a two-dimensional
$\Z_2$-vector space.


\bibliographystyle{JHEP-modified}
\bibliography{DualityAnomaly}

\end{document}